\documentclass[a4paper,11pt]{article} 

\usepackage[left=2.5cm,right=2.5cm,top=2cm,bottom=2.5cm]{geometry}

\usepackage[utf8]{inputenc}
\usepackage[english]{babel}

\usepackage{bbm} 
\usepackage{amsmath}
\usepackage{amsfonts}
\usepackage{amssymb}
\usepackage{amsthm}
\usepackage{dsfont}

\usepackage{multirow}
\usepackage{adjustbox}
\usepackage{booktabs}

\usepackage{caption}
\usepackage{subcaption}

\usepackage{rotating}

\usepackage{natbib}
\usepackage{appendix}

\usepackage{hyperref}

\usepackage{enumerate} 
\usepackage[bottom]{footmisc}
\usepackage{physics}
\usepackage{setspace}

\newcommand{\E}{\mathrm{E}}
\newcommand{\Var}{\mathrm{Var}}
\newcommand{\p}{\mathrm{P}} 

\newcommand{\Cov}{\mathrm{Cov}} 
\newcommand{\Cor}{\mathrm{Cor}}

\newcommand{\bin}{\mathrm{Bin}}
\newcommand{\expon}{\mathrm{Exp}}
\newcommand{\geom}{\mathrm{Geom}}
\newcommand{\bgeom}{\mathrm{BGeom}}
\newcommand{\poi}{\mathrm{Poi}}
\newcommand{\skellam}{\mathrm{Sk}}
\newcommand{\zipf}{\mathrm{Zipf}}

\newcommand{\bbn}{\mathbb{N}}
\newcommand{\bbr}{\mathbb{R}}
\newcommand{\bbz}{\mathbb{Z}}

\newcommand{\copyx}{x^\prime}
\newcommand{\copyy}{y^\prime}
\newcommand{\copyX}{X^\prime}
\newcommand{\copyY}{Y^\prime}

\newcommand{\copytwox}{x^{\prime\prime}}
\newcommand{\copytwoy}{y^{\prime\prime}}
\newcommand{\copytwoX}{X^{\prime\prime}}
\newcommand{\copytwoY}{Y^{\prime\prime}}

\newtheorem{theorem}{Theorem}[section]
\newtheorem{proposition}[theorem]{Proposition}
\newtheorem{corollary}[theorem]{Corollary}
\newtheorem{lemma}[theorem]{Lemma}
\theoremstyle{definition}
\newtheorem{definition}[theorem]{Definition}
\newtheorem{assumption}{Assumption}

\theoremstyle{remark}

\usepackage[colorinlistoftodos,textsize=tiny]{todonotes}
\newcommand{\Comments}{1}
\newcommand{\mynote}[2]{\ifnum\Comments=1\textcolor{#1}{#2}\fi}
\newcommand{\mytodo}[2]{\ifnum\Comments=1%
	\todo[linecolor=#1!80!black,backgroundcolor=#1,bordercolor=#1!80!black]{#2}\fi}

\ifnum\Comments=1               
\fi

\ifnum\Comments=1               
\fi

\ifnum\Comments=1               
\fi

\begin{document}
\title{Asymptotic Inference for Rank Correlations\thanks{We thank Tilmann Gneiting for many fruitful discussions and for providing very helpful feedback, which improved this paper. We are also grateful to Bernhard Klar, Johanna Nešlehová, Eva-Maria Walz and conference participants of the German Probability and Statistics Days 2025, Dresden, the German Statistical Week 2024, Regensburg, the BS-IMS World Congress Pre-Meeting for Young Researchers 2024, Essen, the 11th HKMetrics Workshop (2024), Heidelberg, the Macro and Financial Econometrics Workshop (2025), Heidelberg, and seminar participants at Goethe University Frankfurt and Heidelberg Institute for Theoretical Studies for helpful comments. Marc-Oliver Pohle is grateful for support by the Klaus Tschira Foundation, Germany.}}

\author{Marc-Oliver Pohle\thanks{Karlsruhe Institute of Technology, Institute of Statistics, Blücherstraße 17, 76185 Karlsruhe, Germany, e-mail: \href{mailto: pohle@kit.edu}{pohle@kit.edu} and Heidelberg Institute for Theoretical Studies} \and Jan-Lukas Wermuth\thanks{Goethe University Frankfurt, RuW Building, Theodor-W.-Adorno-Platz 4, 60323 Frankfurt, Germany, e-mail: \href{mailto: wermuth@econ.uni-frankfurt.de}{wermuth@econ.uni-frankfurt.de}} \and Christian H.\ Weiß\thanks{Helmut-Schmidt-University, Holstenhofweg 85, 22043 Hamburg, Germany, e-mail: \href{mailto: weissc@hsu-hh.de}{weissc@hsu-hh.de}. ORCID:  \href{https://orcid.org/0000-0001-8739-6631}{0000-0001-8739-6631}}}

\maketitle

\begin{abstract}

Kendall's tau and Spearman's rho are widely used tools for measuring dependence. Surprisingly, when it comes to asymptotic inference for these rank correlations, some fundamental results and methods have not yet been developed, in particular for discrete random variables and in the time series case, and concerning variance estimation in general. Consequently, asymptotic confidence intervals are not available. We provide a comprehensive treatment of asymptotic inference for classical rank correlations, including Kendall's tau, Spearman's rho, Goodman-Kruskal's gamma, Kendall's tau-b, and grade correlation. We derive asymptotic distributions for both iid and time series data, resorting to asymptotic results for U-statistics, and introduce consistent variance estimators. This enables the construction of confidence intervals and tests, generalizes classical results for continuous random variables and leads to corrected versions of widely used tests of independence. We analyze the finite-sample performance of our variance estimators, confidence intervals, and tests in simulations and illustrate their use in case studies.

\end{abstract}


\newline


\section{Introduction}
\label{Introduction}

Correlation coefficients are fundamental tools for data analysis, quantifying the mutual dependence between two random variables. 
The classical and most popular correlation coefficients are Pearson correlation \citep{pearson1896}, Kendall's $\tau$ \citep{kendall1938} and Spearman's $\rho$ \citep{spearman1904}, which indicate the direction of dependence by their sign, and the strength of dependence by their closeness to 1 in absolute value. Despite being the most widely used of the three coefficients, Pearson correlation, which we denote by $r$, has well-known shortcomings. 
It lacks fundamental theoretical properties such as invariance to monotonic transformations, making it dependent on data transformations, and attainability, in the sense that it can be far away from 1 in absolute value under very strong or perfect dependence, thereby compromising its interpretability \citep{embrechts2002, fissler2023}. Furthermore, it suffers from statistical problems in that its estimator is not robust to outliers and inefficient for heavy-tailed distributions. 
Finally, its applicability is limited: It is only defined for random variables on a metric scale, excluding ordinal variables, and only for variables with finite second moments, excluding heavy-tailed distributions. In the case of continuous random variables, the classical rank correlations Kendall's $\tau$ and Spearman's $\rho$ cure all of these shortcomings, but for discrete random variables, they also suffer from non-attainability \citep{genest2007, neslehova2007}. Popular generalizations of $\rho$ and $\tau$ mitigating this problem are Kendall's $\tau_b$ \citep{kendall1948} and grade correlation, which we denote by $\rho_b$. For $\tau$, there is even a simple and appealing generalization in Goodman-Kruskal's $\gamma$ \citep{goodmanKruskal1954}, which fully cures this problem \citep{pohle2025}. 

Given the widespread use of these rank correlations and their favorable theoretical properties, the availability of confidence intervals to quantify estimation uncertainty around point estimators, and tests, e.g., of the null hypotheses of uncorrelatedness or independence, ought to be taken for granted. Surprisingly, we have been unable to find a comprehensive treatment of asymptotic inference for rank correlations in the extant literature. 
In particular, often only the cases of continuous and independent and identically distributed (iid) observations are considered, excluding time series and discrete (and mixed) variables. 
Contrary to the case of Pearson correlation, the formulas for the asymptotic variances of rank correlations in the discrete case do not just carry over from the continuous case due to the role that ties play for rank correlations. 

This lack of key results is even more surprising given that \citet{hoeffding1948} in his seminal paper on U-statistics already in the iid case established asymptotic normality and derived the corresponding limiting variances of empirical Kendall's $\tau$, $\widehat{\tau}$ (here even including the discrete case), and Spearman's $\rho$, $\widehat{\rho}$ (limiting himself to the continuous case), making use of the facts that $\widehat{\tau}$ is a U-statistic and $\widehat{\rho}$ is asymptotically equivalent to a U-statistic. Also in the iid case, \cite{genest2013estimation} provide the limiting distribution of multivariate generalizations of $\widehat{\rho}$, paying particular attention to the discrete case. 

In the time series case, \citet{dehling2017} derive an invariance principle for U-statistics of order 2 and in particular $\widehat{\tau}$ under weak dependence. The limiting distributions for the other rank correlations mentioned above, in particular $\widehat{\rho}$, do not even seem to be available in the continuous case from the extant literature.

Moreover, estimators for asymptotic variances of rank correlations have been neglected in the literature with the consequence that classical tools for statistical inference -- asymptotic confidence intervals and tests -- are not available. The only exceptions are for Kendall's $\tau$, where \citet{noether1967} discusses an estimator for the asymptotic variance (see also \citealp{samara1988}) and again \citet{dehling2017}, who propose a variance estimator in the continuous time series case. Instead of consistent variance estimators, approximations to the asymptotic variances of $\widehat{\tau}$ and $\widehat{\rho}$ in the iid case seem to be in widespread use, see the proposal by \citet{fieller1957} being derived under bivariate normality, or modifications thereof \citep{bonett2000}. Already \citet{borkowf2002} demonstrated them to be inaccurate.

The only settings, where classical asymptotic inference for $\tau$ and $\rho$ seems to be in widespread use, is in testing independence in the iid case. There, the limiting distributions under independence assuming continuous variables are used, where the limiting variances of $\widehat{\tau}$ and $\widehat{\rho}$ are $4/9$ and $1$, respectively.\footnote{The basic \texttt{R} function \texttt{cor.test} from the \texttt{stats} package \citep{RCoreTeam} only provides a confidence interval for Pearson correlation. For $\tau$ and $\rho$, only those tests of independence are available.} The former result is due to \citet{kendall1938}, and according to \citet[p.\ 321]{hoeffding1948}, the latter is due to Student and was published by \citet{pearson1907}. Thus, here, the variances are constant and do not have to be estimated. In the discrete case, however, the limiting variances under independence have not been available from the literature. \citet{hoeffding1948} conjectured that they are different from their values in the continuous case and depend on single-tie and double-tie probabilities of the two variables of interest.



We provide a comprehensive treatment of asymptotic inference for classical rank correlations, namely for Kendall's $\tau$ and Spearman's $\rho$ as well as their generalizations to the discrete case, Goodman-Kruskal's $\gamma$, $\tau_b$ and $\rho_b$. First, we establish asymptotic normality of their empirical versions, $\widehat{\tau}$, $\widehat{\rho}$, $\widehat{\gamma}$, $\widehat{\tau}_b$ and $\widehat{\rho}_b$, respectively, and derive the corresponding asymptotic variances under classical time series assumptions
, which includes the important special case of iid processes. Here, we make use of asymptotic theory for U-statistics for weakly dependent stochastic processes (see \citet{dehling2006} for an overview). Secondly, we introduce consistent estimators for the asymptotic variances. 
Those two elements -- the asymptotic distribution and consistent variance estimators -- allow for the construction of confidence intervals and tests. Thirdly, we derive formulas for the asymptotic variances under independence between the two variables or processes of interest, which is important to execute widely-used tests of independence.
In the iid case, our results generalize the aforementioned classical results from the continuous case. The variance formulas indeed depend on the double-tie probabilities and are easily implementable, providing simple and valid tests of independence also in the discrete case. In time series settings, our results generalize the results of \citet{lun2023}, who establish the asymptotic distribution of $\widehat{\tau}$ and $\widehat{\rho}$ in the continuous case under independence between the two stochastic processes of interest. We also provide a multivariate central limit theorem for a vector of U-statistics, which enables the asymptotic analysis of $\widehat{\gamma}$, $\widehat{\tau}_b$ and $\widehat{\rho}_b$ as they are functions of several U-statistics. This result also demonstrates that vectors of empirical rank correlations are asymptotically multivariate normal and allows for the calculation of the corresponding covariance matrices. 

In an extensive simulation study, we analyze the finite-sample performance of our proposed confidence intervals and independence tests for a variety of continuous and discrete iid as well as time series data generating processes. 
We illustrate the use of our inferential procedures in two real-world data applications.

The remainder of the paper is structured as follows. Section~\ref{sec:rank_correlations} introduces fundamental concepts and the classical rank correlations $\rho$ and $\tau$ as well as their generalizations to the discrete case, while Section~\ref{Empirical Rank Correlations} considers their empirical counterparts. In Section~\ref{Asymptotic Theory: IID Data}, we provide the corresponding asymptotic distributions and propose consistent estimators for the asymptotic variances in the iid case. Furthermore, we discuss the limiting distributions and variance estimation under independence. Section~\ref{Asymptotic Theory: Time Series} presents generalizations of the results in Section~\ref{Asymptotic Theory: IID Data} to the time series case. Section~\ref{sec:simulations} summarizes simulation results and Section~\ref{Illustrative Data Applications} discusses empirical applications. Finally, we conclude in Section~\ref{Conclusions}, where we outline directions for future research. 

Appendix~\ref{Asymptotic Theory for U-Statistics under Weak Dependence} contains a discussion of asymptotic theory for U-statistics under weak dependence and, in particular, the above-mentioned multivariate central limit theorem. Appendix~\ref{Proofs and Additional Lemmas} presents all proofs, and Appendix~\ref{Cross-dependence of Bivariate Geometric Distibution} an illustrative example involving a bivariate geometric distribution. An extensive description of the simulated DGPs from Section~\ref{sec:simulations} and a detailed interpretation of the simulation results as well as the results itself in a tabulated form are shown in Appendix~\ref{app_sec:simulations}. 
We also provide code implementing our confidence intervals and tests in an \texttt{R} package \citep{RCoreTeam} at https://github.com/jan-lukas-wermuth/RCor and replication material at https://github.com/jan-lukas-wermuth/replication\_RCor.

\section{Rank Correlations} \label{sec:rank_correlations}

\subsection{Fundamental Concepts: Grades, Sign Differences, and Tie Pro\-babilities}
\label{Foundational Concepts}

Throughout the paper, let $(X,Y) \sim F_{X,Y}$ be a bivariate random variable with cumulative distribution function (CDF) $F_{X,Y}$. Similarly, $X\sim F_X$ denotes a univariate random variable with CDF $F_X$. We say that a random variable $X$ is continuous if its CDF $F_X$ is a continuous function. We use $(\copyX,\copyY)$ and $(\copytwoX,\copytwoY)$ to denote independent copies of $(X,Y)$. For a function $f$, we denote the left limit at $x$ as $f(x^{-}) = \lim_{a \rightarrow x, a < x} f(a)$. We write $\stackrel{p} {\rightarrow}$ for convergence in probability and $\stackrel{d} {\rightarrow}$ for convergence in distribution.

We define the mid-distribution function (MDF) of $X$ as $G_X(x):=(F_X(x) + F_X(x^{-}))/2$ for $x \in \mathbb{R}$ (and analogously for $Y$), see e.g.\ \citet{parzen97} for a discussion. Note that the MDF can be rewritten as
\begin{equation} \label{eq:grade_and_spi}
    G_X(x) = \E[H(x-X)], \quad \text{ where } \quad H(x) = \begin{cases}
		1, & \text{if } x>0,\\
		1/2, & \text{if } x=0,\\
		0, & \text{if } x<0
		\end{cases}
\end{equation}
denotes the Heaviside step function using the half-maximum convention. Essentially, the MDF counts the probability mass for values smaller than $x$ fully, and the probability mass for values equal to $x$ half. It is thus closely related to the concept of a midrank, see \eqref{eq:midrank_definition} in the next section. Furthermore, we define the bivariate MDF of $(X, Y)$ as \citep[see e.g.][]{hoeffding1948}
\begin{equation} \label{eq:bivariate_mid-distribution_function}
    G_{X,Y}(x,y) := \frac 1 4 \Big( F_{X,Y}(x,y) + F_{X,Y}(x^-,y) + F_{X,Y}(x,y^-) + F_{X,Y}(x^-,y^-) \Big).
\end{equation}

Plugging-in a univariate or bivariate random variable itself into the respective MDF, we obtain the univariate grade \citep{hoeffding1948} or mid-distribution transform \citep{parzen2004}, $G_X(X)$, and the bivariate grade $G_{X,Y}(X,Y)$, which are of fundamental importance for this paper. 
The univariate grade is a variant of the probability integral transform (PIT) $F_X(X)$ and better-suited for arbitrary real-valued random variables. It reduces to the PIT if $F_X$ is continuous. While the univariate grade does not follow a standard uniform distribution in the discrete case like the PIT (and itself) in the continuous case ($F_X(X) \sim U(0,1)$ if $X$ is continuous), contrary to the PIT its expectation is always $1/2$. Also, the variance is smaller than $1/12$ by a factor only depending on the double-tie probability of $X$, i.e.\ on $\p(X=\copyX=\copytwoX)$, also see the subsequent Lemma~\ref{lem:grade}. For the moments of the bivariate grade, there are no general results available, and its distribution depends on the distribution of $(X,Y)$. We often need its expectation in the paper, which we consequently need to estimate. 

A further crucial concept for this paper is the sign function
\begin{equation} \label{eq:sign_function}
\mathrm{sgn}(x) = \begin{cases}
		1, & \text{if } x>0,\\
		0, & \text{if } x=0,\\
		-1, & \text{if } x<0,
		\end{cases}
\end{equation}
more precisely the sign function of a difference. We collect important results about the grade, the sign difference $\mathrm{sgn}(X-X')$ and their relation, as often needed throughout the rest of the paper, in the following Lemma~\ref{lem:grade}. The first result just amounts to \eqref{eq:grade_and_spi} rewritten in terms of $\mathrm{sgn}(x-X)$, noting that $H(x) = 1/2 ( 1 + \mathrm{sgn}(x))$ holds.

\begin{lemma} \label{lem:grade}
	 It holds that 
  \begin{enumerate}[(i)]
      \item $\E\big[ \mathrm{sgn} (x-X) \big] = 2 G_X(x) - 1$,
        \item $\E\big[\mathrm{sgn}(x-X)\,\mathrm{sgn}(y-Y)\big] = 4G_{X,Y}(x,y)-2G_X(x)-2G_Y(y) + 1 $,
        \item $\E \big[ \mathrm{sgn}(X-X') \big] = 0$ \quad and \quad $\Var\big(\mathrm{sgn}(X-X')\big) = 1 - \p(X=\copyX)$,        \item $ \E\big[ G_X(X) \big] = 1/2$ \quad and \quad $\Var\big(G_X(X)\big) = 1/12 \big[ 1 - \p(X=\copyX=\copytwoX)\big]$.
  \end{enumerate}
\end{lemma}

The single-tie probability $\zeta(X):=\p(X=\copyX)$ and the double-tie probability $\zeta_2(X):=\p(X=\copyX=\copytwoX)$ showing up in the variance formulas are crucial for this paper as well. These probabilities equal zero if~$X$ is continuous and can be computed from the probability mass function (PMF) of~$X$ in the discrete (or discrete-continuous) case: 
\begin{equation} \label{eq:tie_probability_formula}
    \zeta(X) = \sum_{x: \p(X=x)>0} \p(X=x)^2,
    \qquad
    \zeta_2(X) = \sum_{x: \p(X=x)>0} \p(X=x)^3.
\end{equation}
We denote the probability of a tie in $X$ or $Y$ or both by
\begin{equation} \label{eq:nu}
    \nu(X,Y) := \p\big((X-\copyX)(Y-\copyY)=0\big) = \zeta(X) + \zeta(Y) - \p(X=\copyX, Y=\copyY).
\end{equation}
Throughout the rest of the paper, we assume non-degenerate random variables such that $\zeta_X$, $\zeta_2(X)$ and $\nu(X,Y)$ are strictly smaller than 1.

Let us conclude this section by referring to a running example in Appendix~\ref{Cross-dependence of Bivariate Geometric Distibution}, where we illustrate the terms, concepts, and results on the special case of a bivariate geometric distribution.

\subsection{\protect\boldmath Classical Rank Correlations: Kendall's $\tau$ and Spearman's $\rho$}
\label{Classical Rank Correlations}

The idea behind all classical correlation coefficients is to measure if two  random variables $X$ and $Y$ tend to move in the same or in opposite directions, and how strong this tendency is. Making this idea specific requires a reference point. In the case of covariance and Pearson correlation,
\begin{equation} \label{eq:Pearson_correlation}
    \Cov(X,Y) = \E\big[(X-\E[X])(Y-\E[Y])\big] \quad \text{ with } \quad r=\Cor(X,Y) = \frac{\Cov(X,Y)}{\sqrt{\Var(X) \Var(Y)}},
\end{equation}
the reference point is the expectation $(\E[X],\E[Y])$ of $(X,Y)$. 

Kendall's $\tau$ uses $(\copyX,\copyY)$, an independent copy of $(X,Y)$, as the reference point and calculates the probability that $X$ and $Y$ move in the same direction relative to $\copyX$ and $\copyY$ minus the probability that they move in the opposite direction, that is, it is defined as the probability of concordance minus the probability of discordance. It can be rewritten in terms of the sign function from \eqref{eq:sign_function}.

\begin{definition}[Kendall's $\tau$]
\label{def_tau}
	Kendall's $\tau$ is defined as $$\tau(X,Y) = \p\big((X-\copyX)(Y-\copyY) > 0 \big) - \p\big((X-\copyX)(Y-\copyY) < 0 \big) = \E\big[\mathrm{sgn}(X-\copyX)\, \mathrm{sgn}(Y-\copyY)\big].$$
\end{definition}

By part (iii) of Lemma \ref{lem:grade}, it holds that $\tau(X,Y) = \Cov\big(\mathrm{sgn}(X-\copyX),\mathrm{sgn}(Y-\copyY)\big)$, and in the continuous case, we even have $\tau(X,Y) = \Cor\big(\mathrm{sgn}(X-\copyX),\mathrm{sgn}(Y-\copyY)\big)$. Kendall's $\tau$ can also be compactly written in terms of the grade by using parts (ii) and (iv) of Lemma \ref{lem:grade}, namely $\tau(X,Y) = 4 \E[G_{X,Y} (X,Y)] - 1$, which generalizes its relation to the joint CDF and copula from the continuous case \citep{schweizer1981}. For the special case of the bivariate geometric distribution, see Appendix~\ref{Cross-dependence of Bivariate Geometric Distibution}.

Note that the definition of $\tau(X,Y)$ uses a difference between random variables for notational simplicity. However, we do not assume that~$X$ and~$Y$ are necessarily quantitative random variables, but we also allow for ordinal random variables with categories, say, $s_0<\cdots<s_d$. The above difference is valid in the ordinal case if we use the numerical coding $s_k := k$ for $k=0,\ldots,d$, which simply counts how many categories~$s_k$ is apart from the lowest category~$s_0$. Equivalently, one could avoid differences by using (expectations of) the indicator function instead, where $X-\copyX>0$ iff $\mathds{1}(X>\copyX)=1$ etc., also see \eqref{eq:nu} above on how to avoid differences. But to keep notations compact, we shall continue writing differences as in Definition~\ref{def_tau}. 


A suitable population version of Spearman's $\rho$ involves the covariance of the grades, $G_X(X)$ and $G_Y(Y)$, see e.g.\ \citet{hoeffding1948}. 

\begin{definition}[Spearman's $\rho$]
\label{def_rho}
	Spearman's $\rho$ is defined as 
 $$\rho(X,Y) = 12\, \Cov\big(G_X(X), G_Y(Y)\big) = 12\, \E\big[ (G_X(X) - 0.5 )(G_Y(Y) - 0.5)\big].$$
\end{definition}

Note that if $X$ and $Y$ are continuous, then $\Var(G_X(X))=\Var(G_Y(Y))= 1/12$ and $\rho(X,Y)=\Cor(G_X(X), G_Y(Y))$. Spearman's $\rho$ carries the idea behind covariance or Pearson correlation to the rank scale, leading to a measure with much more favourable properties. For the special case of the bivariate geometric distribution, see again Appendix~\ref{Cross-dependence of Bivariate Geometric Distibution}. 

Using the relation between the sign function of a difference and the grade from Lemma \ref{lem:grade}, $\rho$ can be rewritten as
\begin{equation} \label{eq:rho_sign_representation}
    \rho(X,Y) = 3\, \E\big[\mathrm{sgn}(X-\copyX)\, \mathrm{sgn}(Y-\copytwoY)\big].
\end{equation}
Thus, it can also be seen as a modification of $\tau$, which does not use an independent copy $(\copyX,\copyY)$ of the pair $(X,Y)$ as a reference point, but $(\copyX,\copytwoY)$, a pair of independent copies of $X$ and $Y$, respectively, which are also independent of each other. Rewriting \eqref{eq:rho_sign_representation} again such that it becomes symmetric with respect to $(X,Y)$, $(\copyX,\copyY)$ and $(\copytwoX,\copytwoY)$, see \eqref{eq:nonsymmetric_kernel} in the appendix and the discussion around it, we get
\begin{align} \label{eq:rho_symmetric_sign_representation} 
\rho(X,Y) = \E\Big[k^{(\rho)}\big( (X,Y),(\copyX, \copyY), (\copytwoX, \copytwoY) \big)\Big],
\end{align}
where
\begin{align} \label{eq:rho_kernel}
k^{(\rho)} \left( (x,y),(\copyx, \copyy), (\copytwox, \copytwoy) \right) =&\ 1/2 \bigl( \mathrm{sgn}(x-\copyx)\, \mathrm{sgn}(y-\copytwoy) + \mathrm{sgn}(x-\copytwox)\, \mathrm{sgn}(y-\copyy) \nonumber \\ &+ \mathrm{sgn}(\copyx-x)\, \mathrm{sgn}(\copyy-\copytwoy) + \mathrm{sgn}(\copyx-\copytwox)\, \mathrm{sgn}(\copyy-y) \\ &+ \mathrm{sgn}(\copytwox-x)\, \mathrm{sgn}(\copytwoy-\copyy) + \mathrm{sgn}(\copytwox-\copyx)\, \mathrm{sgn}(\copytwoy-y)  \bigr). \nonumber
\end{align}
The kernel $k^{(\rho)}$ can be written more compactly in the form of \eqref{eq:nonsymmetric_kernel} in the appendix, but we prefer not to use the more cumbersome notation from the appendix here, at the cost of writing out the sum in \eqref{eq:rho_kernel}. 
The empirical analogue of representation \eqref{eq:rho_symmetric_sign_representation} is a U-statistic, which is the starting point for our asymptotic analysis of $\widehat{\rho}$, see \eqref{eq:emp_grade_corr} below.

\subsection{\protect\boldmath  Generalized Rank Correlations for Discrete Random Variables: Goodman-Kruskal's $\gamma$, Kendall's $\tau_b$ and Grade Correlation}
\label{Generalized Rank Correlations for Non-Continuous Random Variables}


Kendall's $\tau$ and Spearman's $\rho$ possess all major desirable properties for dependence measures in the continuous case \citep{embrechts2002}. However, for discrete random variables, they lack the property of attainability \citep{neslehova2007, genest2007}, that is, they do not take the values $-1$ and $1$ under perfect negative and positive dependence (counter- and comonotonicity), respectively. In fact, they may be very far away from those values under perfect or strong dependence. This leads to problems with their interpretability as strong dependence should be indicated by values close to 1 in absolute value. \citet{pohle2025} analyze attainable generalizations of $\tau$ and $\rho$, which thus fulfill all desirable properties in general, and refer to them as ``proper''. For Kendall's $\tau$, there exists a particularly nice and simple proper counterpart, namely Goodman-Kruskal's $\gamma$. This measure just computes Kendall's $\tau$ conditional on the event that there are no ties in both $X$ and $Y$, recall \eqref{eq:nu}.

\begin{definition}[Goodman-Kruskal's $\gamma$] \label{def:gamma}
	Goodman-Kruskal's $\gamma$ is defined as $$\gamma(X,Y) = \frac{\tau(X,Y)}{\p\big((X-\copyX)(Y-\copyY) \neq 0 \big)} 
    = \frac{\tau(X,Y)}{1 - \nu(X,Y) }.$$
\end{definition}

A popular generalization of Kendall's $\tau$, which is often used to address its attainability problem, is Kendall's $\tau_b$, which uses the univariate tie probabilities \eqref{eq:tie_probability_formula} instead of $\nu$.

\begin{definition}[Kendall's $\tau_b$] \label{def:taub}
	Kendall's $\tau_b$ is defined as \begin{align*}\tau_b(X,Y)= \Cor\big(\mathrm{sgn}(X-\copyX),\mathrm{sgn}(Y-\copyY)\big)= \frac{\tau(X,Y)}{\sqrt{[ 1 - \zeta(X)] [ 1 - \zeta(Y)]}} = \frac{\tau(X,Y)}{\sqrt{\tau(X,X) \tau(Y,Y)}}.\end{align*}
\end{definition}


However, while $\tau_b$ mitigates the problem in that it is weakly larger in absolute value than $\tau$, \citet{pohle2025} show that it is still non-attainable and thus inferior to $\gamma$. Nevertheless, we consider $\tau_b$ due to its popularity and the widespread use of its empirical analogue in statistical software. 

Unfortunately, for Spearman's $\rho$, there is no simple proper version. We consequently recommend using the grade correlation coefficient, which constitutes, in a way, the analogue to $\tau_b$ as a generalized version of $\rho$ in the discrete case. However, it only mitigates the attainability problem (like $\tau_b$ does), but does not cure it.

\begin{definition}[Grade Correlation] \label{def:grade_correlation}
Grade correlation $\rho_b$ is defined as 
 $$\rho_b(X,Y) = \Cor\big(G_X(X), G_Y(Y)\big) =
 \frac{\rho(X,Y)}{\sqrt{[ 1 - \zeta_2(X)] [ 1 - \zeta_2(Y)]}}=\frac{\rho(X,Y)}{\sqrt{\rho(X,X)\rho(Y,Y)}}.$$
\end{definition}
Again, we refer to Appendix~\ref{Cross-dependence of Bivariate Geometric Distibution} for the special case of the bivariate geometric distribution.

\section{Empirical Rank Correlations}
\label{Empirical Rank Correlations}

\subsection{Classical Empirical Rank Correlations}
\label{Classical Empirical Rank Correlations}

We now define the empirical versions of Kendall's $\tau$ and Spearman's $\rho$ and discuss their (asymptotic in the case of the latter) representations as U-statistics, which are the basis for our approach to establishing asymptotic theory. From now on, let $\{X_i,Y_i\}_{i=1}^n$ be a bivariate time series of length $n$, that is a sample of size $n$ of a stochastic process $\{X_i,Y_i\}_{i\in \mathbb{Z}}$ with $\bbz = \{\ldots, -1,0,1, \ldots\}$. This includes the case of an iid sample of size $n$ as a special case. 
 As we will assume strict stationarity, the time index does not play a role when we deal with marginal distributions involving $X_i$ or $Y_i$, so we omit it in these cases, e.g.\ writing $F_X$ instead of $F_{X_i}$. We also omit the time index when we deal with joint distributions of $(X_i,Y_i)$ only at a single point in time $i$.

\begin{definition}[Empirical Kendall's $\tau$] \label{def:empirical_tau}
	 Empirical Kendall's $\tau$ is defined as
	$$\widehat{\tau} (X,Y)= \frac{2}{n(n-1)}  \sum_{1\leq i < j \leq n} \mathrm{sgn}(X_i-X_j)\, \mathrm{sgn}(Y_i-Y_j).$$
\end{definition}

We note that $\widehat{\tau}$ is a U-statistic with kernel 
\begin{equation} \label{eq:tau_kernel}
k^{(\tau)} \left( (x,y),(\copyx, \copyy) \right) = \mathrm{sgn}(x-\copyx)\, \mathrm{sgn}(y-\copyy).
\end{equation}
We present a short introduction to U-statistics in Appendix~\ref{Asymptotic Theory for U-Statistics under Weak Dependence}. 
In particular, see Definition \ref{def:u_statistic}. 

Denote the empirical distribution function computed from a univariate time series $\{X_i\}_{i=1}^n$ by $\widehat{F}_X(x)$, and the empirical MDF by $\widehat{G}_X(x):=\big(\widehat{F}_X(x) + \widehat{F}_X(x^{-})\big)/2$ (and similarly for $\{Y_i\}_{i=1}^n$). Furthermore, define the mid-rank of the observation $X_j$ in $\{X_i\}_{i=1}^n$ as 
\begin{equation} \label{eq:midrank_definition}
\widehat{R}_X(X_j) := \frac 1 2 + \sum_{i=1}^n H(X_j - X_i) = \frac{n+1}{2} + \frac 1 2 \sum_{i=1}^n \mathrm{sgn}(X_j - X_i).
\end{equation}
It holds that $\widehat{G}_X(X_j) = \frac 1 n \big( \widehat{R}_X(X_j) - 1/2\big)$.
\begin{definition}[Empirical Spearman's $\rho$] \label{def:empirical_rho}
	 Empirical Spearman's $\rho$ is defined as
	$$\widehat{\rho} (X,Y)= \frac{12}{n}  \sum_{i=1}^n \left( \widehat{G}_X(X_i) - \frac 1 2 \right) \left( \widehat{G}_Y(Y_i) - \frac 1 2 \right) = \frac{12}{n^3}  \sum_{i=1}^n \left(  \widehat{R}_X(X_i) - \frac{n+1}{2} \right) \left(  \widehat{R}_Y(Y_i) - \frac{n+1}{2} \right).$$
\end{definition}

Contrary to $\widehat{\tau}$, $\widehat{\rho}$ itself is not a U-statistic, but asymptotically equivalent to the empirical analogue of \eqref{eq:rho_symmetric_sign_representation}: 
\begin{align} \label{eq:emp_grade_corr}
\widetilde{\rho} (X,Y) =& \binom{n}{3}^{-1} \sum_{1 \leq i < j < o \leq n} k^{(\rho)}\big( (X_i,Y_i), (X_j,Y_j), (X_o,Y_o) \big),
\end{align}
where $k^{(\rho)}$ is defined in \eqref{eq:rho_kernel}.
More precisely, the relation between $\widehat{\rho}$ and $\widetilde{\rho}$ is the following, see \citet[eq.\ (9.21)]{hoeffding1948}, which arises essentially by the relation between midranks and the difference sign function from \eqref{eq:midrank_definition}:\footnote{The slightly different factors involving $n$ in \citet[eq.\ (9.21)]{hoeffding1948} arise because Hoeffding defines empirical Spearman's $\rho$ as $n^2/(n^2-1)\, \widehat{\rho}$.}
\begin{equation} \label{eq:asymptotic_equivalence_rhohat_rhotilde}
    \widehat{\rho} = \frac{(n-1)(n-2)}{n^2}\, \tilde{\rho} + \frac{3 (n-1)}{n^2}\, \widehat{\tau},
\end{equation}
where the second summand is $\mathcal{O}_{\p}(n^{-1})$ and thus plays no role asymptotically. So $\widetilde{\rho}$ and $\widehat{\rho}$ are asymptotically equivalent in probability, $\widehat{\rho}/\widetilde{\rho} \stackrel{p}{\to} 1$.

\subsection{Generalized Empirical Rank Correlations for Discrete Random Variables}
\label{Generalized Empirical Rank Correlations for Non-Continuous Random Variables}
We now consider the empirical counterparts of Goodman-Kruskal's $\gamma$, Kendall's $\tau_b$, and grade correlation from Section \ref{Generalized Rank Correlations for Non-Continuous Random Variables}.

\begin{definition}[Empirical Goodman-Kruskal's $\gamma$] \label{def:empirical_gamma}
	Empirical Goodman-Kruskal's $\gamma$ is defined as
	$$\widehat{\gamma} (X,Y)= \frac{\widehat{\tau} (X,Y)}{1 - \widehat{\nu}(X,Y)} \quad \text{ with } \quad  \widehat{\nu}(X,Y) := \frac{2}{n(n-1)} \sum_{1\leq i < j \leq n} \mathds{1}\big((X_i-X_j)(Y_i-Y_j)=0\big).$$
\end{definition}

We note that $\widehat{\nu}$ is a U-statistic with kernel 
$$k^{(\nu)} \left( (\copyx, \copyy),(x,y) \right) = \mathds{1}\big((x-\copyx)(y-\copyy)=0\big).$$

\begin{definition}[Empirical Kendall's $\tau_b$] \label{def:empirical_tau_b}
	Empirical Kendall's $\tau_b$ is defined as
	$$\widehat{\tau}_b (X,Y)= \frac{\widehat{\tau} (X,Y)}{\sqrt{\widehat{\tau} (X,X)\, \widehat{\tau} (Y,Y)}}.$$
\end{definition}
\begin{definition}[Empirical Grade Correlation] \label{def:empirical_grade_correlation}
	Empirical grade correlation is defined as
 	$$\widehat{\rho}_b (X,Y)= 
   \frac{\widehat{\rho}(X,Y)}{\sqrt{\widehat{\rho}(X,X)\, \widehat{\rho}(Y,Y)}}.
  $$
\end{definition}
Empirical grade correlation is just the empirical Pearson correlation of the empirical grades. 
$\widehat{\gamma}$, $\widehat{\tau}_b$, and $\widehat{\rho}_b$ are all differentiable functions of U-statistics (at least asymptotically in the case of the latter).

\section{Asymptotic Theory: IID Data}
\label{Asymptotic Theory: IID Data}

\subsection{Consistency}
\label{Assumptions and Consistency IID}

To show consistency and to derive the asymptotic distributions of the five empirical rank correlations from the previous section, we exploit that $\widehat{\tau}$ is a U-statistic, $\widehat{\rho}$ is asymptotically equivalent to the U-statistic $\widetilde{\rho}$, $\widehat{\gamma}$ and $\widehat{\tau}_b$ are functions of U-statistics, and $\widehat{\rho}_b$ is asymptotically equivalent to a function of U-statistics. Thus, we can employ limit theorems for U-statistics. Note that for notational convenience, we will usually suppress the arguments $X$ and $Y$ when writing down our dependence measures, e.g.\ $\tau:=\tau(X,Y)$.
In this section, we work under the iid assumption.  
\begin{assumption} \label{ass:iid}
    $\{X_i,Y_i\}_{i \in \mathbb{Z}}$ is a bivariate iid process.
\end{assumption}
Note that we do not need moment conditions as the considered rank-based measures do not operate with the observations themselves, but with their ranks, grades, or CDFs, respectively. All limit theorems from this section are special cases of the more general results in the next section, where we work under classical time series assumptions, so that the formulas for the asymptotic variances get more involved. We nevertheless dedicate a separate section to the results for iid data because they are of high relevance in their own right and new to a large part. Furthermore, it is beneficial to appreciate the less involved results in this section, in particular regarding variance estimation, before tackling the cumbersome time series results. 
In Appendix~\ref{Asymptotic Theory for U-Statistics under Weak Dependence}, we review a univariate central limit theorem (Proposition \ref{prop:CLT_ustats_univariate}) for U-statistics under the iid assumption \citep{hoeffding1948} and under weak dependence \citep{denker1983}. Furthermore, we present a multivariate central limit theorem for U-statistics (Proposition \ref{prop:clt_Ustats_multivariate})  under both assumptions, which in combination with the delta method \citep[see][]{beutner24} allows for the treatment of $\widehat{\gamma}$, $\widehat{\tau}_b$, and $\widehat{\rho}_b$.

The following consistency result follows from the classical law of large numbers for U-statistics \citep{hoeffding1948} and the continuous mapping theorem.

\begin{proposition}[Consistency] \label{prop:consistency_iid}
	Under Assumption \ref{ass:iid}, it holds that $\widehat{\tau} \stackrel{p} {\rightarrow} \tau$, $\widehat{\rho} \stackrel{p} {\rightarrow} \rho$, $\widehat{\gamma} \stackrel{p} {\rightarrow} \gamma$, $\widehat{\tau}_b \stackrel{p} {\rightarrow} \tau_b$, and $\widehat{\rho}_b \stackrel{p} {\rightarrow} \rho_b$ as $n \rightarrow \infty$.
\end{proposition}

\subsection{Asymptotic Distributions}


\label{Asymptotic Distributions IID}

The following proposition on the asymptotic distributions of $\widehat{\tau}$ and $\widehat{\rho}$ follows from a central limit theorem for U-statistics, namely Proposition \ref{prop:CLT_ustats_univariate} in the appendix. 

\begin{proposition}[Asymptotic Distribution of $\widehat{\tau}$ and $\widehat{\rho}$] \label{prop:asymptotic_distribution_tau_iid} 
	   Under Assumption \ref{ass:iid}, it holds that $$\sqrt{n} \left( \widehat{\tau} -\tau \right) \stackrel{d}{\rightarrow} \mathcal{N}(0, \sigma_{\tau,\textup{iid}}^2) \quad \text{ with } \quad \sigma_{\tau,\textup{iid}}^2 
 = 4\, \E\big[ k^{(\tau)}_1 (X,Y)^2 \big], \quad \text{ where }$$
 \begin{align} \label{eq:k1tau_formula}
		k^{(\tau)}_1 (x,y) 
		&= 4\, G_{X,Y} (x,y)  - 2 \big(G_X(x) + G_Y (y)\big) + 1 - \tau,
	\end{align} 
 and $$\sqrt{n} \left( \widehat{\rho} -\rho \right) \stackrel{d}{\rightarrow} \mathcal{N}(0, \sigma_{\rho,\textup{iid}}^2) \quad \text{ with } \quad \sigma_{\rho,\textup{iid}}^2 
 = 9\, \E\big[ k^{(\rho)}_1 (X,Y)^2 \big], \quad \text{ where }$$
	\begin{align} \label{eq:k1rho_formula}
		k^{(\rho)}_1 (x,y) 
		&= 4 \big( \E[G_{X,Y} (x,Y) ] + \E[G_{X,Y} (X,y) ]  + G_X(x)G_Y (y) - G_X(x) - G_Y(y)\big) + 1 - \rho.
	\end{align} 
\end{proposition}


The variance formulas arise from the general formula for the asymptotic variance of a U-statistic (see again Proposition \ref{prop:CLT_ustats_univariate} in the appendix), where $k_1$ is the leading term of the Hoeffding decomposition of U-statistics (see \eqref{eq:Hoeffding_decomposition} and \eqref{eq_k1} in the appendix), and where the factors~4 and~9, respectively, arise from the squared orders of the U-statistics, here $2^2$ and $3^2$. For example, this leading term is computed for $\tau$ as
$$
k_1^{(\tau)}(x,y)
\ =\ \E\big[k^{(\tau)}\big((x,y), (X_i,Y_i)\big)\big]-\tau
\ =\ \E\big[\mathrm{sgn}(x-X_i)\, \mathrm{sgn}(y-Y_i)\big]-\tau,
$$
and the latter is computed by part~(ii) of Lemma \ref{lem:grade}. The full derivations of the other leading terms appearing in Proposition \ref{prop:asymptotic_distribution_tau_iid} and \ref{prop:asymptotic_distribution_gamma_iid} are provided by Lemma \ref{lem:variance} in the appendix. The asymptotic distribution of $\widehat{\tau}$ has already been established by \cite{hoeffding1948} and, in the continuous case, also the asymptotic distribution of $\widehat{\rho}$.

The asymptotic distributions of the generalized rank correlations in the iid case follow from our multivariate central limit theorem for U-statistics in Proposition \ref{prop:clt_Ustats_multivariate} in the appendix together with the delta method.

\begin{proposition}[Asymptotic Distribution of $\widehat{\gamma}$, $\widehat{\tau_b}$ and $\widehat{\rho}_b$] \label{prop:asymptotic_distribution_gamma_iid}
Let us use the short-hand notations $\tau:=\tau(X,Y)$, $\tau_X:=\tau(X,X)$, $\tau_Y:=\tau(Y,Y)$, $\rho:= \rho(X,Y)$, $\rho_{X} := \rho(X,X)$ and $\rho_{Y} := \rho(Y,Y)$.

	Under Assumption \ref{ass:iid}, it holds that $$\sqrt{n} \left( \widehat{\gamma} -\gamma \right) \stackrel{d}{\rightarrow} \mathcal{N}(0, \sigma_{\gamma,\textup{iid}}^2) \quad \text{ with } \quad  \sigma_{\gamma,\textup{iid}}^2 =  \frac{1}{(1-\nu)^2} \left( \sigma_{\tau,\textup{iid}}^2 + \gamma^2 \sigma_{\nu,\textup{iid}}^2 + 2 \gamma \sigma_{\tau \nu,\textup{iid}} \right), $$
   $$\sqrt{n} \left( \widehat{\tau}_b -\tau_b \right) \stackrel{d}{\rightarrow} \mathcal{N}(0, \sigma_{\tau_b, \textup{iid}}^2)\quad \text{with}$$
    $$\sigma_{\tau_b, \textup{iid}}^2 =     \frac{1}{\tau_{X}\tau_{Y}}\left(\sigma_{\tau,\textup{iid}}^2-\tau\left(\frac{\sigma_{\tau\tau_{X},\textup{iid}}}{\tau_{X}}-\frac{\sigma_{\tau\tau_{Y},\textup{iid}}}{\tau_{Y}}\right)+0.25\tau^2\left(\frac{\sigma_{\tau_{X},\textup{iid}}^2}{\tau_{X}^2}+\frac{\sigma_{\tau_{Y},\textup{iid}}^2}{\tau_{Y}^2}+2\frac{\sigma_{\tau_{X}\tau_{Y},\textup{iid}}}{\tau_{X}\tau_{Y}}\right)\right),$$
    and
    $$\sqrt{n} \left( \widehat{\rho}_b -\rho_b \right) \stackrel{d}{\rightarrow} \mathcal{N}(0, \sigma_{\rho_b, \textup{iid}}^2)\quad \text{with}$$
    $$\sigma_{\rho_b, \textup{iid}}^2 = 
    \frac{1}{\rho_{X}\rho_{Y}}\left(\sigma_{\rho,\textup{iid}}^2-\rho\left(\frac{\sigma_{\rho\rho_{X},\textup{iid}}}{\rho_{X}}-\frac{\sigma_{\rho\rho_{Y},\textup{iid}}}{\rho_{Y}}\right)+0.25\rho^2\left(\frac{\sigma_{\rho_{X},\textup{iid}}^2}{\rho_{X}^2}+\frac{\sigma_{\rho_{Y},\textup{iid}}^2}{\rho_{Y}^2}+2\frac{\sigma_{\rho_{X}\rho_{Y},\textup{iid}}}{\rho_{X}\rho_{Y}}\right)\right).
    $$
	Here,
	$$\sigma_{lm,\textup{iid}} =\  r^{(l)} r^{(m)}\,  \E \big[ k_{1}^{(l)} (X,Y)\, k_{1}^{(m)} (X,Y) \big]$$ 
    with $l,m \in \{\tau, \tau_X, \tau_Y, \rho, \rho_X, \rho_Y, \nu\}$ and $\sigma^2_i := \sigma_{ii}$ as well as $r^{(\tau)}=r^{(\tau_X)} = r^{(\tau_Y)} = r^{(\nu)} =2$, $\ r^{(\rho)}=r^{(\rho_X)} = r^{(\rho_Y)} = 3.$
	
    The kernels $k_1^{(\tau)}$ and $k_1^{(\rho)}$ are defined as in \eqref{eq:k1tau_formula} and \eqref{eq:k1rho_formula}, and it holds that
	\begin{align*}
		k_{1}^{(\nu)} (x,y) 
		= \p(X=x) + \p(Y=y) - \p(X=x,Y=y) - \nu,
	\end{align*}
  \begin{align*}
		k^{(\tau_X)}_1 (x,y):=k^{(\tau_X)}_1 (x)
		:= 1 - \tau_X-\p(X=x),\quad 
  k^{(\tau_Y)}_1 (x,y):=k^{(\tau_Y)}_1 (y)
		:= 1 - \tau_Y-\p(Y=y),
	\end{align*}
    \begin{align*}
		k_{1}^{(\rho_X)} (x,y) := k_{1}^{(\rho_X)} (x):=1 - \rho_X-\p(X=x)^2,\ k_{1}^{(\rho_Y)} (x,y):=k_{1}^{(\rho_Y)} (y) &:= 1 - \rho_Y-\p(Y=y)^2.
	\end{align*}
\end{proposition}

Proposition \ref{prop:clt_Ustats_multivariate} in the appendix is not only the key ingredient to obtain this proposition, but also shows that vectors of the empirical rank correlations considered here are asymptotically multivariate normal and enables the calculation of the corresponding variance matrix.

\subsection{Variance Estimation}
\label{subsec:Variance Estimation IID}


In the iid case, our proposed variance estimator requires, in a first step, the estimation of $k_1^{(\tau)}, k_1^{(\tau_X)}, k_1^{(\tau_Y)}, k_1^{(\rho)}, k_1^{(\rho_X)}, k_1^{(\rho_Y)}$ and $k_1^{(\nu)}$ (depending on which of those functionals are relevant for the chosen coefficient). In a second step, the expectation in the asymptotic variances is replaced by a sample mean. In the time series case, a third step comes in additionally, namely a heteroscedasticity and autocorrelation consistent (HAC) estimator being well known from variance estimation of the sample mean under temporal dependence.

To go into more detail, let us have a closer look at the asymptotic (co)variances in the Subsection~\ref{Asymptotic Distributions IID}, which are spelled out in general form in the formula for $\sigma_{lm,\textup{iid}}$ in Proposition \ref{prop:asymptotic_distribution_gamma_iid}. On the one hand, they contain the factor $r^{(l)}$, which is the (known) order of the $i$-th U-statistic. On the other hand, they contain the functionals $k_1^{(l)}$ depending on the joint CDF $F_{X,Y}$ of $(X,Y)$, and these need to be estimated.
To estimate them, we replace the theoretical CDFs $F_{X,Y}$ by empirical CDFs $\widehat{F}_{X,Y}$, and accordingly the theoretical univariate or multivariate MDFs or PMFs by the empirical ones. For example, for $\tau$, we replace the MDFs~$G_X$ and $G_Y$ in \eqref{eq:k1tau_formula} by their empirical counterparts $\widehat{G}_X$ and $\widehat{G}_Y$, which arise by replacing the theoretical CDFs $F_X$ and $F_Y$ in the definitions of the MDFs in and above $\eqref{eq:bivariate_mid-distribution_function}$ with empirical CDFs $\widehat{F}_X$ and $\widehat{F}_Y$:\begin{equation*}
    \widehat{k}^{(\tau)}_1 (x,y) = 4\, \widehat{G}_{X,Y} (x,y)  - 2 \big(\widehat{G}_X(x) + \widehat{G}_Y (y)\big) + 1 - \widehat{\tau}.
\end{equation*}

In the cases of $\widehat{\rho}$ and $\widehat{\rho}_b$, an additional difficulty arises due to functions of the form $g_X(x):=\E[G_{X,Y} (x,Y) ]$ arising in $k_1^{(\rho)}$. However, they can be estimated by $\widehat{g}_X (x) = \frac 1 n \sum_{i=1}^n \widehat{G}_{X,Y} (x,Y_i)$, and thus we estimate $k_1^{(\rho)}$ by
\begin{equation*}
    \widehat{k}^{(\rho)}_1 (x,y) 
		= 4 ( \widehat{g}_X(x) + \widehat{g}_Y(y)   + \widehat{G}_X(x) \widehat{G}_Y (y) - \widehat{G}_X(x) - \widehat{G}_Y(y) ) + 1 - \widehat{\rho}.
\end{equation*}

We finally estimate the expectation $\sigma_{lm,\textup{iid}}$ via
\begin{equation} \label{eq:variance_estimator_iid}
      \widehat{\sigma}_{lm,\textup{iid}} = r^{(l)} r^{(m)} \frac 1 n \sum_{i=1}^n \widehat{k}_{1}^{(l)} (X_i,Y_i) \widehat{k}_{1}^{(m)} (X_i,Y_i).
\end{equation}

We establish the consistency of this estimator. Because the estimators contain averages over random functions, e.g.\ $\widehat{k}^{(\tau)}_1 (x,y)$, and the asymptotic (co)variances featuring $\rho$ even contain averages of averages of random functions, the consistency is not trivial. Yet, it follows by standard arguments on random functions from empirical process theory using the Glivenko-Cantelli theorem. For details, see the proof in Appendix \ref{Proofs and Additional Lemmas} and \citet[Chapter 19.4]{vandervaart2000}.

\begin{proposition}[Consistency of Variance Estimator] \label{prop:consistency_variance_iid}
	Under Assumption \ref{ass:iid}, it holds for $l,m \in \{\tau, \tau_X, \tau_Y, \rho, \rho_X, \rho_Y, \nu\}$ that $\widehat{\sigma}_{lm,\textup{iid}} \stackrel{p} {\rightarrow} \sigma_{lm,\textup{iid}}$ as $n \rightarrow \infty.$
\end{proposition}


\subsection{Confidence Intervals, Tests, and Fisher Transformation}
\label{Fisher Transformation}


Building on the derived asymptotic distributions and proposed variance estimators, classical asymptotic confidence intervals and tests can be constructed. For any of our dependence measures, say $\delta$, the corresponding consistent estimator $\widehat{\delta}$ and variance estimator $\widehat{\sigma}^2_{\delta,\textup{iid}}$, the confidence interval at level $1-\alpha$ is
\begin{equation} \label{eq:confidence_interval}
    CI_{1-\alpha} (\delta) = \left[\widehat{\delta} - z_{1-{\alpha}/2} \cdot \frac{\widehat{\sigma}_{\delta,\textup{iid}}}{\sqrt{n}},\ \widehat{\delta} + z_{1-{\alpha}/2}  \cdot \frac{\widehat{\sigma}_{\delta,\textup{iid}}}{\sqrt{n}} \right],
\end{equation}
where $z_{1-{\alpha}/2}$ denotes the $1-{\alpha}/2$-quantile of the standard normal distribution.
The test statistic for the test 
for the null hypothesis $H_0: \delta=\delta_0$ is
\begin{equation} \label{eq:ttest}
T_{\delta_0} = \sqrt{n}\frac{ \widehat{\delta} - \delta_0}{\widehat{\sigma}_{\delta,\textup{iid}}},
\end{equation}
and the corresponding p-value against $H_1: \delta \neq \delta_0$ is $2(1-\Phi (|T_{\delta_0}|))$, where $\Phi$ denotes the CDF of the standard normal distribution.

The Fisher transformation introduced next leads to improved finite-sample performance of confidence intervals and tests if the true parameter is close to the bounds $\pm 1$.

It is well known that for the empirical Pearson correlation, the normal limiting distribution does not yield a good approximation if the true Pearson correlation is close to $-1$ or $1$, because then the true sampling distribution is heavily skewed. A solution to this problem is to transform the empirical correlation to an unbounded scale (the real line) by the Fisher transformation \citep{fisher1915}, and to use the limiting distribution of the Fisher transformation instead. This yields a much better approximation close to the boundaries and leaves the situation essentially unchanged away from the boundaries. This reasoning carries over to other correlation coefficients as well, see \citet[Appendix B]{pohle2024measuring} for a detailed discussion of the Fisher transformation for generic dependence measures. The Fisher transformation is defined as 
\begin{equation}
\label{eq:Fisher_transform}
z(\delta) := \mathrm{arctanh}(\delta) = \frac 1 2 \log \left( \frac{1+\delta}{1-\delta} \right)
\end{equation}
for a generic dependence measure $\delta$ (or an empirical dependence measure $\widehat{\delta}$). By the delta method, if the asymptotic distribution of $\widehat{\delta}$ is normal with variance $\sigma^2_{\delta}$, the asymptotic distribution of $z(\widehat{\delta})$ is normal as well, but with variance $\sigma^2_{\delta}/(1-\delta^2)^2 $, which can be estimated by plugging in estimators $\widehat{\delta}$ and $\widehat{\sigma}^2_{\delta}$. Tests can then be executed directly based on the limiting distribution of the Fisher transformation. For confidence intervals, the bounds of the confidence interval from this limiting distribution need to be transformed back to the original scale by applying the inverse Fisher transformation $z^{-1}(x) = \mathrm{tanh}(x) = (e^{2x}-1)/(e^{2x}+1)$.

\subsection{Asymptotic Distributions under Independence}
\label{Asymptotic Distributions under Independence IID}

Rank correlations equal 0 under independence of $X$ and $Y$. An important and classical application of rank correlations is testing the null hypothesis of independence via the implication of uncorrelatedness. We thus derive the formulas for the asymptotic variances under independence. As the basis for this, note that $k_1^{(\tau)}$ and $k_1^{(\rho)}$ from Proposition \ref{prop:asymptotic_distribution_tau_iid} simplify considerably and become equal to each other, $k_1^{(\tau)}=k_1^{(\rho)}$, under independence, see Lemma \ref{lem:kernels_under_independence} in the appendix.

\begin{corollary} \label{cor:iid_and_independent_processes_asymptotic_distribution_tau}
	Let $\{X_i\}_{i \in \mathbb{Z}}$ and $\{Y_i\}_{i \in \mathbb{Z}}$ be two independent iid processes. Then, it holds that 
    $$\sqrt{n}\,  \widehat{\tau}  \stackrel{d}{\rightarrow} \mathcal{N}(0, \sigma^2_{\tau,\textup{iid},\textup{ind}}) \quad \text{ with } \quad 
    \sigma^2_{\tau,\textup{iid},\textup{ind}} = \frac{4}{9} \big( 1 - \zeta_2(X) \big) \big(1 - \zeta_2(Y) \big),
    $$
    and
    $$\sqrt{n}\,  \widehat{\rho}  \stackrel{d}{\rightarrow} \mathcal{N}(0, \sigma^2_{\rho,\textup{iid},\textup{ind}}) \quad \text{ with } \quad 
    \sigma^2_{\rho,\textup{iid},\textup{ind}} =  \big( 1 - \zeta_2(X) \big) \big(1 - \zeta_2(Y) \big).
    $$ 
\end{corollary}
 
Corollary~\ref{cor:iid_and_independent_processes_asymptotic_distribution_tau} nicely extends the classical results for the continuous case, where the asymptotic variances are $4/9$ and $1$, respectively, to the general case, where the additional factors involving the double-tie probabilities of $X$ and $Y$ arise. Even though the asymptotic variances are quite fundamental and well known for the continuous case, see \citet{kendall1938} and \citet{hoeffding1948}, we are not aware of any reference for their neat formulas in the general case. 
Interestingly, already \citet[p.\ 317, 320]{hoeffding1948} conjectured that they depend on the tie probabilities $\p(X=\copyX)$ and $\p(X=\copyX=\copytwoX)$. 

Note that by the equality of $k_1^{(\tau)}$ and $k_1^{(\rho)}$, and our multivariate CLT for U-statistics (Proposition \ref{prop:clt_Ustats_multivariate} in the appendix), the asymptotic correlation between $\widehat{\tau}$ and $\widehat{\rho}$ under independence is 1. Considering the ratio of the standard deviations, this implies that they asymptotically have the deterministic relationship $\widehat{\rho} = 3/2\, \widehat{\tau}$, which was already proven by \citet{daniels1944}. Thus, the higher asymptotic variance of $\widehat{\rho}$ is not caused by it being less efficient, but by it taking higher (absolute) values under independence. 

As an immediate consequence from Proposition \ref{prop:asymptotic_distribution_gamma_iid}, using that $\tau=\rho=\gamma=0$ under independence of~$X$ and~$Y$,
we obtain the limiting distribution of the generalized rank correlations under independence. 
For $\widehat{\gamma}$, we also use that $1-\nu$ from \eqref{eq:nu} simplifies to
\begin{align}
1-\nu(X,Y)
 = \p(X\not=\copyX, Y\not=\copyY)
  =\p(X\not=\copyX)\, \p(Y\not=\copyY)
 = \left( 1 - \zeta(X) \right) \left(1 - \zeta(Y) \right)
 \label{eq:nu_under_independence}
\end{align}
 under independence.

\begin{corollary} \label{cor:iid_and_independent_processes_asymptotic_distribution_gamma}
	Let $\{X_i\}_{i \in \mathbb{Z}}$ and $\{Y_i\}_{i \in \mathbb{Z}}$ be two independent iid processes. Then it holds that 
    $$\sqrt{n}\,  \widehat{\gamma}  \stackrel{d}{\rightarrow} \mathcal{N}(0, \sigma^2_{\gamma,\textup{iid},\textup{ind}}) \quad \text{ with } \quad 
    \sigma^2_{\gamma,\textup{iid},\textup{ind}} = \frac{4}{9}\, \frac{1 - \zeta_2(X)}{\big(1 - \zeta(X)\big)^2}\, \frac{1 - \zeta_2(Y)}{\big(1 - \zeta(Y)\big)^2},
    $$ 
    $$\sqrt{n}\,  \widehat{\tau}_b  \stackrel{d}{\rightarrow} \mathcal{N}(0, \sigma^2_{\tau_b,\textup{iid},\textup{ind}}) \quad \text{ with } \quad 
    \sigma^2_{\tau_b,\textup{iid},\textup{ind}} = \frac{4}{9}\,  \frac{1 - \zeta_2(X)}{1 - \zeta(X)}\, \frac{1 - \zeta_2(Y)}{1 - \zeta(Y)},
    $$
    and
    $$\sqrt{n}\,  \widehat{\rho}_b  \stackrel{d}{\rightarrow} \mathcal{N}(0, \sigma^2_{\rho_b,\textup{iid},\textup{ind}}) \quad \text{ with } \quad \sigma^2_{\rho_b,\textup{iid},\textup{ind}} =  1.$$ 
\end{corollary}

The higher asymptotic variance of $\widehat{\tau}_b$ compared to $\widehat{\tau}$ (as well as that of $\widehat{\rho}_b$ compared to $\widehat{\rho}$) reflects the fact that it is weakly larger, mitigating the attainability problem. The even higher variance of $\widehat{\gamma}$ is again due to the fact that it takes weakly larger values than $\widehat{\tau}_b$ and solves the attainability problem. Nicely, the asymptotic variance of the grade correlation $\widehat{\rho}_b$ is 1, like the variance of $\widehat{\rho}$ itself in the continuous case. Thus, $\widehat{\rho}_b$ carries the nice property of an asymptotic variance not containing any unknown quantities that need to be estimated, which $\widehat{\rho}$ and $\widehat{\tau}$ possess in the continuous case, to the discrete world. Note that for the sake of testing independence, one could also define a variant of $\tau$ where the empirical counterpart has a constant asymptotic variance of $4/9$ under independence, namely by using the denominator of $\rho_b$: 
\begin{equation}
\label{eq:tau_b_mod}
\tau_{b,\textup{mod}}(X,Y) =  \frac{\tau(X,Y)}{\sqrt{[ 1 - \zeta_2(X)] [ 1 - \zeta_2(Y)]}}.
\end{equation}

To apply Corollaries~\ref{cor:iid_and_independent_processes_asymptotic_distribution_tau} or \ref{cor:iid_and_independent_processes_asymptotic_distribution_gamma} in practice, in order to test for cross-dependence, the probabilities $\zeta(X) = \p(X=\copyX)$ and $\zeta_2(X) = \p(X=\copyX=\copytwoX)$ (analogously for~$Y$) in the asymptotic variances need to be specified (except for the case of grade correlation $\rho_b$ as well as $\tau_{b,\textup{mod}}$). In the continuous case, $\zeta(X)=\zeta_2(X)=0$. In the discrete case, these probabilities have to be estimated from the given data. Recall that by \eqref{eq:tie_probability_formula}, it suffices to estimate the marginal distribution of $\{X_i\}$. In a non-parametric setup, this is done by computing relative frequencies instead of the probability masses, whereas simplifications are possible in parametric setups. For example, if~$X$ is geometrically distributed according to $\geom(\pi)$, then $\zeta(X)
 = \pi/(2-\pi)$ and $\zeta_2(X) = \pi^2/(3-(3-\pi)\pi)$, recall Appendix~\ref{Cross-dependence of Bivariate Geometric Distibution}.
Then, only the estimation of~$\pi$ is necessary.
A plot of the corresponding theoretical tie probabilities for $X\sim \geom\left(1/(1+\mu) \right)$ is provided by Figure~\ref{fig:BGeom_tieprob_asympVar} (left). The tie probabilities quickly converge towards 0 for increasing mean~$\mu$. Figure~\ref{fig:BGeom_tieprob_asympVar} (right) shows an analogous plot for the asymptotic variances from Corollaries~\ref{cor:iid_and_independent_processes_asymptotic_distribution_tau} or \ref{cor:iid_and_independent_processes_asymptotic_distribution_gamma}, where the variables $X,Y$ are assumed to be independent and identically distributed according to $\geom\left(1/(1+\mu) \right)$. 

\begin{figure}
    \centering
    \includegraphics[width=0.8\linewidth]{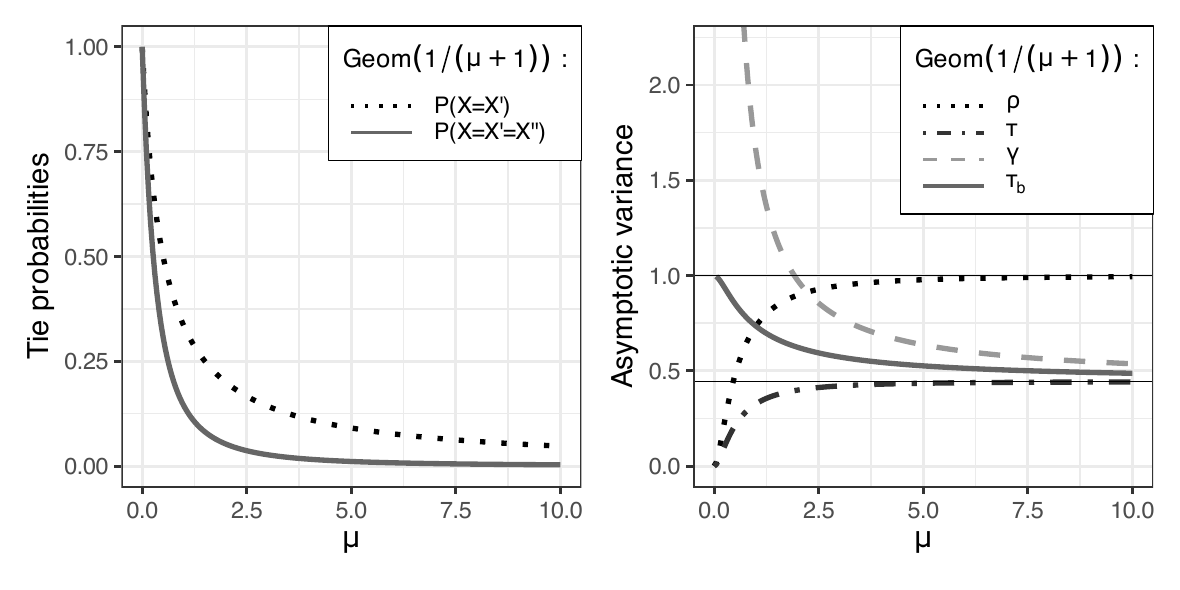}
    \caption{Left: Plot of tie probabilities against the mean~$\mu$ of a geometric distribution. Right: Plot of asymptotic variances from Corollaries~\ref{cor:iid_and_independent_processes_asymptotic_distribution_tau} and \ref{cor:iid_and_independent_processes_asymptotic_distribution_gamma} when $X$ and $Y$ are independent and identically geometrically distributed. The horizontal lines at 1 and 4/9 correspond to the variances of $\rho_b$ and $\tau_{b,\textup{mod}}$, respectively.}
    \label{fig:BGeom_tieprob_asympVar}
\end{figure}

Estimating the tie probabilities as described above and plugging those estimators into the variance formulas from the two preceding corollaries, we arrive at estimators for $\widehat{\sigma}^2_{\delta,\textup{iid},\textup{ind}}$. For testing the null hypothesis of independence via the implication of uncorrelatedness, we simply use the test statistic from \eqref{eq:ttest}, plugging in $\delta_0=0$ and $\widehat{\sigma}_{\delta,\textup{iid},\textup{ind}}$ as the estimator for the standard deviation.

\section{Asymptotic Theory: Time Series}
\label{Asymptotic Theory: Time Series}

\subsection{Consistency}
\label{Assumptions and Consistency TS}

In this section, we work under classical time series assumptions.

\begin{assumption} \label{ass:time_series}
    $\{X_i,Y_i\}_{i \in \mathbb{Z}}$ is a strictly stationary, ergodic and absolutely regular process with mixing coefficients $\{\beta_h\}_{h \in \mathbb{N}}$ satisfying $\sum_{h=1}^{\infty} \beta_h^{\delta/(2+\delta)} < \infty$ for some $\delta > 0$.
\end{assumption}

Absolute regularity or $\beta$-mixing is a form of asymptotic independence. It is stronger than strong mixing and weaker than uniform mixing, see \citet{bradley2005} for details.

Assumption \ref{ass:iid} is a special case of Assumption \ref{ass:time_series}, so we have to generalize the results from the previous section. The proofs for the limiting theorems are analogous, we just have to replace the law of large numbers and the central limit theorem under the iid assumption by more general limit theorems under Assumption \ref{ass:time_series}. 

Strong consistency of all five empirical rank correlations follows directly by an appropriate law of large numbers for U-statistics under weak dependence \citep[Theorem 1]{dehling2006}, together with the continuous mapping theorem in the cases of $\widehat{\gamma}$, $\widehat{\tau}_b$, and $\widehat{\rho}_b$.

\begin{proposition}[Consistency] \label{prop:consistency}
	Under Assumption \ref{ass:time_series}, it holds that $\widehat{\tau} \stackrel{p} {\rightarrow} \tau$, $\widehat{\rho} \stackrel{p} {\rightarrow} \rho$, $\widehat{\gamma} \stackrel{p} {\rightarrow} \gamma$, $\widehat{\tau}_b \stackrel{p} {\rightarrow} \tau_b$, and $\widehat{\rho}_b \stackrel{p} {\rightarrow} \rho_b$ as $n \rightarrow \infty$ .
\end{proposition}

\subsection{Asymptotic Distributions}
\label{Asymptotic Distributions TS}


To generalize the results from Propositions \ref{prop:asymptotic_distribution_tau_iid} and \ref{prop:asymptotic_distribution_gamma_iid}, we need a univariate central limit theorem under Assumption \ref{ass:time_series} (see \citet{denker1983} and Proposition \ref{prop:CLT_ustats_univariate} in the appendix), a multivariate version of it (Proposition \ref{prop:clt_Ustats_multivariate} in the appendix), and again the delta method. Furthermore, for the derivation of the asymptotic variances, we need Lemma \ref{lem:variance} in the appendix.

\begin{proposition}[Asymptotic Distribution, Time Series Case] \label{prop:asymptotic_distribution_tau} 
	Under Assumption \ref{ass:time_series} and replacing the quantities $\sigma_{lm,\textup{iid}}$ defined in Proposition \ref{prop:asymptotic_distribution_gamma_iid} by
 $$\sigma_{lm} =  r^{(l)}r^{(m)} \sum_{h=-\infty}^{\infty} \E \big[ k_{1}^{(l)} (X_i,Y_i)\, k_{1}^{(m)} (X_{i+h},Y_{i+h}) \big]$$ with $l,m \in \{\tau, \tau_X, \tau_Y, \rho, \rho_X, \rho_Y, \nu\}$, $\sigma^2_i := \sigma_{ii},$
 and using $k_1^{(l)}$ defined as in Propositions \ref{prop:asymptotic_distribution_tau_iid} and \ref{prop:asymptotic_distribution_gamma_iid} as well as $r^{(\tau)}=r^{(\tau_X)}=r^{(\tau_Y)}=r^{(\nu)} =2$, $r^{(\rho_X)} = r^{(\rho_Y)} = r^{(\rho)}=3,$
     the limiting distributions of Propositions \ref{prop:asymptotic_distribution_tau_iid} and \ref{prop:asymptotic_distribution_gamma_iid} for $\widehat{\tau}$, $\widehat{\rho}$, $\widehat{\gamma}$, $\widehat{\tau}_b$, and $\widehat{\rho}_b$ continue to apply. 
\end{proposition}


Thus, only the expressions for the asymptotic variances and covariances of the U-statistics occurring in the definitions of the rank correlations change compared to the iid case. In the continuous case, \citet{dehling2017} obtained the limiting distribution for Kendall's $\tau$ under similar assumptions.

We denote the cross- and autocovariances between two zero-mean processes $\{V_i\}_{i \in \mathbb{Z}}$ and $\{W_i\}_{i \in \mathbb{Z}}$ by $\alpha_{V,W} (h) := \Cov(V_i,W_{i+h})$.\footnote{Note that we denote the autocovariance function by $\alpha$ instead of using the classical notation $\gamma$ as we already use this letter for Goodman-Kruskal's $\gamma$.}
An infinite sum over the cross-autocovariances between the processes $\{ k^{(l)}_1(X_i,Y_i) \}_{i \in \mathbb{Z}}$ and $\{ k^{(m)}_1(X_i,Y_i) \}_{i \in \mathbb{Z}}$ emerges in $\sigma_{lm}$, that is, it can be rewritten as
\begin{equation}
\label{eq:covariances_ts_case}
\sigma_{lm} = r^{(l)}r^{(m)} \sum_{h=-\infty}^{\infty} \alpha_{k^{(l)}_1(X,Y),\, k^{(m)}_1(X,Y)}(h). 
\end{equation}

The following lemma shows that the asymptotic covariance $\sigma_{lm}$ is always finite under our assumptions. 
\begin{lemma}[Short Memory] \label{lem:short_memory}
Given Assumption \ref{ass:time_series} and the kernels $k^{(l)}_1$ and  $k^{(m)}_1$, $l,m \in \{\tau, \tau_X, \tau_Y, \rho, \rho_X, \rho_Y, \nu \}$, it holds that
    $$\sum_{h=-\infty}^{\infty} \left|\alpha_{k^{(l)}_1(X,Y),\, k^{(m)}_1(X,Y)}(h) \right| < \infty.$$
\end{lemma}

Such an absolute summability condition on the autocovariances is very common as an assumption in time series analysis, ruling out long memory (e.g.\ \citealp{hassler2018}). Here, it is already implied by Assumption \ref{ass:time_series} and the boundedness of our kernels.

We note that the structure of the asymptotic covariance $\sigma_{lm}$ is similar to the asymptotic variance of a sample mean $\overline{f}=\frac 1 n \sum_{i=1}^n f(X_i,Y_i)$, where $f$ is a known function and $\{X_i,Y_i\}_{i\in \mathbb{Z}}$ a bivariate stochastic process satisfying Assumption \ref{ass:time_series}. More precisely, it holds that $\sigma^2_{\mu_f} = n\, \Var( \overline{f}) = \sum_{h=-\infty}^{\infty} \alpha_{f(X,Y),f(X,Y)} (h)$  and similarly for the asymptotic covariance. 

\subsection{Variance Estimation}
\label{subsec:VarianceEstimationTS}
Compared to the estimator $\widehat{\sigma}_{lm,\textup{iid}}$ from Subsection~\ref{subsec:Variance Estimation IID}, the additional difficulty arises from the infinite sum in \eqref{eq:covariances_ts_case}. Besides estimating the functions $k_1^{(l)}$ and $k_1^{(m)}$ as laid out in Subsection~\ref{subsec:Variance Estimation IID}, we now employ classical techniques for long-run variance estimation, which are well known from the case of the sample mean \citep{newey1987}. That is, we estimate $\sigma_{lm}$ via

\begin{align} \label{eq:variance_estimation_time_series}
\widehat{\sigma}_{lm} =&\ r^{(l)} r^{(m)} \Biggl( \frac 1 n \sum_{i=1}^{n}  \widehat{k}_{1}^{(l)} (X_i,Y_i)\, \widehat{k}_{1}^{(m)} (X_i,Y_i)  \\
	 &  +  \sum_{h=1}^{n-1} w\biggl(\frac{h}{b_n+1}\biggr) \biggl( \frac 1 n \sum_{i=1}^{n-h}  \widehat{k}_{1}^{(l)} (X_i,Y_i)\, \widehat{k}_{1}^{(m)} (X_{i+h},Y_{i+h})  + \frac 1 n \sum_{i=1}^{n-h}  \widehat{k}_{1}^{(m)} (X_i,Y_i)\, \widehat{k}_{1}^{(l)} (X_{i+h},Y_{i+h}) \biggr) \Biggr) \nonumber \\
     =&\ r^{(l)} r^{(m)} \Biggl( \widehat{\alpha}_{\widehat{k}^{(l)}_1(X,Y),\, \widehat{k}^{(m)}_1(X,Y)}(0)\nonumber \\
     &+  \sum_{h=1}^{n-1} w\biggl(\frac{h}{b_n+1}\biggr) \biggl( \widehat{\alpha}_{\widehat{k}^{(l)}_1(X,Y),\, \widehat{k}^{(m)}_1(X,Y)}(h) + \widehat{\alpha}_{\widehat{k}^{(m)}_1(X,Y),\, \widehat{k}^{(l)}_1(X,Y)}(h)  \biggr) \Biggr),
     \nonumber
\end{align}
where the weighting function $w(h/(b_n+1))$ is also called a kernel or a window, $b_n$ is the bandwidth, and we define the empirical cross-autocovariance function between the zero-mean processes $\{V_i\}_{i \in \mathbb{Z}}$ and $\{W_i\}_{i \in \mathbb{Z}}$ by
$$
\widehat{\alpha}_{V,W} (h) := \frac 1 n \sum_{i=1}^{n-h}V_iW_{i+h}.
$$
We refer to  \citet{lazarus2018} for an overview of the vast literature on long-run variance estimation, or heteroscedastiticy and autocorrelation consistent (HAC)  estimation, in the case of sample means or regression coefficients. \citet{dehling2017} propose the use of a HAC estimator of this type for U-statistics of order 2 and especially for Kendall's $\tau$, but under a different set of assumptions. We establish consistency of our variance estimators under our assumptions on the time series and standard assumptions on the kernel and the bandwidth.

\begin{assumption}[Kernel and Bandwidth] \label{ass:kernel_and_bandwidth}
    
The kernel $w:\mathbb{R}\to\mathbb{R}$ is continuous, normalized: $w(0) = 1$, bounded: $|w(x)| \leq 1$, symmetric: $w(x) = w(-x)$ for all $x \in \mathbb{R}$, and has bounded support: $\operatorname{supp}(w)=[-1,1]$. The bandwidth $b_n$ depends on the sample size $n$ and satisfies: $b_n \to \infty, b_n/\sqrt{n} \to 0$ as $n \to \infty$.
\end{assumption}

\begin{proposition}[Consistency of Variance Estimator] \label{prop:consistency_variance_time_series}
	Under Assumptions \ref{ass:time_series} and \ref{ass:kernel_and_bandwidth}, it holds for $l,m \in \{\tau, \tau_X, \tau_Y, \rho, \rho_X, \rho_Y, \nu\}$ that $\widehat{\sigma}_{lm} \stackrel{p} {\rightarrow} \sigma_{lm}$ as $n \rightarrow \infty.$
\end{proposition}


We use the Bartlett kernel, or triangular kernel, which has been popularized by \cite{newey1987}:
\[
w_B\biggl(\frac{h}{b_n+1}\biggr)=\left\{
\begin{array}{cc}
	1-\frac{\left\vert h\right\vert }{b_n + 1} & ,\ \left\vert h\right\vert \leq b_n\\
	0 & ,\ \left\vert h\right\vert > b_n%
\end{array}%
\right. .
\]
As the bandwidth, we choose $b_n = \lfloor 2 n^{1/3} \rfloor$ as in \citet{dehling2017}. With these choices, Assumption \ref{ass:kernel_and_bandwidth} is satisfied. We conjecture that the bounded support assumption for the kernel can be replaced by an assumption on its rate of decay.

It is well known that variance estimation under temporal dependence is a notoriously difficult problem, usually plagued by oversized tests and confidence intervals with undercoverage, with the problem getting worse with the strength of temporal dependence and only disappearing for large sample sizes, see again \cite{lazarus2018} on this issue. Alternatives to the HAC estimator such as the moving block bootstrap also exhibit this behaviour \citep{fitzenberger1998}. As expected, these problems also arise in our simulations, see Section \ref{sec:simulations}. 

Except for replacing the variance estimators for the iid case with the ones proposed in this section, confidence intervals and tests are constructed as described around \eqref{eq:confidence_interval} and \eqref{eq:ttest}.

\subsection{Variance Formulas under Independence}
\label{Variance Formulas under Independence TS}


The case where $\{X_i\}_{i \in \mathbb{Z}}$ and $\{Y_i\}_{i \in \mathbb{Z}}$ are mutually independent, but themselves serially dependent processes, is of interest for testing the null hypothesis of independence between the two processes. We thus generalize Corollaries \ref{cor:iid_and_independent_processes_asymptotic_distribution_tau} and \ref{cor:iid_and_independent_processes_asymptotic_distribution_gamma} from the iid case to serial dependence. The resulting expressions are considerably simpler than in Proposition~\ref{prop:asymptotic_distribution_tau}, which is essentially due to the simpler form of $k_1^{(\tau)}$ and $k_1^{(\rho)}$ stated in Lemma \ref{lem:kernels_under_independence}. The latter leads to a neat simplification of the autocovariances in \eqref{eq:covariances_ts_case}, also see Lemma \ref{lem:autocovariances_under_independent_processes} in the appendix.

\begin{corollary} \label{cor:independent_processes_asymptotic_distribution_tau}
	Let $\{X_i\}_{i \in \mathbb{Z}}$ and $\{Y_i\}_{i \in \mathbb{Z}}$ be independent processes satisfying Assumption \ref{ass:time_series}. Then, the limiting distributions from Corollaries \ref{cor:iid_and_independent_processes_asymptotic_distribution_tau} and \ref{cor:iid_and_independent_processes_asymptotic_distribution_gamma} continue to apply, replacing the asymptotic variances $\sigma^2_{i,\textup{iid},\textup{ind}}$ by $\sigma^2_{i,\textup{ind}}$ for $i \in \{\tau, \rho,\gamma,\tau_b,\rho_b \}$:
	$$\sigma_{\tau,\textup{ind}}^2 = \frac{4}{9} \sum_{h=-\infty}^{\infty} \rho_X(h) \rho_Y(h), \quad \sigma_{\rho,\textup{ind}}^2 = \frac{9}{4}\sigma_{\tau,\textup{ind}}^2, \quad \sigma^2_{\gamma,\textup{ind}} = \frac{\sigma_{\tau, ind}^2}{\big(1 - \zeta(X)\big)^2\big(1 - \zeta(Y)\big)^2},$$
    $$\sigma^2_{\tau_b,\textup{ind}} = \frac{\sigma_{\tau, ind}^2}{\big(1 - \zeta(X)\big)\big(1 - \zeta(Y)\big)} \quad \text{and}\quad \sigma^2_{\rho_b,\textup{ind}} =  \sum_{h=-\infty}^{\infty} \rho_{b,X}(h) \rho_{b,Y}(h),$$ 
    where
    $\rho_X(h):= \rho(X_i,X_{i+h})$ denotes the Spearman autocorrelation function at lag~$h$, and $\rho_{b,X}(h):= \rho_b(X_i,X_{i+h}) = \rho_X(h)/\rho_X(0)$ the grade autocorrelation function.
    
\end{corollary}

Note that in order to achieve compact formulas, we have expressed the asymptotic variances for $\tau, \rho, \gamma$ and $\tau_b$ via the autocorrelation in terms of $\rho$, and the asymptotic variance for the grade correlation via the autocorrelation in terms of $\rho_b$. 
We could also express all results in terms of one or the other, see Definition \ref{def:grade_correlation} and Lemma \ref{lem:autocovariances_under_independent_processes} in the appendix for their relation. 
In the continuous case, this result for $\widehat{\tau}$ and $\widehat{\rho}$ has recently been derived by \citet[Corollary 2.1]{lun2023}.

In order to estimate the asymptotic variances from Corollary \ref{cor:independent_processes_asymptotic_distribution_tau}, we estimate the Spearman autocorrelations of $X$ on the sample $\{X_i,X_{i+h}\}_{i=1}^{n-h}$ via $\widehat{\rho}_X (h) := \widehat{\rho} (X_i,X_{i+h})$, see Definition \ref{def:empirical_rho} (and analogously for $Y$), plug them in for their theoretical counterparts and downweight the summands in the same way as for our variance estimator discussed in and below \eqref{eq:variance_estimation_time_series}. Consistency of this estimator follows by standard arguments for the consistency of a HAC estimator as employed in the proof of Proposition \ref{prop:consistency_variance_time_series}. For $\rho_b$, we use the values of the Spearman autocorrelations but estimate the denominator (i.e.\ the tie probabilities) only once on the full sample.

\section{Simulation Study} \label{sec:simulations}
In order to evaluate the finite-sample performance of the derived asymptotic distributions and proposed variance estimators for the different dependence measures, we report on a comprehensive simulation study. For many different types of data-generating processes (DGPs) -- both continuous and discrete, and both iid and serially dependent, we simulated $MC=1,000$ bivariate iid samples or time series, respectively, which are used for either testing for independence, or for computing a confidence interval for the respective dependence measure.  We summarize the most important results of the simulation study here. For details on the DGPs, tabulated simulation results, and a more in-depth analysis and interpretation, we refer to Appendix \ref{app_sec:simulations}.

We start with the independence tests for iid DGPs (see Corollaries \ref{cor:iid_and_independent_processes_asymptotic_distribution_tau} and \ref{cor:iid_and_independent_processes_asymptotic_distribution_gamma} and the corresponding variance estimators), where we also compare the size and power values of the tests based on rank correlations to tests based on Pearson correlation (see equations \ref{eq:Pearson_as_distribution} and \ref{eq:Pearson_as_distribution_TS} for its asymptotic distribution under independence in the iid and the time series case). Pearson correlation shows mild size distortions outside the bivariate normal case and strong size distortions for heavy-tailed DGPs. This pattern does not change between continuous and discrete DGPs. However, the size distortions further deteriorate with growing sample size when second moments are infinite and improve otherwise. The rank correlations on the other hand are robust to distributional characteristics such as heavy tails and hold the nominal size well across all DGPs. Only Kendall's $\tau$ shows mild oversizing in small samples ($n=50$). In terms of power, Pearson correlation again shows a worse performance than the rank correlations for the heavy-tailed DGPs. Additionally, asymmetries in the distributions seem to worsen the performance of Pearson correlation in terms of both size and power. Besides that, all coefficients show rather similar rejection rates with slight advantages for Spearman's $\rho$ in the continuous and for Goodman-Kruskal's $\gamma$ in the discrete case. Expectedly, all independence tests based on rank correlations exhibit increasing power with increasing sample size.

Turning to the coverage rates of the confidence intervals for iid DGPs (see Propositions \ref{prop:asymptotic_distribution_tau_iid}, \ref{prop:asymptotic_distribution_gamma_iid} and \ref{prop:consistency_variance_iid}), all coefficients show an empirical coverage close to the theoretical confidence level. Only for strong cross-dependencies, the former sometimes falls short of the latter, a misbehavior that diminishes when applying the Fisher transformation (see Subsection \ref{Fisher Transformation}).

For the time series DGPs, independence tests amount to tests of whether the two considered time series are independent of each other (and we thus employ Corollary \ref{cor:independent_processes_asymptotic_distribution_tau} and the corresponding variance estimators). Due to the rather strong serial dependence we use in our time series DGPs (our DGPs involve AR(1) processes with a coefficient value of 0.8) and the well-known difficulties with variance estimation under temporal dependence discussed in Subsection \ref{subsec:VarianceEstimationTS}, we observe strong oversizing for all DGPs and all coefficients even for $n=800$. Again, Pearson correlation performs worst for the heavy-tailed and asymmetric DGPs. For the independence tests based on the asymptotic theory in this paper, we again observe improving power with increasing length of the time series.

For the confidence intervals in the time series case, the empirical coverages, in turn, suffer from the biased variance estimates (see Propositions \ref{prop:asymptotic_distribution_tau} and \ref{prop:consistency_variance_time_series}) in the same way as the sizes of the independence tests. We observe strong undercoverage, especially for short time series, that improves with growing length of the time series. In addition to that, the strength of cross-dependence does also seem to play an important role for the empirical coverage, with some (continuous) DGPs showing decreasing coverage for increasing strength of dependence and some other DGPs exhibiting the opposite relationship.

\section{Illustrative Data Applications}
\label{Illustrative Data Applications}



To illustrate the practical application of the different dependence measures and the proposed confidence intervals and tests, we consider two data examples. The first one involves an ecological data set, where the data can be assumed to arise from an iid DGP, see Section~\ref{Botanical Data Example} for details. In Section~\ref{Accidents Data Example}, we consider a time series of daily numbers of road accidents. In both cases, the data are discrete counts such that our novel asymptotics are crucial for executing dependence tests and computing confidence intervals.

\subsection{Ecological Data}
\label{Botanical Data Example}
We consider the bivariate count data $(x_1,y_1),\ldots,(x_{100},y_{100})$ in Table~1 of \citet{holgate66}, which were further analyzed by \citet{kokonendji18}. The data arose in a study of a secondary rain forest in Trinidad, and they express the numbers of plants of the species
\textit{Lacistema aggregatum} ($x_i$) and \textit{Protium guianense} ($y_i$) in each of 100 contiguous quadrants. Figure \ref{fig:Bubbleplot_botanical} visualizes the data using a bubble plot.

\begin{figure}[h]
    \centering
    \includegraphics[width=0.8\linewidth]{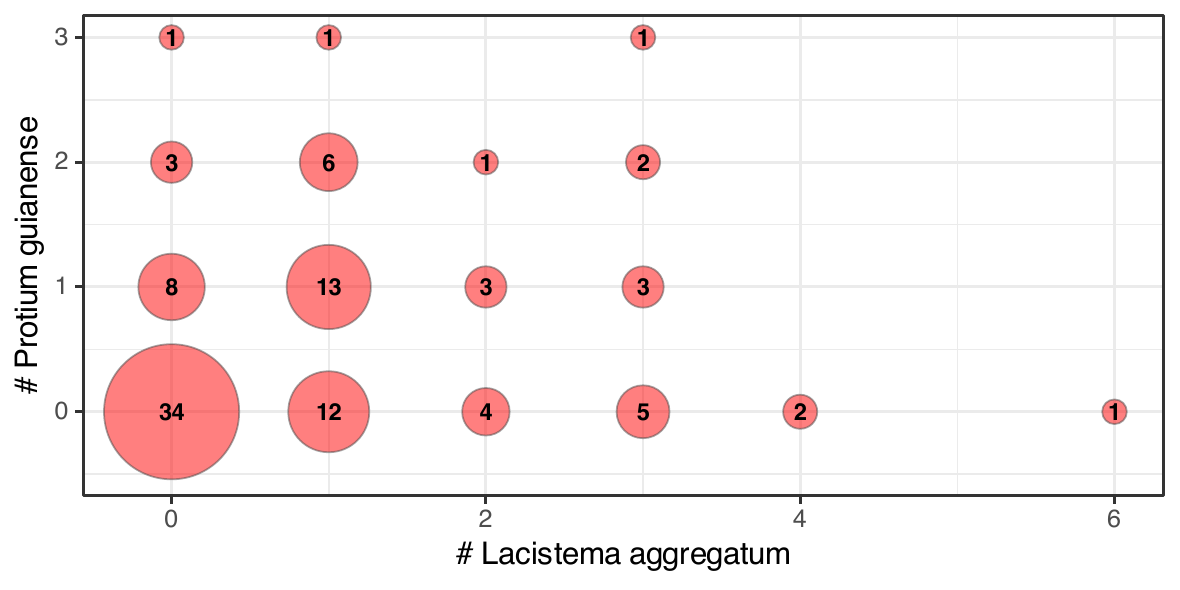}
    \caption{Bubble plot visualizing the ecological data example: The x-axis represents the number of plants of the species Lacistema aggregatum and the y-axis the number of plants of the species Protium guianense in each of 100 contiguous quadrants. The bubbles contain and their sizes visualize the counts of certain combinations of x- and y-values.}
    \label{fig:Bubbleplot_botanical}
\end{figure}

Due to the experimental design, the data can be understood as the realization of an iid~DGP, so we make use of the asymptotics in Section~\ref{Asymptotic Theory: IID Data}. Analyzing the bivariate data for possible cross-dependence is of practical relevance in order to find out if the species tend to occur together, to be unrelated or to be mutually exclusive. In Table \ref{tab:botanical} we report the values of empirical Pearson correlation $\widehat{r}$, the classical empirical rank correlation $\widehat{\tau}$ and $\widehat{\rho}$ and their generalizations $\widehat{\tau}_b$, $\widehat{\gamma}$ and $\widehat{\rho}_b$, which are better suited for the discrete case. This table also reports $90\%$ confidence intervals for the theoretical counterparts of these coefficients as given in \eqref{eq:confidence_interval} and p-values for tests of independence and uncorrelatedness as in and around \eqref{eq:ttest}, where for the independence tests the simplified variance formulas from Corollaries \ref{cor:iid_and_independent_processes_asymptotic_distribution_tau} and \ref{cor:iid_and_independent_processes_asymptotic_distribution_gamma} (and for Pearson correlation from \eqref{eq:Pearson_as_distribution}) were used.


\begin{table}[ht]
\centering
\caption{Correlation estimates for the ecological data example together with 90\%-CIs and p-values for tests for independence and uncorrelatedness. The first p-value uses an estimate for the asymptotic variance under independence, whereas the second p-value uses the same variance estimate as the confidence intervals.}
\begin{tabular}{lcccc}
\toprule
\textbf{coefficient} & \textbf{estimate} & \textbf{90\%-CI} & \textbf{p-value (indep.)} & \textbf{p-value} \\
\midrule
$\widehat{r}$      & 0.114 & [-0.055, 0.277] & 0.2537 & 0.2619 \\
$\widehat{\tau}$      & 0.124 & [0.033, 0.213] & 0.0238 & 0.0232 \\
$\widehat{\rho}$      & 0.183 & [0.047, 0.312] & 0.0266 & 0.0236 \\
$\widehat{\tau}_b$    & 0.199 & [0.048, 0.340] & 0.0253 & 0.0264 \\
$\widehat{\gamma}$    & 0.302 & [0.078, 0.497] & 0.0351 & 0.0192 \\
$\widehat{\rho}_b$    & 0.222 & [0.056, 0.376] & 0.0266 & 0.0240\\
\bottomrule
\end{tabular}
\label{tab:botanical}
\end{table}

When first being interested in the question if there is any dependence at all, that is, considering the tests of independence, all rank-correlation-based tests speak the same language, yielding p-values below $5\%$. The test based on Pearson correlation, on the contrary, yields a p-value way above $10\%$. This could be due to a monotonic, but nonlinear form of dependence, which the rank correlations can capture, implying that the respective tests have power in this direction, whereas Pearson correlation is tailored to measuring linear dependence. The tests for uncorrelatedness, i.e., for the respective correlations being 0, have very similar p-values as the tests for independence. Thus, in this case using the null hypothesis of independence to simplify the variance formula does not lead to efficiency gains. 

The values of the correlation coefficients inform us about direction and strength of dependence. All coefficients are positive, but Pearson correlation with a value of 0.114 is smaller than all the rank correlations (and insignificant as discussed above), in particular than the generalized rank correlations, which range from 0.199 to 0.302 and which should be used here for measuring strength of dependence due to the discreteness of the data. Thus, we can conclude on a (mild) tendency of co-occurrence for both species, where the dependence does not appear to be a linear one.


\subsection{Accidents Data}
\label{Accidents Data Example}
Our second example is about bivariate time series data, namely the bivariate counts $(x_i,y_i)$, $i=1,\ldots,365$, which refer to the daily numbers of daytime ($x_i$) and nighttime ($y_i$) road accidents in the Schiphol area (Netherlands) during the year 2001. The  data have been read from Figure~1 in \citet{pedeli11}. These authors showed that the data exhibit significant serial dependence and propose to model the data by a bivariate INAR$(1)$ model. In Figure \ref{fig:bubble_accidents} we again depict the empirical joint distribution in a bubble plot. Figure \ref{fig:timeseries_accidents} presents a time plot of both series. 
\begin{figure}[h]
  \centering
    \includegraphics[width=\linewidth]{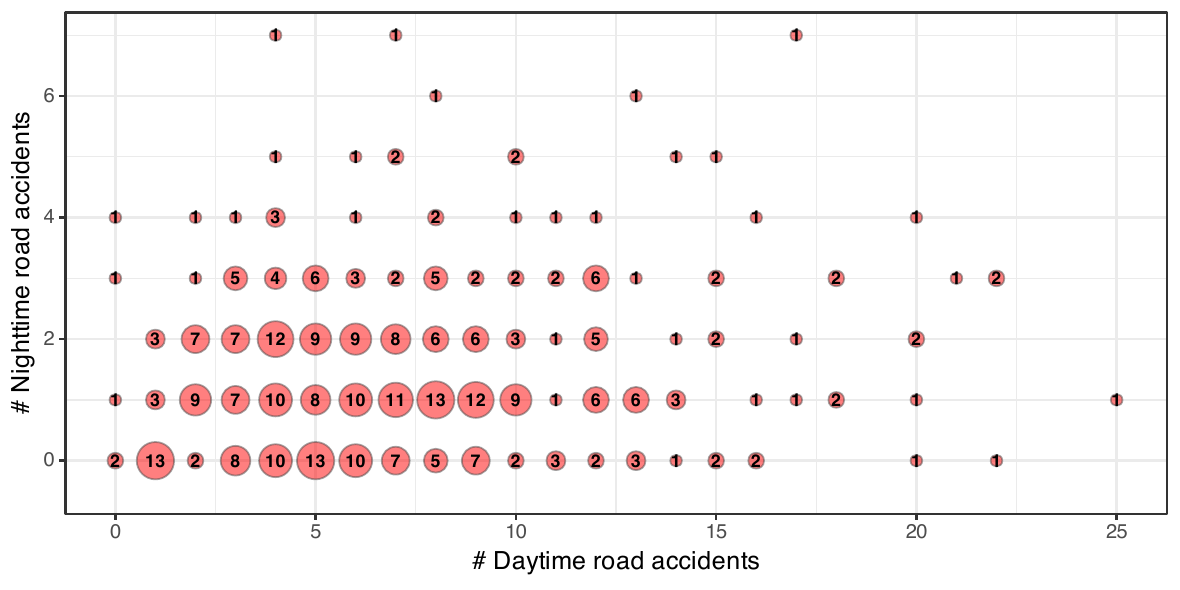}
    \caption{Bubble plot: The x-axis represents the number of daytime and the y-axis the number of nighttime road accidents in the Schiphol area (Netherlands) during the year 2001. The bubbles contain and their sizes visualize the counts of certain combinations of x- and y-values.}
    \label{fig:bubble_accidents}
\end{figure}

\begin{figure}[h]
\centering
\includegraphics[width=\linewidth]{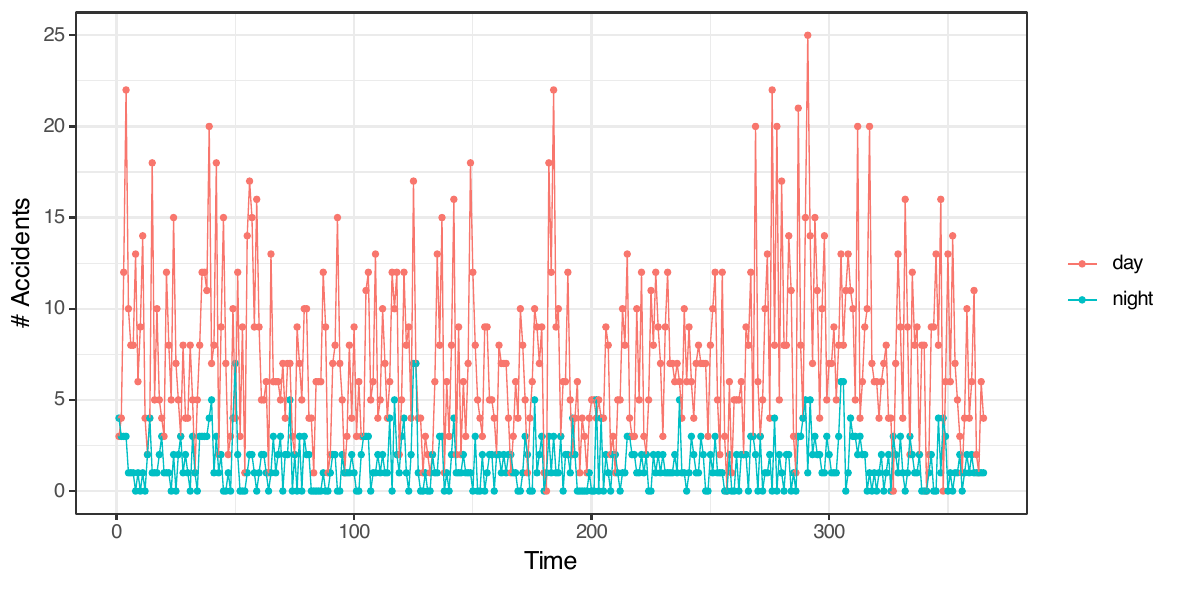}
\caption{Time series plot: The x-axis measures the days in 2001 and the y-axis the numbers of road accidents that occurred during that day in the Schiphol area (Netherlands) either during daytime or during nighttime.}
\label{fig:timeseries_accidents}
\end{figure}

Here, our interest is in the cross-dependence between the daytime and nighttime accidents, where hypothesis tests and confidence intervals have to account for the serial dependence (without imposing a specific time series model) as described in Section~\ref{Asymptotic Theory: Time Series}. All rank correlations and Pearson correlation indicate mild positive positive dependence (see Table \ref{tab:accidents}) with values around 0.1. Interestingly and contrary to the first data example, Pearson correlation is slightly larger than all of the rank correlations here. Usually, Pearson correlation being larger than Kendall's $\tau$ or Spearman's $\rho$ is associated with linear dependence, that is, a form of dependence that can be captured well by Pearson correlation. A more in-depth explanation is as follows: The shape of the bubble plot suggests that the joint distribution can be well approximated by an elliptical distribution. For elliptical distributions, firstly, the dependence is fully characterized by Pearson correlation $r$, and secondly, there are well-known functional relationships between $r$ and the classical rank correlations $\tau$ and $\rho$, which show that these rank correlations are always smaller than $r$ \citep{embrechts2002}. The respective generalized rank correlations are greater than $\tau$ and $\rho$, but still smaller than $r$ here. 

\begin{table}[ht]
\centering
\caption{Correlation estimates for the accidents data example together with 90\%-CIs and p-values for tests for independence and uncorrelatedness. The first p-value uses an estimate for the asymptotic variance under independence, whereas the second p-value uses the same variance estimate as the confidence intervals.}
\begin{tabular}{lcccc}
\toprule
\textbf{coefficient} & \textbf{estimate} & \textbf{90\%-CI} & \textbf{p-value (indep.)} & \textbf{p-value} \\
\midrule
$\widehat{r}$      & 0.142  & [0.039, 0.243] & 0.0071 & 0.0218 \\
$\widehat{\tau}$      & 0.081  & [0.010, 0.151] & 0.0170 & 0.0610 \\
$\widehat{\rho}$      & 0.120  & [0.014, 0.224] & 0.0177 & 0.0596\\
$\widehat{\tau}_b$    & 0.096  & [0.011, 0.179] & 0.0173 & 0.0620 \\
$\widehat{\gamma}$    & 0.113  & [0.014, 0.210] & 0.0178  & 0.0588\\
$\widehat{\rho}_b$    & 0.125  & [0.015, 0.232] & 0.0177& 0.0591 \\
\bottomrule
\end{tabular}\label{tab:accidents}
\end{table}

When it comes to inference, all correlations are significantly different from 0. Interestingly, again contrary to the first data example, the p-values nevertheless strongly depend on the actual kind of estimating the variance. If we make use of the component-wise independence as postulated by the null hypothesis (see Section~\ref{Variance Formulas under Independence TS}), we get p-values $<2\%$ throughout, i.e., we conclude on significant dependence even on the 5\%-level. If estimating the variance without imposing component-wise independence as in Section~\ref{subsec:VarianceEstimationTS} (as when computing confidence intervals), the p-values of the tests become larger with values around 6\%.  Thus, there are efficiency gains from imposing the null hypothesis of independence on the variance. Note that since the temporal dependence is mild in this application (all first order-autocorrelations and cross-autocorrelations for the two series are between 0.1 and 0.2), our HAC-type variance estimators should not suffer from substantial size distortions.

Altogether, we conclude on mild positive dependence between the daytime and nighttime accidents on a day~$i$, which appears reasonable as for example general weather conditions or holidays agree between~$x_i$ and $y_i$ for a given day~$i$.


\section{Conclusion}
\label{Conclusions}


We provide a comprehensive treatment of asymptotic inference for the classical rank correlations Kendall's $\tau$ and Spearman's $\rho$ as well as their generalizations for discrete random variables $\tau_b$, Goodman-Kruskal's $\gamma$, and grade correlation. Under iid or time series assumptions, and not restricting ourselves to continuous variables, we derive limiting distributions by using asymptotic results for U-statistics and propose variance estimators that enable the construction of confidence intervals and tests. Furthermore, we derive simplified variance formulas in the case of independence, which form the basis for independence tests. 

Although asymptotic confidence intervals and tests are the most classical approach to statistical inference, and Kendall's $\tau$ and Spearman's $\rho$ have favorable properties compared to Pearson correlation and are commonly employed in descriptive statistics, it is quite surprising that asymptotic inference for these rank correlations has been lacking fundamental results and, thus, has rarely been used in inductive statistics. Consequently, sound inferential procedures (with the exception of independence tests) are not implemented in most statistical software. The results of this paper, which of course build on earlier foundational work such as \citet{hoeffding1948} and \citet{dehling2017}, allow for a change in this respect. 

In the iid setting, the proposed asymptotic inference performs very well and we posit that it ought to be the method of choice as in related settings. Certainly, resampling-based inference for rank correlations constitutes a valid alternative. It thus seems to be a fruitful endeavor for future research to compare the performance of the two approaches. 

For time series data, we discover some problems with variance estimation under strong temporal dependence and in small to medium-sized samples. This, however, is not surprising as variance estimation under temporal dependence in similar settings (e.g.\ for the sample mean or linear regression coefficients) is a notoriously difficult problem, plagued by oversized tests and confidence intervals with undercoverage. 
Despite these issues, we consider our asymptotic confidence intervals and tests to be a substantial improvement on the status quo in time series settings, where inferential procedures are hardly available so far. Future research in this direction is undoubtedly crucial. For example, self-normalization might be a way to tackle these problems \citep{shao2015}, or at least the improvement of the HAC estimator by systematically comparing different choices of the bandwidth or kernel.

While Kendall's $\tau$, Spearman's $\rho$, and their variants for the discrete case, are the classical and most widely-used rank correlations, inference for other representatives of this class is also relevant and can be approached by using the results of this paper. 
For example, there are asymmetric versions of $\tau_b$ and grade correlation: Somers' $D$ defined as $\tau(X,Y)/\tau(X,X)$ \citep{somers1962, newson2002}, and the recently introduced asymmetric grade correlation defined as $\rho(X,Y)/\rho(X,X)$ \citep{walz2025}, measure directed dependence of $Y$ on $X$, e.g.\ to assess the predictive potential. In the continuous case, they equal $\tau$ and $\rho$, respectively, and for a dichotomous $Y$, they are both related by a one-to-one mapping to the area under the receiver operating characteristic (ROC) curve, short area under the curve (AUC), which is a popular measure to assess the performance of binary classifiers \citep{walz2025}. For all these measures, inference tools can be constructed by using the techniques introduced in this paper. \citet{walz2025} do this for the case of asymmetric grade correlation and develop tests for the equality of two grade correlations as well. Spearman's $\rho$ and asymmetric grade correlation are also closely related to the slope coefficients of rank-rank regressions, for which inference has recently been treated by \citet{chetverikov2023}. 

\bibliography{literature_inference_rank_correlations}
\clearpage

\appendix
\appendixpage
\numberwithin{equation}{section}
\numberwithin{table}{section}
\numberwithin{figure}{section}
\numberwithin{theorem}{section}

\section{Asymptotic Theory for U-Statistics under Weak Dependence} 
\label{Asymptotic Theory for U-Statistics under Weak Dependence}

A U-statistic can be viewed as a generalized sample mean, namely as an average over a function of $r$ sample elements instead of only one element. See e.g.\ \citet{lee1990} for an introduction to U-statistics, and \citet{dehling2006} for a discussion of limit theorems for U-statistics under weak dependence. Let $\{\mathbf{X}_i\}_{i \in \mathbb{Z}}$ be a $K$-variate stochastic process, $\mathbf{X}_i^\prime =  \left(X_{1,i}, \ldots, X_{K,i} \right)$, and $\{\mathbf{X}_i\}_{i=1}^n$ a sample of size $n$ arising from it, or in other words, a time series of length $n$. By $\mathbf{X}^{(2)}_i, \mathbf{X}^{(3)}_i, \ldots$, we denote independent copies of $\mathbf{X}_i$, where $\mathbf{X}_i$ itself is also written as $\mathbf{X}^{(1)}_i$ in this context.

\begin{definition}[U-Statistic] \label{def:u_statistic}
	Let $\{\mathbf{X}_i\}_{i=1}^n$ be a $K$-variate time series. Let $r \leq n$ and let $k:\mathbb{R}^K \times \ldots \times \mathbb{R}^K \rightarrow \mathbb{R},\ \left( \mathbf{x}^{(1)}, \ldots, \mathbf{x}^{(r)} \right) \mapsto k \left( \mathbf{x}^{(1)}, \ldots, \mathbf{x}^{(r)} \right)$ be a measurable and symmetric function. Then the U-statistic $U_{\mathbf{X}}(k)$ is defined as
	$$ U_{\mathbf{X}}(k) = \binom{n}{r}^{-1} \sum_{1 \leq i_1 < \ldots < i_r \leq n} k \left( \mathbf{X}_{i_1}, \ldots, \mathbf{X}_{i_r} \right),$$
    that is, the sum runs over all increasing $r$-tuples of indices from 
    $\{1, \ldots, n\}$. The function $k$ is called kernel (of size $r$) of the U-statistic $U_{\mathbf{X}}(k)$ (of order $r$). 
\end{definition}

A U-statistic is an unbiased estimator for the parameter
\begin{equation} \label{eq:theta}
\theta := \E\big[ k ( \mathbf{X}^{(1)}, \ldots \mathbf{X}^{(r)} )  \big].
\end{equation}
The use of symmetric kernels in the definition of a U-statistic is not a strong restriction, because for any non-symmetric kernel $k$, there exists a U-statistic with the symmetric kernel
\begin{equation} \label{eq:nonsymmetric_kernel}
\tilde{k} \left( \mathbf{x}^{(1)}, \ldots \mathbf{x}^{(r)} \right) = \frac{1}{r!} \sum_{1 \leq j_1 \neq \ldots \neq j_r \leq r} k \left( \mathbf{x}^{(j_1)}, \ldots, \mathbf{x}^{(j_r)} \right),
\end{equation}
where the sum is taken over all permutations of $\{1,\ldots,r\}$, which has the same expectation $\theta$.

We usually leave out the subscript ${\mathbf{X}}$ in $U_{\mathbf{X}}$ if it is not necessary to emphasize based on which process the U-statistic is defined, and analogously the argument $k$ if the same holds true for the kernel. Furthermore, we omit the time index $i$ in situations where we do not need to distinguish different points in time.

In the main part of the paper, we are only concerned with bivariate time series, where we use the notation $(X,Y)$ instead of $\mathbf{X}$, which allows us to skip the index for the elements of the vector $\mathbf{X}$. Furthermore, only U-statistics of order 2 and 3 show up there. We thus only need up to two independent copies of $(X,Y)$, which we denote by $(X^\prime,Y^\prime)$ and $(X^{\prime\prime},Y^{\prime\prime})$ to avoid the index counting the copies.

A key tool for the analysis of U-statistics is the Hoeffding decomposition \citep{hoeffding1961}. For our purposes, we focus on the leading term, as the lower-order terms do not play a role in the asymptotic analysis, and we consider an asymptotic version of the decomposition: A U-statistic can be rewritten as (under our assumptions spelled out in Proposition \ref{prop:CLT_ustats_univariate} below)  
\begin{equation} \label{eq:Hoeffding_decomposition}
	U = \theta + \frac r n \sum_{i=1}^n k_1(\mathbf{X}_i) + \mathcal{O}_{\p}(n^{-1})
\end{equation}
with
\begin{equation} \label{eq_k1}
	k_1(\mathbf{x}) = \E\big[k(\mathbf{x}, \mathbf{X}^{(2)}, \ldots \mathbf{X}^{(r)} )\big] - \theta,
\end{equation}
where $\theta$ is defined in \eqref{eq:theta}.

This decomposition shows that the U-statistic $U$ can essentially be treated as a sample mean of $k_1(\mathbf{X}_i)$.  This is the essence behind central limit theorems for U-statistics, which consequently can be derived from usual central limit theorems. The kernel $k_1$ (derived from $k$ and the distribution of $(\mathbf{X}^{(2)}, \ldots \mathbf{X}^{(r)} )$) is of order 1 and (amongst others) plays a central role in determining limiting variances in central limit theorems for U-statistics. The following central limit theorem is due to \citet[Theorem 1c]{denker1983}. We restrict attention to the case of bounded kernels, which we exclusively deal with in this paper. Boundedness could be replaced by a moment condition at the cost of strengthening the mixing condition.
\begin{proposition}[Central Limit Theorem for U-Statistics] \label{prop:CLT_ustats_univariate}
	Let all components of $\{\mathbf{X}_i\}_{i \in \mathbb{Z}}$ satisfy Assumption \ref{ass:iid} or Assumption \ref{ass:time_series}, and let $U$ be a U-statistic with a bounded symmetric kernel $k$ of order $r$. Then, it holds that $$\sqrt{n} \left( U -  \theta \right) \stackrel{d}{\rightarrow} \mathcal{N}(0,r^2 \sigma_{U}^2),$$
	where under Assumption \ref{ass:iid}, we have
    $$\sigma_U^2 =  \E \big[ k_1(\mathbf{X})^2 \big],$$ 
    and under Assumption \ref{ass:time_series}, we have
	$$\sigma_U^2 =  \sum_{h=-\infty}^{h=\infty} \alpha_{k_1(\mathbf{X})}(h).$$
	Here, $$\alpha_{k_1(\mathbf{X})}(h) = \E \big[ k_1(\mathbf{X}_i) k_1(\mathbf{X}_{i+h}) \big] $$
	denotes the autocovariance function at lag $h$ of $\{k_1(\mathbf{X}_i)\}_{i \in \mathbb{Z}}$, and $\theta$ and $k_1 (\mathbf{x})$ are defined in \eqref{eq:theta} and \eqref{eq_k1}, respectively.
\end{proposition}

The asymptotic Hoeffding decomposition \eqref{eq:Hoeffding_decomposition} explains the formula for the asymptotic variance $\sigma_U^2$. It has the form of the asymptotic variance of the process $\{k_{1}(\mathbf{X}_i)\}_{i \in \mathbb{Z}}$, i.e.\ the variance of the sample mean (scaled by $\sqrt{n}$) of this process, which is a long-run variance matrix in the time series case. Note that if the kernel $k_1$ is degenerate, i.e.\ $\Var(k_1(X_i))=0$, the limiting distribution becomes degenerate, so a stronger scaling factor than $\sqrt{n}$ is needed to obtain a non-degenerate limiting distribution. But this is not relevant for our paper as the kernels showing up here are only degenerate in pathological cases, see \citet{hoeffding1948, dengler2010} for a discussion in the case of $\widehat{\tau}$.

There is only a small literature on variance estimation for U-statistics: For the iid case, see \citet{shirahata1992}, \citet{wang2014} and \citet[Section 4]{iwashita2024}; for the time series case the only paper that we are aware of is \citet{dehling2017}, which focuses on U-statistics of order two.



We finally provide a multivariate central limit theorem, which is the basis for many results in this paper and which follows from the univariate version in Proposition \ref{prop:CLT_ustats_univariate} by the Cram\'{e}r-Wold device and a combinatorial argument. We consider a vector of U-statistics and allow them to depend on different underlying processes. For a multivariate central limit theorem for U-statistics in the iid case see \citet[Theorem 7.1]{hoeffding1948}. 

\begin{proposition}[Multivariate Central Limit Theorem for U-Statistics] \label{prop:clt_Ustats_multivariate}
	Let $\{(X_{1,i},\ldots,X_{K,i})\}_{i \in \mathbb{Z}}$ be a $K$-dimensional process, where all components satisfy Assumption \ref{ass:iid} or Assumption \ref{ass:time_series}. Let $\mathbf{U} = (U^{(1)}_{\mathbf{X}_{I^{(1)}}}(k^{(1)}),\ldots,U^{(D)}_{\mathbf{X}_{I^{(D)}}}(k^{(D)}))^\prime$ be a $D$-dimensional vector of U-statistics with bounded symmetric kernels $\mathbf{k} = (k^{(1)},\ldots,k^{(D)})^\prime$ of orders $(r^{(1)},\ldots,r^{(D)})^\prime$ and underlying stochastic processes $\{(\mathbf{X}_{I^{(1)},i},\ldots, \mathbf{X}_{I^{(D)},i})\}_{i \in \mathbb{Z}}$, where $I^{(d)} \subset \{1,\ldots,K\}$ for $d=1,\ldots,D$ are subsets of the original index set. Furthermore, let $\theta^{(d)}$ and $k_1^{(d)}$, $d=1,\ldots,D$, denote the quantities defined in \eqref{eq:theta} and \eqref{eq_k1}, respectively, for the $d$-th U-statistic, and define $\mbox{\boldmath $\theta$} := (\theta^{(1)},\ldots,\theta^{(D)})^\prime$. Then it holds that $$\sqrt{n} \left( \mathbf{U} -  \mbox{\boldmath $\theta$} \right) \stackrel{d}{\rightarrow} \mathcal{N}(\mathbf{0},\mathbf{\Sigma}_{\mathbf{U}})
 \quad\text{with}\quad
 \mathbf{\Sigma}_{\mathbf{U}} =   \left( r^{(v)} r^{(w)} \sigma_{vw} \right)_{v=1,\ldots,D, w=1,\ldots,D},$$
	where under Assumption \ref{ass:iid}, we have
	 $$ \sigma_{vw} =  \E\big[ k^{(v)}_{1} (\mathbf{X}_{I^{(v)},i})\, k^{(w)}_{1}(\mathbf{X}_{I^{(w)},i}) \big],$$
    and under Assumption \ref{ass:time_series}, we have
	$$\sigma_{vw} = \sum_{h=-\infty}^{\infty} \E\big[ k^{(v)}_{1} (\mathbf{X}_{I^{(v)},i})\, k^{(w)}_{1}(\mathbf{X}_{I^{(w)},i+h}) \big].$$
\end{proposition}

\begin{proof}
	Consider a linear combination of U-statistics, $\widetilde{U} = \mbox{\boldmath $\lambda$}^\prime \mathbf{U}$ with $\mbox{\boldmath $\lambda$}^\prime = (\lambda_1,\ldots,\lambda_D) \in \mathbb{R}^D$. Define $r^{\max} = \max (r^{(1)},\ldots,r^{(D)})$, and note that every U-statistic 
 $$U^{(d)}_{\mathbf{X}_{I^{(d)}}}(k^{(d)}) =  \binom{n}{r^{(d)}}^{-1} \sum_{1 \leq i_1 < \ldots < i_{r^{(d)}} \leq n} k^{(d)} \left( \mathbf{X}_{I^{(d)},i_1}, \ldots, \mathbf{X}_{I^{(d)},i_{r^{(d)}}} \right)$$
 can be rewritten as 
    $$U^{(d)}_{\mathbf{X}_{I^{(d)}}}(k^{(d)}) = \binom{n}{r^{\max}}^{-1} \sum_{1 \leq i_1 < \ldots < i_{r^{\max}} \leq n} \widetilde{k}^{(d)} \left( \mathbf{X}_{I^{(d)},i_1}, \ldots, \mathbf{X}_{I^{(d)},i_{r^{max}}} \right),$$
    where
    $$\widetilde{k}^{(d)} \left( \mathbf{X}_{I^{(d)},i_1}, \ldots, \mathbf{X}_{I^{(d)},i_{r^{max}}} \right):=\binom{r^{max}}{r^{(d)}}^{-1}\sum_{1\le j_1<...<j_{r^{(d)}}\le r^{max}}k^{(d)} \left( \mathbf{X}_{I^{(d)},i_1}, \ldots, \mathbf{X}_{I^{(d)},i_{r^{(d)}}} \right).$$
    Here, each of the summands in the upper representation shows up $\binom{n}{r^{\max}}\binom{r^{\max}}{r^{(d)}} / \binom{n}{r^{(d)}}$ times, but is then divided by this factor to ensure equality.
    Using this, $\widetilde{U}$ can be rewritten as  
	\begin{align*}
		\widetilde{U} = \binom{n}{r^{\max}}^{-1} \sum_{1 \leq i_1 < \ldots < i_{r^{\max}} \leq n} \left( \lambda_1 \widetilde{k}^{(1)}(\mathbf{X}_{I^{(1)},i_1}, \ldots, \mathbf{X}_{I^{(1)},i_{r^{max}}} ) + \ldots  +  \lambda_D \widetilde{k}^{(D)}(\mathbf{X}_{I^{(D)},i_1}, \ldots, \mathbf{X}_{I^{(D)},i_{r^{max}}} ) \right)
	\end{align*}
	 and is thus itself a U-statistic with a bounded and symmetric kernel. Hence, we can apply the univariate central limit theorem from Proposition \ref{prop:CLT_ustats_univariate} to $\widetilde{U}$. Invoking the Cram\'{e}r-Wold device \citep[Chapter 2]{vandervaart2000} completes the proof.
\end{proof}

\clearpage
\section{Proofs and Additional Lemmas}
\label{Proofs and Additional Lemmas}

\begin{proof}[Proof of Lemma \ref{lem:grade}]
    The first statement immediately follows from \eqref{eq:grade_and_spi} and the relation between $H(x)$ and $\mathrm{sgn}(x)$. 
 For (ii), we calculate:
    \begin{align*}
        &\ \E[\mathrm{sgn}(x-X)\, \mathrm{sgn}(y-Y)]\\ 
        =&\ \p(X<x, Y<y)+\p(X>x, Y>y)-\p(X<x, Y>y)-\p(X>x, Y<y)\\
        =&\ F_{X,Y}(x^-, y^-)+1-F_X(x)-F_Y(y)+F_{X,Y}(x,y)+F_{X,Y}(x^-,y)-F_X(x^-)+F_{X,Y}(x, y^-)-F_{Y}(y^-)\\
        =&\ 4G_{X,Y}(x,y)-2G_X(x)-2G_Y(y)+1.
    \end{align*}
    
 The proof of (iii) is elementary. 

For (iv), the expectation of $G_X(X)$ is calculated in \citet[Equation (13)]{neslehova2007}, and the variance of $G_X(X)$ follows from \citet[(27)]{neslehova2007} after correcting a typo (the square occurring in (27) has to be within the expectation) and noting that $\E[(F_X(X) - F_X(X^-))^2] = \p(X=\copyX = \copytwoX)$.
\end{proof}

\begin{lemma}\label{lem:GradeExpec_cond}
It holds that 
\begin{enumerate}[(i)]
    \item $\E\big[G_X(X)|X\le x\big]=0.5\, F_X(x)$,
    \item $\E\big[G_X(X)|X< x\big]=0.5\, F_X(x^-)$,
    \item $\E\big[F_X(X^-)|X\le x\big]=\frac{1}{F_X(x)}\left(F_X(x^-)\, \E\big[F_X(X^-)|X<x\big]+\p(X=x)\, F_X(x^-)\right)$.
\end{enumerate}

\end{lemma}
\begin{proof}
    For (i), calculate
\begin{align}\label{eq:cond_grade1}
    \E\big[G_X(X)|X\le x\big]&=\E\big[\E\big[\mathds{1}_{\{X'<X\}}+0.5\, \mathds{1}_{\{X'=X\}}|X\big]|X\le x\big]\notag \\
    &=\frac{1}{F_X(x)}\E\big[\left(\mathds{1}_{\{X'<X\}}+0.5\, \mathds{1}_{\{X'=X\}}\right)\mathds{1}_{\{X\le x\}}\big]\notag \\
    &=\frac{1}{F_X(x)}\E\big[\mathds{1}_{\{X'<X, X\le x\}}+0.5\, \mathds{1}_{\{X'=X, X\le x\}}\big].
\end{align}
Since $X$ and $X'$ are iid copies, it also holds that 
\begin{align}\label{eq:cond_grade2}
    \E\big[G_X(X')|X'\le x\big]&=\frac{1}{F_X(x)}\E\big[\mathds{1}_{\{X<X', X'\le x\}}+0.5\, \mathds{1}_{\{X'=X, X'\le x\}}\big],
\end{align}
and the left hand sides of equations (\ref{eq:cond_grade1}) and (\ref{eq:cond_grade2}) coincide. As a consequence,
adding both equations yields
\begin{align*}
    2\E\big[G_X(X)|X\le x\big]
    &=\frac{1}{F_X(x)}\E\big[\mathds{1}_{\{X'<X, X\le x\}}+\mathds{1}_{\{X'=X, X\le x\}}+\mathds{1}_{\{X<X', X'\le x\}}\big].
\end{align*}
Hence, with the event equality 
\begin{align*}
    \{X'<X, X\le x\}\cup\{X'=X, X\le x\}\cup \{X<X', X'\le x\}=\{X\le x, X'\le x\},
\end{align*}
we obtain
\begin{align*}
    \E\big[G_X(X)|X\le x\big]
    =\frac{1}{2F_X(x)}\E\big[\mathds{1}_{\{X\le x, X'\le x\}}\big]=\frac{1}{2F_X(x)}F_X(x)^2=\frac{F_X(x)}{2},
\end{align*}
where the second equality used that $X$ and $X'$ are iid copies.

(ii) follows by the same calculations as (i), and (iii) follows by the definition of the conditional expectation.
\end{proof}

The following lemma contains the key ingredients of the asymptotic variance formulas of all the empirical rank correlations discussed in this paper.

\begin{lemma}\label{lem:variance}
    It holds that 
    $$
        k_1^{(\tau)}(x,y)= 4\, G_{X,Y} (x,y)  - 2 \big(G_X(x) + G_Y (y)\big) + 1 - \tau,
    $$
    $$
        k_1^{(\rho)}(x,y)= 4\big(\E[G_{X,Y} (x,Y')] + \E[G_{X,Y} (X',y)] + G_X(x)G_Y(y)-G_X(x)-G_Y(y)\big)+1-\rho,
    $$
	\begin{align*}
		k_{1}^{(\nu)} (x,y) 
		= \p(X=x) + \p(Y=y) - \p(X=x,Y=y) - \nu,
	\end{align*}
  \begin{align*}
		k^{(\tau_X)}_1 (x,y):=k^{(\tau_X)}_1 (x)
		:= 1-\p(X=x) - \tau_X,\quad 
  k^{(\tau_Y)}_1 (x,y):=k^{(\tau_Y)}_1 (y)
		:= 1-\p(Y=y) - \tau_Y,
	\end{align*}
    \begin{align*}
		k_{1}^{(\rho_X)} (x,y) &:=k_{1}^{(\rho_X)} (x):= 1-\p(X=x)^2  - \rho_X, \ k_{1}^{(\rho_Y)} (x,y):=k_{1}^{(\rho_Y)} (y) &:= 1-\p(Y=y)^2 - \rho_Y.
	\end{align*}
   \end{lemma}

    \begin{proof}
    For the first equality, note that
    \begin{align*}
        k_1^{(\tau)}(x,y)=
        &\ \E\big[k^{(\tau)}\big((x,y), (X',Y')\big)\big]-\tau\\
        =&\ \E\big[\mathrm{sgn}(x-X')\, \mathrm{sgn}(y-Y')\big]-\tau\\
        =&\ 4G_{X,Y}(x,y)-2G_X(x)-2G_Y(y)+1-\tau,
    \end{align*}
    where part (ii) of Lemma \ref{lem:grade} was used in the last step.
    For the second equality, it holds that
    \begin{align*}
        k_1^{(\rho)}(x,y) =&\ \E\big[k^{(\rho)}\big((x,y), (X',Y'), (X'', Y'')\big)\big]-\rho\\
         = &\ \E\big[\mathrm{sgn}(x-X')\, \mathrm{sgn}(y-Y'')\big] + \E\big[\mathrm{sgn}(X'-x)\, \mathrm{sgn}(Y'-Y'')\big] \\
         &+ \E\big[\mathrm{sgn}(X'-X'')\, \mathrm{sgn}(Y'-y)\big] - \rho.
    \end{align*}

    We use the independence of the two factors and part (i) of Lemma \ref{lem:grade} to obtain
    $$\E\big[\mathrm{sgn}(x-X')\, \mathrm{sgn}(y-Y'')\big]=\big(2\,G_X(x)-1\big)\big(2\,G_Y(y)-1\big).$$
    By part (i) of Lemma \ref{lem:grade},
    \begin{align*}
        \E\big[\mathrm{sgn}(X'-x)\, \mathrm{sgn}(Y'-Y'')\big]=&\ \E\Big[\E\big[\mathrm{sgn}(X'-x)\, \mathrm{sgn}(Y'-Y'')|Y''\big]\Big]\\
        =&\ \E\big[4\, G_{X,Y} (x,Y'')  - 2\, G_X(x) -2\, G_Y (Y'') + 1\big]\\
        =&\ 4\, \E\big[ G_{X,Y} (x,Y'') \big]  - 2\, G_X(x).
    \end{align*}
    Analogously, it holds that
    \begin{align*}
        \E\big[\mathrm{sgn}(X'-X'')\, \mathrm{sgn}(Y'-y)\big] = 4\, \E\big[ G_{X,Y} (X'',y)\big]  - 2\, G_Y (y).
    \end{align*}
    Collecting terms, we arrive at
    \begin{align*}
        k_1^{(\rho)}(x,y)=&\ \big(2\,G_X(x)-1\big)\big(2\,G_Y(y)-1\big)+4\, \E\big[ G_{X,Y} (x,Y'') \big]  - 2\, G_X(x)  \\
        &+4\, \E\big[ G_{X,Y} (X'',y)\big]  - 2\, G_Y (y) -\rho\\
        =&\ 4\big(\E[G_{X,Y} (x,Y'')] + \E[G_{X,Y} (X'',y)] + G_X(x)G_Y(y)-G_X(x)-G_Y(y)\big)+1-\rho.
    \end{align*}
   Furthermore, 
 \begin{align*}
        k_1^{(\nu)}(x,y) =&\ \E\big[k^{(\nu)}\big((x,y), (X,Y)\big)\big]-\nu
        \ =\ 
        \E\big[\mathds{1}\big((x-X)(y-Y)=0\big)\big]-\nu\\
        =&\
        \p\big(X=x \text{ or } Y=y\big)-\nu
        \ =\  \p(X=x) + \p(Y=y) - \p(X=x,Y=y) - \nu,
\end{align*}
     $$
    k_{1}^{(\tau_X)} (x) =
    \E\big[k^{(\tau)}\big((x, x), (\copyX, \copyX)\big)\big]-\tau_X
        \ =\ \E\big[\mathrm{sgn}(x-X)^2\big]-\tau_X \ =\ \p(X\ne x)-\tau_X.
     $$
     Finally,
     \begin{align*}
    k_{1}^{(\rho_X)} (x) &=
    \E\big[k^{(\rho)}\big((x,x),(\copyX,\copyX),(\copytwoX, \copytwoX)\big)\big]-\rho_X\\
         &= \E\big[\mathrm{sgn}(x-\copyX)\mathrm{sgn}(x-\copytwoX)+\mathrm{sgn}(\copyX-x)\mathrm{sgn}(\copyX-\copytwoX)+\mathrm{sgn}(\copytwoX-x)\mathrm{sgn}(\copytwoX-\copyX)\big]-\rho_X
     \end{align*}
     because $\copyX$ and $\copytwoX$ play interchangeable roles. We can use the independence of $\copyX$ and $\copytwoX$ and simplify $\E\big[\mathrm{sgn}(x-\copyX)\mathrm{sgn}(x-\copytwoX)\big]=(2G_X(x)-1)^2$ due to part (i) of Lemma \ref{lem:grade} as well as 
     \begin{align*}
     \E\big[\mathrm{sgn}(\copyX-x)\mathrm{sgn}(\copyX-\copytwoX)\big]&=\E\big[\E\big[\mathrm{sgn}(\copyX-x)\mathrm{sgn}(\copyX-\copytwoX)|\copytwoX\big]\big]\\
     &=\E\big[4G_{X,X}(x,X)\big]-2G_X(x)
     \end{align*}
     due to parts (ii) and (iv) of Lemma \ref{lem:grade}. We can further simplify
     \begin{align*}
         4\E\big[G_{X,X}(x,X)\big]=&\ \E\big[\min\{F_X(x),F_X(X)\}+\min\{F_X(x^-),F_X(X)\}\\
         &+\min\{F_X(x),F_X(X^-)\}+\min\{F_X(x^-),F_X(X^-)\}\big]
     \end{align*}
     due to the comonotonicity of $(X,X)$. Moreover, 
\begin{align*}
    &\E\big[\min\{F_X(x),F_X(X)\}+\min\{F_X(x^-),F_X(X^-)\}\big]\\
    =&\ \p(X\le x)\E\big[F_X(X)|X\le x\big]+\p(X>x)F_X(x)+\p(X\le x)\E\big[F_X(X^-)|X\le x\big]+\p(X>x)F_X(x^-)\\
    =&\ 2F_X(x)\E\big[G_X(X)|X\le x\big]+(1-F_X(x))(F_X(x)+F_X(x^-))\\
    =&\ F_X(x)+F_X(x^-)-F_X(x)F_X(x^-),
\end{align*}
where Lemma \ref{lem:GradeExpec_cond} was used in the last equation. The other two summands can be transformed as
\begin{align*}
    &\E\big[\min\{F_X(x^-),F_X(X)\}+\min\{F_X(x),F_X(X^-)\}\big]\\
    =&\ \p(X< x)\E\big[F_X(X)|X< x\big]+\p(X\ge x)F_X(x^-)+\p(X\le x)\E\big[F_X(X^-)|X\le x\big]+\p(X>x)F_X(x)\\
    =&\ 2F_X(x^-)\E\big[G_X(X)|X<x\big]+(1-F_X(x^-))F_X(x^-)+(F_X(x)-F_X(x^-))F_X(x^-)+(1-F_X(x))F_X(x)\\
    =&\ F_X(x)+F_X(x^-)-F_X(x)^2-F_X(x^-)^2+F_X(x)F_X(x^-),
\end{align*}
where again the second and third equations are due to Lemma \ref{lem:GradeExpec_cond}. We conclude that
\begin{align*}
        k_{1}^{(\rho_X)} (x)
         =&\ (2G_X(x)-1)^2+4(2\E\big[G_{X,X}(x,X)\big]-G_X(x))-\rho_X\\
         =&\ 2 \big( \E\big[\min\{F_X(x),F_X(X)\}+\min\{F_X(x^-),F_X(X)\}+\min\{F_X(x),F_X(X^-)\}\\ &+\min\{F_X(x^-),F_X(X^-)\}\big]+ 2G_X(x)^2 - 4G_X(x)\big) + 1 - \rho_X\\
         =&\ 4F_X(x)+4F_X(x^-)-2F_X(x)^2-2F_X(x^-)^2+ 4G_X(x)^2 - 8G_X(x) + 1 - \rho_X\\
         =&\ -(F_X(x)-F_X(x^-))^2 + 1 - \rho_X.
\end{align*}
     \end{proof}

\begin{proof}[Proof of Propositions \ref{prop:asymptotic_distribution_gamma_iid} and \ref{prop:asymptotic_distribution_tau}]
	By our multivariate central limit theorem for bounded U-statistics (Proposition \ref{prop:clt_Ustats_multivariate} in Appendix~\ref{Asymptotic Theory for U-Statistics under Weak Dependence}), it holds that
	\begin{align*} \sqrt{n}  \left( \begin{pmatrix}
		\widehat{\tau} \\
		\widehat{\nu}   
	\end{pmatrix}
	-
	\begin{pmatrix}
		\tau \\
		\nu  
	\end{pmatrix}
	\right)
	\stackrel{d}{\rightarrow}&\ \mathcal{N} \left( \mathbf{0}, 
	\begin{pmatrix}
		\sigma_{\tau}^2 & \sigma_{\tau \nu}\\
		\sigma_{\tau \nu} & \sigma_{\nu}^2
	\end{pmatrix}
	\right),\\
 	 \sqrt{n}  \left( \begin{pmatrix}
		\widehat{\tau} \\
		\widehat{\tau}_X\\
        \widehat{\tau}_Y
	\end{pmatrix}
	-
	\begin{pmatrix}
		\tau \\
		\tau_X\\
        \tau_Y
	\end{pmatrix}
	\right)
	\stackrel{d}{\rightarrow}&\ \mathcal{N} \left( \mathbf{0}, 
	\begin{pmatrix}
		\sigma_{\tau}^2 & \sigma_{\tau \tau_X}& \sigma_{\tau\tau_Y}\\
		\sigma_{\tau \tau_X} & \sigma_{\tau_X}^2&\sigma_{\tau_X\tau_Y}\\
        \sigma_{\tau\tau_Y}&\sigma_{\tau_X\tau_Y}&\sigma_{\tau_Y}^2
	\end{pmatrix}
	\right),\\
  	 \sqrt{n}  \left( \begin{pmatrix}
		\widehat{\rho} \\
		\widehat{\rho}_X\\
        \widehat{\rho}_Y
	\end{pmatrix}
	-
	\begin{pmatrix}
		\rho \\
		\rho_X\\
        \rho_Y
	\end{pmatrix}
	\right)
	\stackrel{d}{\rightarrow}&\ \mathcal{N} \left( \mathbf{0}, 
	\begin{pmatrix}
		\sigma_{\rho}^2 & \sigma_{\rho \rho_X}& \sigma_{\rho\rho_Y}\\
		\sigma_{\rho \rho_X} & \sigma_{\rho_X}^2&\sigma_{\rho_X\rho_Y}\\
        \sigma_{\rho\rho_Y}&\sigma_{\rho_X\rho_Y}&\sigma_{\rho_Y}^2
	\end{pmatrix}
	\right).
 \end{align*}
We note that $\gamma, \tau_b$, and $\rho_b$ are differentiable functions of $(\tau, \nu), (\tau, \tau_X, \tau_Y)$, and $(\rho, \rho_X, \rho_Y)$, respectively: $$\gamma(\tau,\nu) = \frac{\tau}{1-\nu},\quad \tau_b(\tau, \tau_X, \tau_Y)=\frac{\tau}{\sqrt{\tau_X\tau_Y}}, \quad 
  \rho_b(\rho, \rho_X, \rho_Y)=\frac{\rho}{\sqrt{\rho_X\rho_Y}}.$$
    As partial derivatives, we obtain \begin{align*}\frac{\partial \gamma}{\partial \tau} =&\ \frac{1}{1-\nu}, \quad \frac{\partial \gamma}{\partial \nu} = \frac{\tau}{(1-\nu)^2},\\
    \frac{\partial \tau_b}{\partial \tau}=&\frac{1}{\sqrt{\tau_X\tau_Y}}, \quad \frac{\partial \tau_b}{\partial \tau_X}=\frac{-0.5\tau}{\tau_X^{3/2}\sqrt{\tau_Y}}, \quad \frac{\partial \tau_b}{\partial \tau_Y}=\ \frac{-0.5\tau}{\tau_Y^{3/2}\sqrt{\tau_X}}, 
\end{align*}
where the derivatives for $\rho_b$ are analogous to those for $\tau_b$. 
These derivatives always exist as $\nu < 1$ and $\tau_X, \tau_Y, \rho_X, \rho_Y>0$ except for pathological cases.

By the delta method \citep[Theorem 3.1]{vandervaart2000}, we thus have
	$$\sqrt{n} \left( \widehat{\gamma} -\gamma \right) \stackrel{d}{\rightarrow} \mathcal{N}(0, \sigma_{\gamma}^2),\quad \sqrt{n} \left( \widehat{\tau}_b -\tau_b \right) \stackrel{d}{\rightarrow} \mathcal{N}(0, \sigma_{\tau_b}^2)\quad \text{and}\quad \sqrt{n} \left( \widehat{\rho}_b -\rho_b \right) \stackrel{d}{\rightarrow} \mathcal{N}(0, \sigma_{\rho_b}^2)$$ 
	with
	\begin{align*}
		\sigma_{\gamma}^2 &= \begin{pmatrix}
			\frac{1}{1-\nu} & \frac{\tau}{(1-\nu)^2}  		
		\end{pmatrix} 
		\begin{pmatrix}
			\sigma_{\tau}^2 & \sigma_{\tau \nu}\\
			\sigma_{\tau \nu} & \sigma_{\nu}^2
		\end{pmatrix}
		\begin{pmatrix}
			\frac{1}{1-\nu} \\
			\frac{\tau}{(1-\nu)^2}  		
		\end{pmatrix} 
		=\frac{1}{(1-\nu)^2} \left( \sigma_{\tau}^2 + \frac{\tau^2}{(1-\nu)^2} \sigma_{\nu}^2 + \frac{2\,\tau}{1-\nu} \sigma_{\tau \nu} \right)\\
		&= \frac{1}{(1-\nu)^2} \left( \sigma_{\tau}^2 + \gamma^2 \sigma_{\nu}^2 + 2\, \gamma \sigma_{\tau \nu} \right)
	\end{align*}
 and 
     \begin{align*}\sigma_{\tau_b}^2 &= \begin{pmatrix}
			\frac{1}{\sqrt{\tau_X\tau_Y}} \\[1em] \frac{-0.5\tau}{\tau_X^{3/2}\sqrt{\tau_Y}} \\[1em] \frac{-0.5\tau}{\tau_Y^{3/2}\sqrt{\tau_X}}
		\end{pmatrix}^{\prime}
	\begin{pmatrix}
		\sigma_{\tau}^2 & \sigma_{\tau \tau_X}& \sigma_{\tau\tau_Y}\\
		\sigma_{\tau \tau_X} & \sigma_{\tau_X}^2&\sigma_{\tau_X\tau_Y}\\
        \sigma_{\tau\tau_Y}&\sigma_{\tau_X\tau_Y}&\sigma_{\tau_Y}^2
	\end{pmatrix}
		\begin{pmatrix}
			\frac{1}{\sqrt{\tau_X\tau_Y}} \\[1em] \frac{-0.5\tau}{\tau_X^{3/2}\sqrt{\tau_Y}} \\[1em] \frac{-0.5\tau}{\tau_Y^{3/2}\sqrt{\tau_X}}
		\end{pmatrix}\\
  &=
  \frac{1}{\tau_{X}\tau_{Y}}\left(\sigma_{\tau}^2-\tau\left(\frac{\sigma_{\tau\tau_{X}}}{\tau_{X}}-\frac{\sigma_{\tau\tau_{Y}}}{\tau_{Y}}\right)+0.25\tau^2\left(\frac{\sigma_{\tau_{X}}^2}{\tau_{X}^2}+\frac{\sigma_{\tau_{Y}}^2}{\tau_{Y}^2}+2\frac{\sigma_{\tau_{X}\tau_{Y}}}{\tau_{X}\tau_{Y}}\right)\right).
  \end{align*}
 The expression for $\sigma_{\rho_b}^2$ is the same as that for $\sigma_{\tau_b}^2$, but by substituting $\tau\mapsto\rho$, $\tau_X\mapsto\rho_X$, $\tau_Y\mapsto\rho_Y$. 
The expressions for $k_1^{(\nu)},  k_{1}^{(\tau_X)}, k_{1}^{(\tau_Y)},  k_{1}^{(\rho_X)}, k_{1}^{(\rho_Y)}$ are provided by Lemma~\ref{lem:variance}.
\end{proof}

\begin{proof}[Proof of Proposition \ref{prop:consistency_variance_iid}]
    Adding 0 and applying the triangle inequality yields:

\begin{align*} 
      \left|  \widehat{\sigma}_{lm,\textup{iid}} - \sigma_{lm,\textup{iid}} \right| &\leq r^{(l)} r^{(m)} \bigg| \frac 1 n \sum_{i=1}^n \widehat{k}_{1}^{(l)} (X_i,Y_i) \widehat{k}_{1}^{(m)} (X_i,Y_i) - \frac 1 n \sum_{i=1}^n k_{1}^{(l)} (X_i,Y_i) k_{1}^{(m)} (X_i,Y_i) \bigg|\\
      &\quad\;\; + \bigg|\frac 1 n \sum_{i=1}^n k_{1}^{(l)} (X_i,Y_i) k_{1}^{(m)} (X_i,Y_i) - \E \big[ k_{1}^{(l)} (X,Y)\, k_{1}^{(m)} (X,Y) \big]   \bigg|.
\end{align*}

For the first term,
\begin{align*} 
&   \left| \frac 1 n \sum_{i=1}^n \widehat{k}_{1}^{(l)} (X_i,Y_i) \widehat{k}_{1}^{(m)} (X_i,Y_i) - \frac 1 n \sum_{i=1}^n k_{1}^{(l)} (X_i,Y_i) k_{1}^{(m)} (X_i,Y_i) \right| \\
& \leq \sup_{x,y} \left| \widehat{k}_{1}^{(l)} (x,y) \widehat{k}_{1}^{(m)} (x,y) -  k_{1}^{(l)} (x,y) k_{1}^{(m)} (x,y) \right| \stackrel{p}{\rightarrow} 0.
\end{align*}
The convergence in probability for the products of kernels with indices $l,m \in \{\tau, \tau_X, \tau_Y, \nu, \rho_X, \rho_Y\}$ follows by the classical multivariate Glivenko--Cantelli theorem and the continuous mapping theorem since the $k_{1}^{(l)}$ are continuous functionals of the joint CDF $F_{X,Y}(x,y)$. The estimated kernel~$\widehat{k}_{1}^{(\rho)}$ contains an additional estimation step since the function $\widehat{g}_X (x) = \frac 1 n \sum_{i=1}^n \widehat{G}_{X,Y} (x,Y_i)$  emerges as an estimator of $g_X (x) = \E[G_{X,Y} (x,Y) ]$. Due to the Glivenko-Cantelli theorem for $F_{X,Y}(x,y)$ and the continuous mapping theorem (the expectation of a bounded function and thus uniformly integrable function is a continuous mapping), $\widehat{g}_X (x)$ converges uniformly in probability in $x$ to $g_X (x)$, and analogously $\widehat{g}_Y (y)$ in $y$ to $g_Y (y)$. So the above convergence in probability holds also for products involving the kernel $\widehat{k}_{1}^{(\rho)}$.

The second term converges to 0 in probability by the classical weak law of large numbers (the expectation $\E \big[ k_{1}^{(l)} (X,Y)\, k_{1}^{(m)} (X,Y) \big]$ is always finite because the $k_{1}^{(l)}$ are bounded functionals).

Thus, 
$$\left|  \widehat{\sigma}_{lm,\textup{iid}} - \sigma_{lm,\textup{iid}} \right| \stackrel{p}{\rightarrow} 0.$$
\end{proof}

The essence behind Corollary \ref{cor:iid_and_independent_processes_asymptotic_distribution_tau} is the subsequent Lemma~\ref{lem:kernels_under_independence}, which follows from Lemma~\ref{lem:variance} by the additional independence assumption.

\begin{lemma} \label{lem:kernels_under_independence}
    Let $X$ and $Y$ be independent. Then, it holds for $k_1^{(\tau)}$ and $k_1^{(\rho)}$ from \eqref{eq:k1tau_formula} and \eqref{eq:k1rho_formula} that 
    $$k_1^{(\tau)} (x,y) =k_1^{(\rho)} (x,y) = 4 \big( G_{X} (x) - 1/2\big) \big( G_{Y} (y) - 1/2\big).$$
    Furthermore, we have that
    $$
    \E [ k^{(\tau)}_1 (X,Y)^2  ] = \E [ k^{(\rho)}_1 (X,Y)^2  ] = \frac{1}{9}  \big(1 - \zeta_2(X)\big) \big(1 - \zeta_2(Y)\big).
    $$
    Finally, $k_1^{(\nu)}$ from Proposition \ref{prop:asymptotic_distribution_gamma_iid} simplifies to
    $$
    k_1^{(\nu)}(x,y)=  \p(X=x) + \p(Y=y) - \p(X=x)\, \p(Y=y) - \nu.
    $$
\end{lemma}

\begin{proof}
    Under independence of $X$ and $Y$, we have
    $F_{X,Y}(x,y) = F_X(x)\,F_Y(y)$ and thus
    $$
    G_{X,Y}(x,y) = \frac{1}{4}\,\big(F_X(x)+F_X(x^-)\big)\,\big(F_Y(y)+F_Y(y^-)\big) = G_X(x)\,G_Y(y),
    $$
    see \eqref{eq:bivariate_mid-distribution_function}.     Analogously, $\p(X=x, Y=y) = \p(X=x)\, \p(Y=y)$ holds. Together with $\tau=\rho=0$, Lemma~\ref{lem:variance} implies
    $$
    k^{(\tau)}_1 (x,y) =  \big( 	2 G_{X} ( x ) - 1\big) \big( 	2 G_{Y} (y) - 1\big) = 4 \big( 	G_{X} (x) - 1/2\big) \big( 	G_{Y} (y) - 1/2\big).
     $$
     Furthermore, also using part~(iv) of Lemma \ref{lem:grade}, we get $\E[G_{X,Y} (x,Y')] = G_X(x)\,\E[G_{Y} (Y')] = \frac{1}{2}\, G_X(x)$ and thus
    $$
        k_1^{(\rho)}(x,y)= 4\big(G_X(x)G_Y(y)-\tfrac{1}{2}\,G_X(x)-\tfrac{1}{2}\,G_Y(y) + 1/4\big) = k^{(\tau)}_1 (x,y).
    $$
 These two results imply that
 \begin{align*}
 \E \big[ k^{(\tau)}_1 (X_i,Y_i)^2  \big] =& \E \big[ k^{(\rho)}_1 (X_i,Y_i)^2  \big] = 16\, \Var\big(G_X(X)\big)\, \Var\big(G_Y(Y)\big)\\ 
 =& \frac{16}{144}  \big(1 - \p(X=\copyX=\copytwoX)\big) \big(1 - \p(Y=\copyY=\copytwoY)\big),
 \end{align*}
 where again Lemma \ref{lem:grade}\,(iv) was used for the last equality.
\end{proof}

\begin{proof}[Proof of Corollary \ref{cor:iid_and_independent_processes_asymptotic_distribution_tau}]   
    Combining Lemma \ref{lem:kernels_under_independence} with Proposition \ref{prop:asymptotic_distribution_tau_iid} yields the claim.
\end{proof}

\begin{proof}[Proof of Corollary \ref{cor:iid_and_independent_processes_asymptotic_distribution_gamma}]   
Substituting $\gamma=\tau=\rho=0$ into Proposition \ref{prop:asymptotic_distribution_gamma_iid}, we get
\begin{align*}
\sigma_{\gamma,\textup{iid},\textup{ind}}^2 =&  \frac{\sigma_{\tau,\textup{iid},\textup{ind}}^2}{(1-\nu)^2}
\ =\ 
\frac{\frac{4}{9} (1-\zeta_3(X))(1-\zeta_3(Y))}{(1-\nu)^2},
\\
\sigma_{\tau_b, \textup{iid},\textup{ind}}^2 =& \frac{\sigma_{\tau, \textup{iid},\textup{ind}}^2}{\tau_X\tau_Y}
\ =\ 
\frac{\frac{4}{9} (1-\zeta_3(X))(1-\zeta_3(Y))}{(1-\zeta(X))(1-\zeta(Y))},
\\
\sigma_{\rho_b, \textup{iid},\textup{ind}}^2 =& \frac{\sigma_{\rho, \textup{iid},\textup{ind}}^2}{\rho_X\rho_Y}
\ =\ 
\frac{(1-\zeta_3(X))(1-\zeta_3(Y))}{(1-\zeta_3(X))(1-\zeta_3(Y))} = 1,
\end{align*}
where we plugged in the terms from Corollary \ref{cor:iid_and_independent_processes_asymptotic_distribution_tau} into the respective second equality. The final expressions in Corollary \ref{cor:iid_and_independent_processes_asymptotic_distribution_gamma} follow using \eqref{eq:nu_under_independence}.
\end{proof}
\begin{proof}[Proof of Lemma \ref{lem:short_memory}]
   By Theorem 3 in \cite{doukhan_mixing_1994}, equation (1.11) in \cite{bradley2005} and the fact that mixing coefficients can only decrease when applied to functions of some original time series, it holds that \begin{align*}\left|\Cov\left(k^{(l)}_1(X_i,Y_i),k^{(m)}_1(X_{i+h},Y_{i+h})\right)\right|\le 8\beta_h^{1-2/p}\norm{k^{(l)}_1(X_i,Y_i)}_p\norm{k^{(m)}_1(X_{i+h},Y_{i+h})}_p,
   \end{align*}
   where $p>2$ and $\norm{Z}_p=\sqrt[p]{\E[Z^p]}$. Since $k^{(l)}_1$ and $k^{(m)}_1$ are bounded, all moments exist and one can choose some $p>2$ that satisfies the summability condition in Assumption \ref{ass:time_series}. Absolute summability of the crosscovariances follows.
\end{proof}
\begin{proof}[Proof of Proposition \ref{prop:consistency_variance_time_series}]
Splitting the absolute difference between the covariance and its estimator into three terms, by adding 0 and applying the triangle inequality, yields:
\begin{align*} 
      &\frac{\left|  \widehat{\sigma}_{lm}  - \sigma_{lm} \right|}{r^{(l)}r^{(m)}}  \\
      &\leq  \left| \sum_{h=-n+1}^{n-1} w\left(\frac{h}{b_n+1}\right) \ \alpha_{k^{(l)}_1(X,Y),\, k^{(m)}_1(X,Y)}(h) - \sum_{h=-\infty}^{\infty} \alpha_{k^{(l)}_1(X,Y),\, k^{(m)}_1(X,Y)}(h) \right| \\
      & \quad + \left|  \widehat{\alpha}_{k^{(l)}_1(X,Y),\, k^{(m)}_1(X,Y)}(0) +  \sum_{h=1}^{n-1} w\left(\frac{h}{b_n+1}\right) \Bigl( \widehat{\alpha}_{k^{(l)}_1(X,Y),\, k^{(m)}_1(X,Y)}(h) + \widehat{\alpha}_{k^{(m)}_1(X,Y),\, k^{(l)}_1(X,Y)}(h)  \Bigr) \right. \\ 
      & \quad -  \left. \sum_{h=-n+1}^{n-1} w\left(\frac{h}{b_n+1}\right) \ \alpha_{k^{(l)}_1(X,Y),\, k^{(m)}_1(X,Y)}(h) \right|\\
      & \quad + \left|  \widehat{\alpha}_{\widehat{k}^{(l)}_1(X,Y),\, \widehat{k}^{(m)}_1(X,Y)}(0) +  \sum_{h=1}^{n-1} w\left(\frac{h}{b_n+1}\right) \Bigl( \widehat{\alpha}_{\widehat{k}^{(l)}_1(X,Y),\, \widehat{k}^{(m)}_1(X,Y)}(h) + \widehat{\alpha}_{\widehat{k}^{(m)}_1(X,Y),\, \widehat{k}^{(l)}_1(X,Y)}(h)  \Bigr) \right. \\ 
      & \quad -  \left. \widehat{\alpha}_{k^{(l)}_1(X,Y),\, k^{(m)}_1(X,Y)}(0) -  \sum_{h=1}^{n-1} w\left(\frac{h}{b_n+1}\right) \Bigl( \widehat{\alpha}_{k^{(l)}_1(X,Y),\, k^{(m)}_1(X,Y)}(h) + \widehat{\alpha}_{k^{(m)}_1(X,Y),\, k^{(l)}_1(X,Y)}(h)  \Bigr) \right|.
\end{align*}

We now show that all three terms converge to 0 in probability as $n \rightarrow \infty$.

The first summand contains no stochasticity, so we consider the deterministic limit and apply the dominated convergence theorem with respect to the counting measure. 
It holds that 
\begin{align*}
    \lim_{n \rightarrow \infty} \sum_{h=-n+1}^{n-1} w\left(\frac{h}{b_n+1}\right) \ \alpha_{k^{(l)}_1(X,Y),\, k^{(m)} _1(X,Y)}(h) &=  \sum_{h=-\infty}^{\infty} \lim_{n \rightarrow \infty} w\left(\frac{h}{b_n+1}\right) \ \alpha_{k^{(l)}_1(X,Y),\, k^{(m)} _1(X,Y)}(h)\\
    &= \sum_{h=-\infty}^{\infty}  \alpha_{k^{(l)}_1(X,Y),\, k^{(m)} _1(X,Y)}(h),
\end{align*}
where for the first equality, the dominated convergence theorem was used noting that the summands are dominated by the function $|\alpha_{k^{(l)}_1(X,Y),\, k^{(m)} _1(X,Y)}(h)|$, which is integrable with respect to the counting measure due to Lemma \ref{lem:short_memory}.  For the second equality it was used that $\lim_{n \rightarrow \infty} w(h/(b_n+1)) = 1$ by continuity of $w$, $w(0)=1$, and $\lim_{n\to \infty}b_n=\infty$.

Denoting the second summand by $S_n$ and writing $\widehat{\alpha}_{k^{(l)}_1(X,Y),\, k^{(m)}_1(X,Y)}(-h):=\widehat{\alpha}_{k^{(m)}_1(X,Y),\, k^{(l)}_1(X,Y)}(h)$, we calculate
\begin{align*}
    \E[S_n^2]=& \sum_{|h_1|\le b_n}\sum_{|h_2|\le b_n} w\left(\frac{h_1}{b_n+1}\right)w\left(\frac{h_2}{b_n+1}\right) \Cov\left(\widehat{\alpha}_{k^{(l)}_1(X,Y),\, k^{(m)}_1(X,Y)}(h_1) ,\widehat{\alpha}_{k^{(l)}_1(X,Y),\, k^{(m)}_1(X,Y)}(h_2)\right)\\
    & \le  \sum_{|h_1|\le b_n}\sum_{|h_2|\le b_n} \sqrt{\Var\left(\widehat{\alpha}_{k^{(l)}_1(X,Y),\, k^{(m)}_1(X,Y)}(h_1)\right)\Var\left(\widehat{\alpha}_{k^{(l)}_1(X,Y),\, k^{(m)}_1(X,Y)}(h_2)\right)}
\end{align*}
by the Cauchy-Schwarz inequality and $w(x)\le 1\ \forall x\in \mathbb{R}$. Since the empirical autocovariances are, after multiplication by $n/(n-h)$, also U-statistics (of order 1) and its underlying processes $k^{(l)}_1(X,Y)$ and $k^{(m)}_1(X,Y)$ (as well as their multiplication) inherit the mixing properties in Assumption \ref{ass:time_series}, Proposition \ref{prop:CLT_ustats_univariate} is also valid for them. Therefore, it holds that $$\Var\left(\widehat{\alpha}_{k^{(l)}_1(X,Y),\, k^{(m)}_1(X,Y)}(h)\right)=\mathcal{O}\left(\frac{1}{n}\right)\ \forall h\in \mathbb{Z}.$$
This, together with $b_n=o(\sqrt{n})$, delivers the limit $\lim_{n \to \infty}\E[S_n^2]=0$, which in turn implies the desired convergence in probability.

The third summand we denote by $T_n$. The convention $\widehat{\alpha}_{\widehat{k}^{(l)}_1(X,Y),\, \widehat{k}^{(m)}_1(X,Y)}(-h):=\widehat{\alpha}_{\widehat{k}^{(m)}_1(X,Y),\, \widehat{k}^{(l)}_1(X,Y)}(h)$ simplifies it to
\begin{align*}
     T_n&=\left|  \sum_{h=-n+1}^{n-1} w\left(\frac{h}{b_n+1}\right) \Bigl( \widehat{\alpha}_{\widehat{k}^{(l)}_1(X,Y),\, \widehat{k}^{(m)}_1(X,Y)}(h) - \widehat{\alpha}_{k^{(l)}_1(X,Y),\, k^{(m)}_1(X,Y)}(h)  \Bigr) \right|\\
     &=\left|  \sum_{|h|\le b_n} w\left(\frac{h}{b_n+1}\right) \Bigl( \widehat{\alpha}_{\widehat{k}^{(l)}_1(X,Y),\, \widehat{k}^{(m)}_1(X,Y)}(h) - \widehat{\alpha}_{k^{(l)}_1(X,Y),\, k^{(m)}_1(X,Y)}(h)  \Bigr) \right|\\
     &\le \sum_{|h|\le b_n}\left| \widehat{\alpha}_{\widehat{k}^{(l)}_1(X,Y),\, \widehat{k}^{(m)}_1(X,Y)}(h) - \widehat{\alpha}_{k^{(l)}_1(X,Y),\, k^{(m)}_1(X,Y)}(h) \right|
\end{align*}
and the last inequality is due to $w(x)\le 1\ \forall x\in \mathbb{R}$ and the triangle inequality. Inspecting the summands further, we obtain
\begin{align}\label{eq:twostep_convergencespeed}
&\ \left| \widehat{\alpha}_{\widehat{k}^{(l)}_1(X,Y),\, \widehat{k}^{(m)}_1(X,Y)}(h) - \widehat{\alpha}_{k^{(l)}_1(X,Y),\, k^{(m)}_1(X,Y)}(h) \right|\nonumber\\=&\ \left| \frac{1}{n}\sum_{i=1}^{n-h}\widehat{k}^{(l)}_1(X_i,Y_i)\widehat{k}^{(m)}_1(X_{i+h},Y_{i+h}) - k^{(l)}_1(X_i,Y_i)k^{(m)}_1(X_{i+h},Y_{i+h}) \right|\nonumber\\
 \leq&\ \sup_{x,y} \left| \widehat{k}_{1}^{(l)} (x,y) \widehat{k}_{1}^{(m)} (x,y) -  k_{1}^{(l)} (x,y) k_{1}^{(m)} (x,y) \right|=\mathcal{O}_p(1/\sqrt{n}).
\end{align}
The speed of convergence follows from the speed in which the empirical probability measure converges to the theoretical probability measure, also under Assumption \ref{ass:time_series} \citep{doukhan1995invariance}. Indexing 
\begin{align*}
    \frac{1}{\sqrt{n}}\sum_{i=1}^n\left(\delta_{(X_i,Y_i)}-\mathbb{P}_{X,Y}\right),
\end{align*}
where $\delta_{(X_i,Y_i)}$ is the two-dimensional Dirac measure and $\mathbb{P}_{X,Y}$ is the theoretical probability measure of the vector $(X_1,Y_1)$ by the class of indicator functions $\mathds{1}_{(-\infty,x)\times (-\infty,\infty)}$,
$\mathds{1}_{(-\infty,x]\times (-\infty,\infty)}$,
$\mathds{1}_{(-\infty,\infty)\times (-\infty,y)}$,
$\mathds{1}_{(-\infty,\infty)\times (-\infty,y]}$,
$\mathds{1}_{(-\infty,x)\times (-\infty,y)}$,
$\mathds{1}_{(-\infty,x)\times (-\infty,y]}$,
$\mathds{1}_{(-\infty,x]\times (-\infty,y)}$ and
$\mathds{1}_{(-\infty,x]\times (-\infty,y]}$, $x,y\in \mathbb{R}$
yields $\sqrt{n}$ convergence for any continuous function of (bivariate) CDFs and their respective left limits by the continuous mapping theorem (for convergence in distribution). $\widehat{\tau}, \widehat{\tau}_X, \widehat{\tau}_Y, \widehat{\rho}_X, \widehat{\rho}_Y$ and $\widehat{\nu}$ are all also continuous functions of (bivariate) CDFs or their left limits and the claim follows.

The covariances involving $\rho$ need a special treatment due to the additional layer of estimation described in Subsection \ref{subsec:Variance Estimation IID}. We add and subtract a term involving
\begin{equation*}
    \widetilde{k}^{(\rho)}_1 (x,y) 
	:= 4(\widetilde{g}_X(x)+\widetilde{g}_Y(y)+G_X(x)G_Y(y)-G_X(x)-G_Y(y))+1- \rho,
\end{equation*}
where $\widetilde{g}_X (x) = \frac 1 n \sum_{i=1}^n G_{X,Y} (x,Y_i)$ and $\widetilde{g}_Y(y)$ is defined accordingly. By the triangle inequality, it holds that
\begin{align}\label{eq:threestep_convergencespeed}
&\ \left| \frac{1}{n}\sum_{i=1}^{n-h}\widehat{k}^{(l)}_1(X_i,Y_i)\widehat{k}^{(m)}_1(X_{i+h},Y_{i+h}) - k^{(l)}_1(X_i,Y_i)k^{(m)}_1(X_{i+h},Y_{i+h}) \right|\nonumber\\
=&\ \left| \frac{1}{n}\sum_{i=1}^{n-h}\widehat{k}^{(l)}_1(X_i,Y_i)\widehat{k}^{(m)}_1(X_{i+h},Y_{i+h})+\widetilde{k}^{(l)}_1(X_i,Y_i)\widetilde{k}^{(m)}_1(X_{i+h},Y_{i+h})\right.\nonumber\\
&\ - \left. \frac{1}{n}\sum_{i=1}^{n-h}\widetilde{k}^{(l)}_1(X_i,Y_i)\widetilde{k}^{(m)}_1(X_{i+h},Y_{i+h}) + k^{(l)}_1(X_i,Y_i)k^{(m)}_1(X_{i+h},Y_{i+h}) \right|\nonumber\\
\le &\ \left| \frac{1}{n}\sum_{i=1}^{n-h}\widehat{k}^{(l)}_1(X_i,Y_i)\widehat{k}^{(m)}_1(X_{i+h},Y_{i+h})-\widetilde{k}^{(l)}_1(X_i,Y_i)\widetilde{k}^{(m)}_1(X_{i+h},Y_{i+h})\right|\nonumber\\
&\ + \left| \frac{1}{n}\sum_{i=1}^{n-h}\widetilde{k}^{(l)}_1(X_i,Y_i)\widetilde{k}^{(m)}_1(X_{i+h},Y_{i+h})-k^{(l)}_1(X_i,Y_i)k^{(m)}_1(X_{i+h},Y_{i+h})  \right|\nonumber\\
 \leq&\ \sup_{x,y} \left| \widehat{k}_{1}^{(l)} (x,y) \widehat{k}_{1}^{(m)} (x,y) -  \widetilde{k}_{1}^{(l)} (x,y) \widetilde{k}_{1}^{(m)} (x,y) \right|+\sup_{x,y} \left| \widetilde{k}_{1}^{(l)} (x,y) \widetilde{k}_{1}^{(m)} (x,y) -  k_{1}^{(l)} (x,y) k_{1}^{(\rho)} (x,y) \right|\nonumber\\
 =&\ \mathcal{O}_p(1/\sqrt{n}),
\end{align}
where either $l=\rho$ or $m=\rho$ or both. The first term is $\mathcal{O}_p(1/\sqrt{n})$ because of the same arguments as for (\ref{eq:twostep_convergencespeed}). The second term is also $\mathcal{O}_p(1/\sqrt{n})$ because the function class $\{y\mapsto G_{X,Y}(x,y), x\in \mathbb{R}\}$ satisfies the conditions in \cite{doukhan1995invariance} and the empirical probability measure converges with $\sqrt{n}$ rate.

The claim follows from (\ref{eq:twostep_convergencespeed}) and (\ref{eq:threestep_convergencespeed}) because we have assumed $b_n=o(\sqrt{n})$.



\end{proof}

\begin{lemma} \label{lem:autocovariances_under_independent_processes}
    Let $\{X_i\}_{i \in \mathbb{Z}}$ and $\{Y_i\}_{i \in \mathbb{Z}}$ be independent processes satisfying Assumption \ref{ass:time_series}. Then it holds that
    \begin{align*}
        \E \big[ k^{(\tau)}_1 (X_i,Y_i)\, k^{(\tau)}_1(X_{i+h},Y_{i+h}) \big] &= \E \big[ k^{(\rho)}_1 (X_i,Y_i)\, k^{(\rho)}_1(X_{i+h},Y_{i+h}) \big] = \frac{1}{9}\, \rho_X(h)\, \rho_Y(h) \\ 
&= \frac{1}{9}\, \big[ 1 - \zeta_2(X)\big] \big[ 1 - \zeta_2(Y)\big]\, \rho_{b,X}(h)\, \rho_{b,Y}(h),
    \end{align*}
    where we denote by
    $\rho_X(h):= \rho(X_i,X_{i+h})$ and $\rho_{b,X}(h):= \rho_b(X_i,X_{i+h})$ the Spearman and grade autocorrelation functions at lag $h$.
\end{lemma}

\begin{proof}
 By Lemma \ref{lem:kernels_under_independence} and Definition \ref{def_rho}, we have
	\begin{align*}
		  &\E \big[ k^{(\tau)}_1 (X_i,Y_i)\, k^{(\tau)}_1(X_{i+h},Y_{i+h}) \big] = \E \big[ k^{(\rho)}_1 (X_i,Y_i)\, k^{(\rho)}_1(X_{i+h},Y_{i+h}) \big] \\
		&= 16 \E \big[ \big( 	G_{X} ( X_i ) - 1/2\big) \big( 	G_{Y} (Y_i) - 1/2\big) \big( 	G_{X} (X_{i+h}) - 1/2\big) \big( 	G_{Y} (Y_{i+h}) - 1/2\big) \big]\\
		&= 16 \E \big[ \big( 	G_{X} ( X_i ) - 1/2\big) \big( 	G_{X} (X_{i+h}) - 1/2\big)\big]\, \E\big[ \big( 	G_{Y} (Y_i) - 1/2\big)  \big( 	G_{Y} (Y_{i+h}) - 1/2\big) \big]\\
		&= \frac{16}{144}\, \rho_X(h)\, \rho_Y(h) = \frac{1}{9}\, \big[ 1 - \zeta_2(X)\big] \big[ 1 - \zeta_2(Y)\big]\, \rho_{b,X}(h)\, \rho_{b,Y}(h).
	\end{align*}
\end{proof}

\begin{proof}[Proof of Corollary \ref{cor:independent_processes_asymptotic_distribution_tau}]

The corollary follows from Proposition \ref{prop:asymptotic_distribution_tau} and Lemma \ref{lem:autocovariances_under_independent_processes} with analogous arguments as in the proof of Corollary \ref{cor:iid_and_independent_processes_asymptotic_distribution_gamma}.
\end{proof}

\clearpage

\section{Cross-dependence of Bivariate Geometric Distibution}
\label{Cross-dependence of Bivariate Geometric Distibution}
In order to illustrate our dependence concepts and the associated computations in the case of discrete random variables, let us consider the three-parameter bivariate geometric distribution for $(X,Y)$ as proposed by \citet{barbiero19}, which we denote as $\bgeom(\pi_X, \pi_Y, \alpha)$ with parameters $\pi_X,\pi_Y\in (0,1)$ and $\alpha\in [-1,1]$. The BGeom~distribution is constructed by applying the Farlie--Gumbel--Morgenstern (FGM) copula \citep[see][p.~77]{nelsen06} to the univariate geometric distributions $X\sim\geom(\pi_X)$ and $Y\sim\geom(\pi_Y)$, where~$\alpha$ controls the extent of cross-dependence (with independence for $\alpha=0$). 
Note that the FGM copula can only model relatively weak dependence \citep[see again][p.~77]{nelsen06}, and in particular does not include the copulas representing perfect positive and negative dependence (the co- and countermonotonicity copulas) as limiting cases. Thus, this family of copulas and with it the BGeom~distribution are of limited use for modeling. However, they are perfectly suited to our purpose since they allow for a closed-form computation of the formulas we want to illustrate. Note that this is also the reason why the dependence measures calculated below can only take values close to zero even in the continuous case and if $\alpha=\pm1$. The CDF and PMF of the BGeom~distribution are given by
\begin{align*}
F_{X,Y}(x,y) =&\ \big(1-(1-\pi_X)^{x+1}\big)\, \big(1-(1-\pi_Y)^{y+1}\big)\, \big(1 + \alpha  (1-\pi_X)^{x+1} (1-\pi_Y)^{y+1}\big), \\
\p(X=x, Y=y) =&\ \pi_X (1-\pi_X)^x\, \pi_Y (1-\pi_Y)^y\\
&\ \times \Big(1+\alpha\,  \big((2-\pi_X) (1-\pi_X)^x-1\big)\, \big((2-\pi_Y) (1-\pi_Y)^y-1\big)\Big),
\end{align*}
where the first factors are the respective univariate CDFs and PMFs, respectively.

Let us start with the univariate $\geom(\pi_X)$ distribution, corresponding to the marginal distribution of~$X$. Its MDF \eqref{eq:grade_and_spi} is computed as
$$
G_X(x) = \frac{1}{2} \big(F_X(x)+F_X(x-1)\big)\ =\ \frac{1}{2} \big(2-(2-\pi_X) (1-\pi_X)^x\big),
$$
and its tie probabilities \eqref{eq:tie_probability_formula} follow as
\begin{align*}
\zeta(X)=&
 \sum_{x=0}^\infty P(X=x)^2
 = \pi_X^2\sum_{x=0}^\infty (1-\pi_X)^{2x}
 = \frac{\pi_X}{2-\pi_X}\ \in (0,1),
 \\
\zeta_2(X)=& \sum_{x=0}^\infty P(X=x)^3
 = \pi_X^3\sum_{x=0}^\infty (1-\pi_X)^{3x}
 = \frac{\pi_X^2}{3-(3-\pi_X)\pi_X}\ \in (0,1).
\end{align*}
Analogous calculations hold for $Y\sim \geom(\pi_Y)$.
The bivariate counterparts for $(X,Y)\sim\bgeom(\pi_X, \pi_Y, \alpha)$ are more complex. After tedious algebra, we get
\begin{align*}
G_{X,Y}(x,y) =\ G_X(x)\, G_Y(y) + \alpha\, (1-\pi_X)^x (1-\pi_Y)^y\, & \left(\frac{\pi_X}{2}\,  F_X(x-1)+(1-\pi_X)\, G_X(x)\right)\\
& \times \left(\frac{\pi_Y}{2}\,  F_Y(y-1)+(1-\pi_Y)\, G_Y(y)\right)
\end{align*}
for the MDF \eqref{eq:bivariate_mid-distribution_function}, as well as
\begin{align*}
\nu(X,Y) =&\ 1-\big(1-\zeta(X)\big) \big(1-\zeta(Y)\big) \\
&\, \times\left(1+\frac{\alpha\, \zeta_2(X)\zeta_2(Y)}{\pi_X \pi_Y} \left(\frac{1}{2} + \frac{\alpha\, (1-\pi_X) (1-\pi_Y)}{\big(1+(1-\pi_X)^2\big) \big(1+(1-\pi_Y)^2\big)}\right)\right)
\end{align*}
for the tie probability \eqref{eq:nu}. Note that the factor in the second row, which is a quadratic polynomial in~$\alpha$, hardly deviates from one, because the coefficients of $\alpha$ and $\alpha^2$ usually remain very close to zero. Hence, $1-\nu(X,Y) \approx \big(1-\zeta(X)\big) \big(1-\zeta(Y)\big)$ holds in close approximation.

If the FGM-copula with dependence parameter $\alpha\in [-1,1]$ is applied to continuous $X,Y$, then it is known that $\tau(X,Y)=\frac{2}{9}\,\alpha$ and $\rho(X,Y)=\frac{3}{2}\,\tau(X,Y)=\frac{1}{3}\,\alpha$, see \citet[pp.~162, 174]{nelsen06}. In the discrete case $\bgeom(\pi_X, \pi_Y, \alpha)$, however, we obtain another result, namely
$$
\tau(X,Y) =\ \frac{2 \alpha\,  (1-\pi_X) (1-\pi_Y)}{\big(3-(3-\pi_X) \pi_X\big) \big(3-(3-\pi_Y) \pi_Y\big)},
\qquad
\rho(X,Y)=\frac{3}{2}\,\tau(X,Y).
$$
Note that $|\tau(X,Y)|<\frac{2}{9}\,|\alpha|$ unless both $\pi_X,\pi_Y= 0$. By contrast, for $\pi_X\to 1$ or $\pi_Y\to 1$, both~$\tau$ and~$\rho$ converge to zero.

Combining the above results, we also get closed-form expressions for $\gamma(X,Y)$, $\tau_b(X,Y)$, and $\rho_b(X,Y)$. While $\gamma$ gets rather complex due to the complex expression for $\nu(X,Y)$ above, expressions for~$\tau_b$ and~$\rho_b$ are concise by using that
$$
1-\zeta(X) = \frac{2 (1-\pi_X)}{2-\pi_X}
\quad\text{and}\quad
1-\zeta_2(X) = \frac{3 (1-\pi_X)}{3-(3-\pi_X)\pi_X}.
$$
A simple approximate expression for~$\gamma$ follows when using $1-\nu(X,Y) \approx \big(1-\zeta(X)\big) \big(1-\zeta(X)\big)$. While still $|\rho_b|<\frac{1}{3}\,|\alpha|$ holds with $\rho_b\to 0$ if $\pi_X\to 1$ or $\pi_Y\to 1$, $|\tau_b|$ and $|\gamma|$ may exceed $\frac{2}{9}\,|\alpha|$.

\clearpage
\section{Simulation Study} \label{app_sec:simulations}


\subsection{IID Data}

\label{subsec:sims_iid}

\subsubsection{Considered DGPs}

We simulate from the following continuous iid DGPs: 
\begin{align}
    (X_i, Y_i):=\big(X_i,\ \alpha\, X_i + \sqrt{1-\alpha^2}\,u_i\big) \text{ with }& X_i, u_i\overset{\textup{iid}}{\sim} \mathcal{N}(0,1),\label{eq:normaliidDGP}\\
    (X_i, Y_i):=\big(X_i,\ \alpha\, X_i + \sqrt{1-\alpha^2}\,u_i\big) \text{ with }& X_i, u_i\overset{\textup{iid}}{\sim} t(4),\label{eq:t4iidDGP}\\
    (X_i, Y_i):=\big(X_i,\ \alpha\, X_i + \sqrt{1-\alpha^2}\,u_i\big) \text{ with }& X_i, u_i\overset{\textup{iid}}{\sim} t(1)\label{eq:t1iidDGP},\\
    (X_i, Y_i):=\big(X_i,\ F^{-1}_{\textup{Exp}}[\Phi(\alpha\, X_i + \sqrt{1-\alpha^2}\,u_i)]\big) \text{ with }& X_i, u_i\overset{\textup{iid}}{\sim} \mathcal{N}(0,1),\label{eq:normalexpiidDGP}
\end{align}
where $\alpha\in (-1,1)$, $\Phi$ denotes the CDF of the standard normal distribution and $F^{-1}_{\textup{Exp}}$ the quantile function of the standard exponential distribution. While \eqref{eq:normaliidDGP} just leads to bivariately normally distributed pairs, \eqref{eq:t4iidDGP}--\eqref{eq:normalexpiidDGP} are more demanding. DGP \eqref{eq:t4iidDGP} leads to heavy tails but with existing first- and second-order moments. In situation \eqref{eq:t1iidDGP}, we do not have existing moments, while DGP \eqref{eq:normalexpiidDGP} is characterized by strong positive skewness in its second component.

As discrete iid DGPs, we use
\begin{align}
    (X_i, Y_i):=& \big(X_i,\ \alpha \circ X_i + (1-\alpha) \circ u_i\big) \text{ with } X_i, u_i \overset{\textup{iid}}{\sim}\poi(1), \alpha \in (0,1)\label{eq:Pois1iidDGP}\\
    (X_i, Y_i):=& \big(X_i,\ \alpha \circ X_i + (1-\alpha) \circ u_i\big) \text{ with } X_i+1, u_i+1 \overset{\textup{iid}}{\sim}\zipf(1), \alpha \in (0,1)\label{eq:Zipf1iidDGP}\\
    (X_i, Y_i):=& \big(X_i,\ \alpha \odot X_i + (\mathds{1}_{\{\alpha\ge 0\}}-\mathds{1}_{\{\alpha< 0\}})  \cdot (1-|\alpha|) \odot u_i\big) \text{ with } X_i, u_i \overset{\textup{iid}}{\sim}\skellam(1,1), \alpha \in (-1,1),\label{eq:Ske11amiidDGP}
\end{align}
where $\circ$ denotes the binomial thinning operator \citep{steutel79} and $\odot$ the signed binomial thinning operator \citep{kim08}. These thinning operators serve as integer-valued substitutes of multiplication \citep{weiss2008c,weiss18}. While the Poisson DGP \eqref{eq:Pois1iidDGP} can be interpreted as the ``normal case'' in the count-data world, the Zipf DGP \eqref{eq:Zipf1iidDGP} causes heavy tails and non-existing moments \citep[Section~11.2.20]{johnson05}, whereas the Skellam DGP \eqref{eq:Ske11amiidDGP} has $\bbz$-valued outcomes \citep[Section~4.12.3]{johnson05} and allows to capture negative cross-dependence.

\subsubsection{Analysis of Simulation Results}

The obtained simulation results are tabulated in Appendix~\ref{Tabulated Simulation results from Section sec:simulations}.
Tables~\ref{tabTestsIIDcont}--\ref{tabTestsIIDdiscrH1} provide the simulated rejection rates for tests of the null hypothesis of component-wise independence on the 10\% significance level based on various rank correlations and using the variance formulas under independence, while Tables~\ref{tabConfIIDcont} and \ref{tabConfIIDdiscr} summarize the empirical coverage of 90\% confidence intervals for rank correlations. Note that instead of reporting the values of the parameter $\alpha$ in the tables, we simulate and report the respective (approximate) values of the rank correlations for better interpretability (for the independence tests, we simply report $\tau$ for the continuous DGPs and $\gamma$ for the discrete DGPs). 
For the independence tests, we also consider Pearson correlation, $r$, as a competitor. Note that under all iid DGPs -- except for the heavy-tailed \eqref{eq:t1iidDGP} and \eqref{eq:Zipf1iidDGP} with non-existing moments -- and for $\alpha = 0$ (or more generally the null hypothesis of independence), it holds that 
\begin{equation} \label{eq:Pearson_as_distribution}
    \sqrt{n}\,\widehat{r}\stackrel{d}{\rightarrow} \mathcal{N}(0,1) \text{ as } n \rightarrow \infty,
\end{equation}
where $\widehat{r}$ denotes the empirical Pearson correlation \citep[pp.~125--126]{serfling80}.


Let us begin our discussion with the simulated rejection rates, i.e.\ with simulated sizes if the null is true, and simulated power values otherwise. For the continuous DGPs, we compare $\tau$ and $\rho$ to Pearson correlation, while for the discrete DGPs, we add the generalized rank correlations.

Table~\ref{tabTestsIIDcont} for the continuous iid DGPs shows that the Pearson correlation~$r$ holds the nominal 10\%-level very well under normally distributed data but shows size distortions otherwise (column $\tau=0$). Especially for the Cauchy-DGP \eqref{eq:t1iidDGP}, we observe strong undersizing with further deterioration for increasing sample size~$n$. Kendall's~$\tau$ from Definitions~\ref{def_tau} and~\ref{def:empirical_tau} and Spearman's~$\rho$ from Definitions~\ref{def_rho} and~\ref{def:empirical_rho}, by contrast, have close to nominal and stable sizes across the different DGPs \eqref{eq:normaliidDGP}--\eqref{eq:normalexpiidDGP}, with the exception that especially~$\tau$ shows mild oversizing for the smallest sample size $n=50$. Looking at the power values (columns with $\tau\not=0$), all tests are similarly powerful for DGPs \eqref{eq:normaliidDGP}--\eqref{eq:t4iidDGP}, whereas the Pearson correlation~$r$ clearly loses power for the ``more demanding'' DGPs \eqref{eq:t1iidDGP} with non-existing moments and \eqref{eq:normalexpiidDGP} with highly skewed $Y$-component. $\tau$ and~$\rho$, by contrast, are again robust against such distributional variations. 

For the discrete iid DGPs in  Tables~\ref{tabTestsIIDdiscrH0} (sizes) and~\ref{tabTestsIIDdiscrH1} (power values), we recognize an analogous behavior of~$r$ as before, namely, strong distortions of size and power for the heavy-tailed DGP \eqref{eq:Zipf1iidDGP}. This is also visible in Figure \ref{fig:iid_rejrates} in the appendix, which shows the simulated rejection rates for all possible strengths of dependence for three selected DGPs. Regarding the rank correlations, we now also considered Goodman-Kruskal's $\gamma$ from Definitions~\ref{def:gamma} and~\ref{def:empirical_gamma}, Kendall's $\tau_b$ from Definitions~\ref{def:taub} and~\ref{def:empirical_tau_b}, the grade correlation $\rho_b$ from Definitions~\ref{def:grade_correlation} and~\ref{def:empirical_grade_correlation}, and the modified $\tau_{b,\textup{mod}}$ from \eqref{eq:tau_b_mod}. However, due to using plug-in estimators, the decision rules for $\rho$ and~$\rho_b$ become equivalent, so we do not include~$\rho$ in  Tables~\ref{tabTestsIIDdiscrH0}--\ref{tabTestsIIDdiscrH1} to avoid redundancy. Also $\tau_{b,\textup{mod}}$ led to virtually identical rejection rates as $\tau_b$, hence also $\tau_{b,\textup{mod}}$ is omitted. Comparing the sizes of $\tau$, $\gamma$, and $\rho_b$ in  Table~\ref{tabTestsIIDdiscrH0}, we recognize the best agreement to the nominal 10\%-level for $\gamma$, and the worst (again) for~$\tau$. However, in any case, notable size distortions only occur for the smallest sample size $n=50$. We arrive at similar conclusions for the power values in  Table~\ref{tabTestsIIDdiscrH1}, where differences in power are only observed for $n=50$. At this point, it should be noted that we simulated the Skellam DGP \eqref{eq:Ske11amiidDGP} also for negative dependence values, but since we observed the same power as for positive dependence values, we again omitted these values for the sake of readability of Table~\ref{tabTestsIIDdiscrH1}.

Next, we investigate the performance of 90\% confidence intervals (CIs) for the different DGPs, where we use our proposed variance estimator from \eqref{eq:variance_estimator_iid}. While we considered all rank correlations for the task of testing for dependence in our simulation study, we restrict our analysis of CIs to such rank correlations that do not suffer (much) from attainability problems. For continuous DGPs, we simply use Kendall's~$\tau$ from Definitions~\ref{def_tau} and~\ref{def:empirical_tau} and Spearman's~$\rho$ from Definitions~\ref{def_rho} and~\ref{def:empirical_rho} again. For discrete DGPs, by contrast, recall our discussion in Sections~\ref{Introduction} and \ref{sec:rank_correlations}, it is known that only Goodman-Kruskal's $\gamma$ from Definitions~\ref{def:gamma} and~\ref{def:empirical_gamma} avoids any attainability problems \citep{pohle2025}, and the grade correlation $\rho_b$ from Definitions~\ref{def:grade_correlation} and~\ref{def:empirical_grade_correlation} at least mitigates them \citep{neslehova2007}. 

We report the simulated coverages for various values of the population rank correlations (namely for $0,0.4,0.8$ and also $-0.8,-0.4$ if the DGPs allow for negative dependence). The simulated coverages in Table~\ref{tabConfIIDcont} refer to continuous iid DGPs and the ones in Table~\ref{tabConfIIDdiscr} to discrete iid DGPs, where the left parts of the tables contain values for confidence intervals directly using the asymptotic normal approximation, see \eqref{eq:confidence_interval}, and the right parts describe the results for confidence intervals employing the Fisher transformation explained around \eqref{eq:Fisher_transform}. The observed coverages are in general close to the desired coverage of $90\%$. For the smaller sample sizes and a strong absolute extent~0.8 of cross dependence, the standard confidence intervals show in some cases a little stronger deviation from the desired coverage, especially for the more ``demanding'' continuous DGPs \eqref{eq:t1iidDGP} and \eqref{eq:normalexpiidDGP} as well as for all discrete DGPs, which improves when using the Fisher transformation. 
Hence, we generally recommend using the Fisher transformation for CI~computation and will do so in the empirical applications in Section \ref{Illustrative Data Applications} as well.

To conclude, the proposed confidence intervals and independence tests work very well for all examined rank correlations in the iid case, showing good coverage or size and power, respectively. For testing independence, Goodman-Kruskal's $\gamma$ from Definitions~\ref{def:gamma} and~\ref{def:empirical_gamma} shows a slightly better performance than the others if testing for independence in discrete DGPs, whereas Spearman's $\rho$ from Definitions~\ref{def_rho} and~\ref{def:empirical_rho} does so for continuous DGPs. Furthermore, we advise against using Pearson correlation for independence testing, because it does not perform notably better than the rank correlations even for normally distributed data, but often shows a much worse performance otherwise. 


\subsection{Time Series}
\label{subsec:sims_time_series}

\subsubsection{Considered DGPs}

In the time series case, we simulate the continuous first-order autoregressive (AR$(1)$) processes $(X_i, Y_i):=\big(X_i,\ \alpha\, X_i + \sqrt{1-\alpha^2}\,u_i\big)$ with $\alpha \in (-1, 1)$ and
\begin{align}
    X_i=0.8X_{i-1}+\varepsilon_i,\ \varepsilon_i\sim \mathcal{N}(0,1) \text{ and }& u_i=0.8u_{i-1}+\nu_i,\ \nu_i\sim \mathcal{N}(0,1),\label{eq:normalTSDGP}\\
    X_i=0.8X_{i-1}+\varepsilon_i,\ \varepsilon_i\sim t(4) \text{ and }& u_i=0.8u_{i-1}+\nu_i,\ \nu_i\sim t(4),\label{eq:t4TSDGP} \\
    X_i=0.8X_{i-1}+\varepsilon_i,\ \varepsilon_i\sim t(1) \text{ and }& u_i=0.8u_{i-1}+\nu_i,\ \nu_i\sim t(1)\label{eq:t1TSDGP}
\end{align}
as counterparts to \eqref{eq:normaliidDGP}--\eqref{eq:t1iidDGP}, and a bivariate version of the TEAR$(1)$ (transposed exponential AR) process of \citet{lawrance81}
\begin{align}
\begin{aligned}
    (X_i, Y_i):= \big(X_i,\ B_i^{(\alpha)} X_i + (1 - B_i^{(\alpha)})\,u_i\big) \text{ with } X_i=B_i^{(0.8)}\, X_{i-1}+0.2\, \varepsilon_i,\ \varepsilon_i\sim\expon(1) \\  \text{ and } u_i=B_i^{(0.8)}\, u_{i-1}+0.2\, \nu_i,\  \nu_i\sim\expon(1) \text{ and }  B_i^{(i)}\overset{\textup{iid}}{\sim} \textup{Be}(i) \text{ for }i\in \{\alpha, 0.8\},\ \alpha \in (0,1),\label{eq:TEAR1DGP}
\end{aligned}
\end{align}
where the stationary marginal distributions for $X_i, u_i$ and $Y_i$ are all $\expon(1)$ as in \eqref{eq:normalexpiidDGP}. Note that the component processes in \eqref{eq:normalTSDGP}, \eqref{eq:t4TSDGP} and \eqref{eq:TEAR1DGP} uniquely have the ACF $r(h) = 0.8^h$, but the TEAR$(1)$ sample paths are characterized by long-lasting rises followed by abrupt falls.

As discrete AR$(1)$-like processes, see \citet{weiss2008c,weiss18} for surveys, we use $(X_i, Y_i):= \big(X_i,\ \alpha \circ X_i + (1-\alpha)\circ u_i\big)$, $\alpha \in (0,1)$, together with the INAR$(1)$ (integer-valued AR) processes \citep{kenzie85}
\begin{align}
        X_i=0.8\circ X_{i-1}+\varepsilon_i,\ \varepsilon_i\sim \poi(0.2) \text{ and }& u_i=0.8 \circ u_{i-1}+\nu_i,\ \nu_i\sim \poi(0.2),\label{eq:Pois1TSDGP}\\
        X_i=0.8\circ X_{i-1}+\varepsilon_i,\ \varepsilon_i\sim \zipf(1.5) \text{ and }& u_i=0.8 \circ u_{i-1}+\nu_i,\ \nu_i\sim \zipf(1.5).\label{eq:ZipfTSDGP}
\end{align}
For DGP \eqref{eq:Pois1TSDGP}, $X_i$ and $u_i$ have the stationary marginal distribution $\poi(1)$ like in \eqref{eq:Pois1iidDGP}. For DGP \eqref{eq:ZipfTSDGP}, we include a pre-run of length 1,000 because the stationary marginal distribution of $X_i$ and $u_i$ is unknown, but is again characterized by heavy tails as in \eqref{eq:Zipf1iidDGP}.

In order to allow for negative dependence, we again replace the binomial thinning operator by its signed counterpart $(X_i, Y_i):= \big(X_i,\ \alpha \odot X_i + (\mathds{1}_{\{\alpha\ge 0\}}-\mathds{1}_{\{\alpha< 0\}})  \cdot (1-|\alpha|) \odot u_i\big)$, $\alpha \in (-1,1)$, with signed INAR$(1)$ (INARS(1)) processes \citep{kim08}
\begin{align}
        X_i=0.8\odot X_{i-1}+\varepsilon_i,\ \varepsilon_i\sim \skellam(0.2, 0.2) \text{ and }& u_i=0.8 \odot u_{i-1}+\nu_i,\ \nu_i\sim \skellam(0.2,0.2),\label{eq:SkellamTSDGP}
\end{align}
again with a pre-run of length 1,000.

\subsubsection{Analysis of Simulation Results}

Next, let us turn to the corresponding time series results for the tests of cross-dependence in Tables~\ref{tabTestsAR1cont}--\ref{tabTestsAR1discrH1}, where we employ the variance formulas under the null of cross-independence estimated by a HAC-type estimator as discussed in Section \ref{Variance Formulas under Independence TS}. Here, we again employ Pearson correlation as a competitor. Note that for linear processes as \eqref{eq:normalTSDGP} and \eqref{eq:t4TSDGP}, it holds under $\alpha = 0$ (or more generally cross-independence) that 
\begin{equation} \label{eq:Pearson_as_distribution_TS}
    \sqrt{n}\,\widehat{r}\stackrel{d}{\rightarrow} \mathcal{N}(0,\sigma_{r, \textup{ind}}^2) \text{ as } n \rightarrow \infty
\end{equation} 
with $\sigma_{r, \textup{ind}}^2=\sum_{h=-\infty}^{\infty}r_{X}(h)r_{Y}(h)$, where $r_{X}(h)$ denotes the autocorrelation function (ACF) of $X$ at lag $h$ 
\citep{haugh1976}. We employ this limiting distribution and estimate the variance $\sigma_{r, \textup{ind}}^2$ with a HAC-type estimator with the same kernel and bandwidth choice as for our other HAC-type estimators.

Recall that we use the rather large AR$(1)$-value~0.8, so we are concerned with a rather strong temporal dependence here. 
This explains why, given the well-known difficulties with estimating asymptotic variances in the time series case discussed in Section \ref{subsec:VarianceEstimationTS}, we note severely increased sizes throughout if comparing Table~\ref{tabTestsAR1cont} to Table~\ref{tabTestsIIDcont} (continuous DGPs). 
While for increasing~$n$, the sizes converge towards the nominal 10\%-level, we are still concerned with notable oversizing even for $n=800$. 
Comparing the different dependence measures, we again note similar size and power except for the ``demanding'' DGPs \eqref{eq:t1TSDGP} and \eqref{eq:TEAR1DGP}, where the Pearson correlation~$r$ leads to a rather poor performance. These properties occur in an analogous 
manner in the discrete time series case, see Tables~\ref{tabTestsAR1discrH0}--\ref{tabTestsAR1discrH1}. 
Again, an illustrative graph (Figure \ref{fig:TS_rejrates} in the appendix) shows simulated power curves for selected DGPs.

In Tables~\ref{tabConfAR1cont} and~\ref{tabConfAR1discr}, we report the coverage performance of the 90\%-CIs for the discrete and continuous time series DGPs using the HAC-type estimator from \eqref{eq:variance_estimation_time_series}. Expectedly, again due to the difficulty of variance estimation under temporal dependence, the simulated coverage rates are generally much farther away from the intended 90\% than in the iid case. Especially for $n=50$, the simulated coverage is usually much below. Here, it is interesting to note that for some continuous AR$(1)$ DGPs, the shortfall is strongest for a large absolute extent of cross-dependence, whereas for discrete AR$(1)$ DGPs, the independence case is usually the most problematic. In any case, we observe a clear improvement of the coverage for increasing~$n$, and it should be kept in mind that we used a very strong autocorrelation level for our AR$(1)$ simulations. We come back to the problem of variance estimation under temporal dependence in the conclusion, where we discuss research directions for tackling this challenge.

\subsection{Tabulated Simulation Results}
\label{Tabulated Simulation results from Section sec:simulations}

\begin{table}[h]                      \centering
	\caption{Size (block ``$\tau=0$'') and power (otherwise) for $r, \tau$ and $\rho$ for continuous iid DGPs with $MC=1,000$ and $10\%$ significance level, see Section~\ref{subsec:sims_iid}.} 
    \label{tabTestsIIDcont}

    \smallskip
\resizebox{\linewidth}{!}{
\begin{tabular}{cc@{\ \ }c@{\ \ }cc@{\ \ }c@{\ \ }cc@{\ \ }c@{\ \ }cc@{\ \ }c@{\ \ }cc@{\ \ }c@{\ \ }c}
			\addlinespace
			\toprule
			DGP (\ref{eq:normaliidDGP})&\multicolumn{3}{c}{$\tau=-0.1$}&\multicolumn{3}{c}{$\tau=-0.05$}&\multicolumn{3}{c}{$\tau=0$}&\multicolumn{3}{c}{$\tau=0.05$}&\multicolumn{3}{c}{$\tau=0.1$}\\
			\midrule 
            & $r$ &  $\tau$& $\rho$ &$r$ &  $\tau$& $\rho$ &$r$ &  $\tau$& $\rho$ &$r$ &  $\tau$& $\rho$ &$r$ &  $\tau$& $\rho$  \\
			$n=50$& 0.296&0.283 &0.267&0.153&0.162 &0.151&0.103&0.127 &0.116&0.172&0.190&0.178&0.318& 0.311&0.294 \\
			$n=200$& 0.716&0.686 &0.685&0.309&0.298 &0.298&0.104&0.108 &0.106&0.319&0.308 &0.308&0.736&0.708&0.699 \\
			$n=800$&0.995&0.990 &0.991&0.722&0.681 &0.678&0.097&0.101 &0.103&0.748&0.710&0.711&0.999& 0.997&0.997 \\
			\midrule
            DGP (\ref{eq:t4iidDGP})&\multicolumn{3}{c}{$\tau=-0.1$}&\multicolumn{3}{c}{$\tau=-0.05$}&\multicolumn{3}{c}{$\tau=0$}&\multicolumn{3}{c}{$\tau=0.05$}&\multicolumn{3}{c}{$\tau=0.1$}\\
            \midrule
            & $r$ &  $\tau$& $\rho$ &$r$ &  $\tau$& $\rho$ &$r$ &  $\tau$& $\rho$ &$r$ &  $\tau$& $\rho$ &$r$ &  $\tau$& $\rho$  \\
            $n=50$& 0.289&0.291 &0.276&0.168&0.159 &0.158&0.119&0.111 &0.107&0.144&0.147&0.130&0.285& 0.292&0.285\\
			$n=200$& 0.637&0.673 &0.671&0.260&0.283 &0.284&0.116&0.114 &0.110&0.275&0.281 &0.278&0.629&0.659&0.659 \\
			$n=800$& 0.987&0.997 &0.997&0.642&0.686 &0.685&0.091&0.089 &0.089&0.641&0.671 &0.665&0.989&0.992& 0.993\\
            \midrule
            DGP (\ref{eq:t1iidDGP})&\multicolumn{3}{c}{$\tau=-0.1$}&\multicolumn{3}{c}{$\tau=-0.05$}&\multicolumn{3}{c}{$\tau=0$}&\multicolumn{3}{c}{$\tau=0.05$}&\multicolumn{3}{c}{$\tau=0.1$}\\
            \midrule
            & $r$ &  $\tau$& $\rho$ &$r$ &  $\tau$& $\rho$ &$r$ &  $\tau$& $\rho$ &$r$ &  $\tau$& $\rho$ &$r$ &  $\tau$& $\rho$  \\
            $n=50$&0.241&0.310 &0.291&0.135&0.181 &0.166&0.067&0.134 &0.119&0.144&0.198 &0.181&0.246&0.340&0.305  \\
			$n=200$&0.367&0.720 &0.707&0.209&0.346 &0.338&0.050&0.102 &0.107&0.208&0.329 &0.318&0.377&0.713&0.697  \\
			$n=800$&0.566&0.995 &0.995&0.307&0.773 &0.770&0.026&0.095 &0.095&0.321&0.760 &0.755&0.545&0.996&0.995  \\
            \midrule
            DGP (\ref{eq:normalexpiidDGP})&\multicolumn{3}{c}{$\tau=-0.1$}&\multicolumn{3}{c}{$\tau=-0.05$}&\multicolumn{3}{c}{$\tau=0$}&\multicolumn{3}{c}{$\tau=0.05$}&\multicolumn{3}{c}{$\tau=0.1$}\\
            \midrule
            & $r$ &  $\tau$& $\rho$ &$r$ &  $\tau$& $\rho$ &$r$ &  $\tau$& $\rho$ &$r$ &  $\tau$& $\rho$ &$r$ &  $\tau$& $\rho$  \\
            $n=50$& 0.265&0.290 &0.271&0.146&0.157 &0.147&0.116&0.125 &0.114&0.160&0.169 &0.154&0.281&0.307 & 0.294 \\
			$n=200$& 0.636&0.693 &0.697&0.257&0.251 &0.249&0.088&0.100 &0.097&0.277&0.314 &0.305&0.691&0.729 &0.731\\
			$n=800$& 0.987&0.992 &0.992&0.665&0.695 &0.697&0.091&0.104 &0.104&0.661&0.703 &0.702&0.994&0.996 &0.996 \\
            \bottomrule
            \end{tabular}}
\end{table}

\begin{table}[h]
\footnotesize\centering
	\caption{Size for $r,\tau, \tau_b, \gamma$  and $\rho_b$ for discrete iid DGPs with $MC=1,000$ and $10\%$ significance level, see Section~\ref{subsec:sims_iid}.}
    \label{tabTestsIIDdiscrH0}

    \smallskip
        	\begin{tabular}{cccccc}
			\addlinespace
			\toprule
			DGP (\ref{eq:Pois1iidDGP})&\multicolumn{5}{c}{$\gamma=0$}\\			
            \midrule 
            &$r$ &  $\tau$& $\tau_b$&$\gamma$& $\rho_b$ \\
			$n=50$&0.113  &0.122 &0.111&0.103&0.113 \\
			$n=200$& 0.106 &0.100 &0.096&0.094&0.097 \\
			$n=800$& 0.102 &0.105 &0.105&0.105&0.107 \\
			\midrule
            DGP (\ref{eq:Zipf1iidDGP})&\multicolumn{5}{c}{$\gamma=0$}\\			
            \midrule 
            &$r$ &  $\tau$& $\tau_b$&$\gamma$& $\rho_b$ \\
			$n=50$& 0.064 &0.119&0.115&0.104 &0.115  \\
			$n=200$& 0.050 &0.109&0.106 &0.102 &0.105 \\
			$n=800$& 0.025 &0.104&0.103&0.107&0.104 \\
			\midrule
            DGP (\ref{eq:Ske11amiidDGP})&\multicolumn{5}{c}{$\gamma=0$}\\			
            \midrule 
            &$r$ &  $\tau$& $\tau_b$&$\gamma$& $\rho_b$ \\
			$n=50$& 0.115 &0.122&0.114&0.104&0.117  \\
			$n=200$& 0.110 &0.103&0.100&0.099&0.099  \\
			$n=800$& 0.104 &0.104&0.104&0.104&0.106  \\
            \bottomrule
		\end{tabular}
\end{table}

\begin{table}[h]                       \centering
	\caption{Power for $r,\tau, \tau_b, \gamma$  and $\rho_b$ for discrete iid DGPs with $MC=1,000$ and $10\%$ significance level, see Section~\ref{subsec:sims_iid}.}
    \label{tabTestsIIDdiscrH1}

    \smallskip
        	\begin{tabular}{c@{\qquad}ccccc@{\qquad}ccccc}
			\addlinespace
			\toprule
			DGP (\ref{eq:Pois1iidDGP})&\multicolumn{5}{c}{$\gamma=0.05$}&\multicolumn{5}{c}{$\gamma=0.1$}\\			
            \midrule 
            &$r$ &  $\tau$& $\tau_b$&$\gamma$& $\rho_b$&$r$ &  $\tau$& $\tau_b$&$\gamma$& $\rho_b$ \\
			$n=50$& 0.123 &0.111 & 0.105  & 0.092 & 0.106&0.175&0.170 & 0.161& 0.145& 0.162\\
			$n=200$& 0.148&0.141 & 0.138 & 0.134& 0.139&0.336 &0.311 & 0.311 & 0.309 & 0.310 \\
			$n=800$& 0.295 &0.256 & 0.255  & 0.255& 0.255&0.806&0.755 & 0.755 & 0.754 & 0.754 \\
			\midrule
            DGP (\ref{eq:Zipf1iidDGP})&\multicolumn{5}{c}{$\gamma=0.05$}&\multicolumn{5}{c}{$\gamma=0.1$}\\
            \midrule 
            &$r$ &  $\tau$& $\tau_b$&$\gamma$& $\rho_b$&$r$ &  $\tau$& $\tau_b$&$\gamma$& $\rho_b$ \\
			$n=50$&  0.094 & 0.124 & 0.115  & 0.093& 0.114 & 0.128& 0.160 & 0.151& 0.117& 0.151\\
			$n=200$&  0.105 & 0.153& 0.151 & 0.134& 0.150& 0.151& 0.288& 0.287& 0.261& 0.285\\
			$n=800$&  0.123 & 0.289& 0.288  & 0.281 & 0.289& 0.237& 0.690  & 0.689& 0.674& 0.683\\
			\midrule
            DGP (\ref{eq:Ske11amiidDGP})&\multicolumn{5}{c}{$\gamma=0.05$}&\multicolumn{5}{c}{$\gamma=0.1$}\\			
            \midrule 
            &$r$ &  $\tau$& $\tau_b$&$\gamma$& $\rho_b$&$r$ &  $\tau$& $\tau_b$&$\gamma$& $\rho_b$ \\
			$n=50$&  0.124 & 0.125& 0.121&  0.113 &  0.118& 0.209&  0.209 &  0.201&  0.185&  0.194 \\
			$n=200$& 0.202 &  0.190& 0.189 &  0.187&  0.185& 0.489&  0.445& 0.440&  0.438&  0.435\\
			$n=800$& 0.460 &  0.450& 0.450&  0.450& 0.450& 0.940&  0.909 &  0.908 &  0.908 &  0.908 \\
            \bottomrule
		\end{tabular}
\end{table}

\begin{figure}
    \centering
    \includegraphics[width=0.6\linewidth]{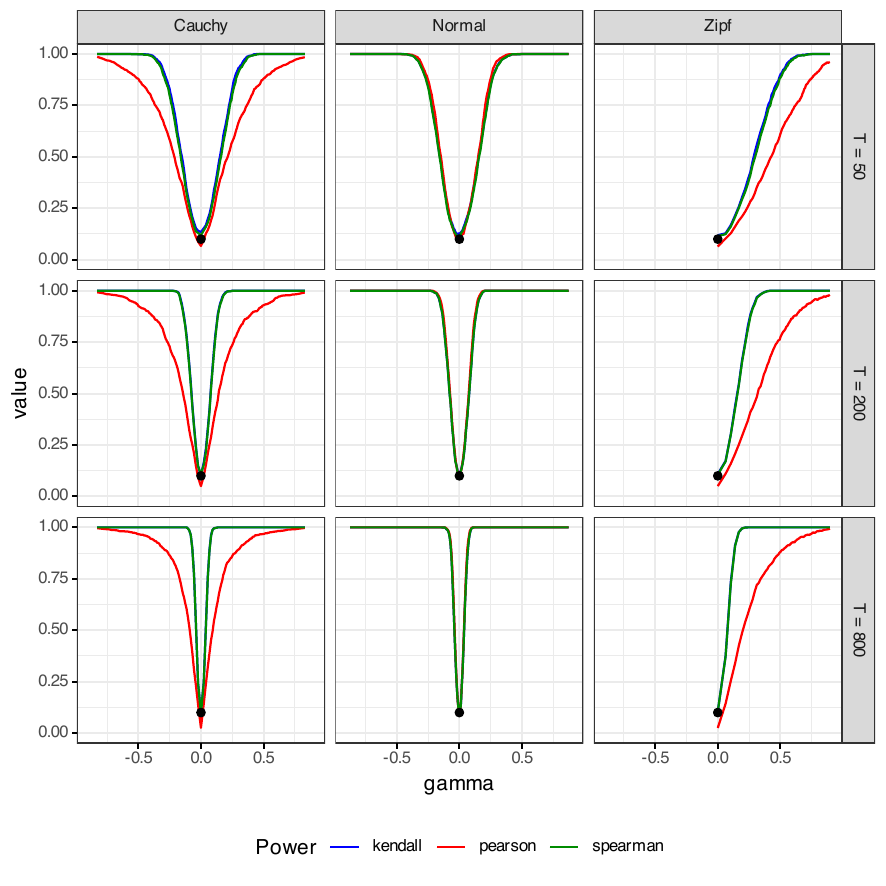}
    \caption{Rejection rates for independence tests based on Kendall's $\tau$, Pearson's $r$ and Spearman's $\rho$ for DGPs (\ref{eq:normaliidDGP}), (\ref{eq:t1iidDGP}) and (\ref{eq:Zipf1iidDGP}) with $MC=1,000$ simulation runs.}
    \label{fig:iid_rejrates}
\end{figure}

\begin{table}[tb]
		\footnotesize\centering
    \caption{Coverage of 90\% confidence intervals for $\tau$ and $\rho$ for continuous iid DGPs with $MC=1,000$, see Section~\ref{subsec:sims_iid}.} 
    \label{tabConfIIDcont}

    \smallskip
        	\begin{tabular}{cccccc@{\qquad}ccccc}
            & \multicolumn{5}{c}{\textbf{without Fisher Transformation}}&\multicolumn{5}{c}{\textbf{with Fisher Transformation}}\\
            \addlinespace
            \toprule
			DGP (\ref{eq:normaliidDGP}) / $\tau=$&$-0.8$&$-0.4$&$0$&$0.4$&$0.8$&$-0.8$&$-0.4$&$0$&$0.4$&$0.8$\\
			\midrule
			$n=50$&0.912 &0.900 &0.870 &0.868 &0.903 &0.925 &0.902 &0.878 &0.880 &0.912 \\
			$n=200$& 0.904 &0.895 &0.896 &0.887& 0.881&0.907 &0.893 &0.897 &0.885 &0.883 \\
			$n=800$& 0.897 &0.911 &0.900 &0.889& 0.886& 0.901 &0.913 &0.900 &0.887 &0.890 \\
            \midrule
            DGP (\ref{eq:normaliidDGP}) / $\rho=$&$-0.8$&$-0.4$&$0$&$0.4$&$0.8$&$-0.8$&$-0.4$&$0$&$0.4$&$0.8$\\ \midrule
			$n=50$& 0.894 &0.884 &0.869 &0.867 &0.892& 0.892 &0.894 &0.887 &0.880 &0.901 \\     
			$n=200$& 0.895& 0.890& 0.895 &0.885 &0.888& 0.894& 0.893 &0.897 &0.883 &0.882  \\  
			$n=800$& 0.908 &0.915 &0.899 &0.893& 0.886& 0.905& 0.915 &0.899 &0.891 &0.886  \\  \midrule
            DGP (\ref{eq:t4iidDGP}) / $\tau=$&$-0.8$&$-0.4$&$0$&$0.4$&$0.8$&$-0.8$&$-0.4$&$0$&$0.4$&$0.8$\\ \midrule
            $n=50$& 0.891 &0.878 &0.883 &0.881 &0.892& 0.902 &0.884 &0.889 &0.890 &0.908 \\     
			$n=200$& 0.891 &0.899 &0.885 &0.882 &0.887& 0.890 &0.898 &0.887 &0.888 &0.885 \\  
			$n=800$& 0.888 &0.902 &0.910 &0.896 &0.901& 0.897 &0.905 &0.910 &0.897 &0.899 \\  \midrule
            DGP (\ref{eq:t4iidDGP}) / $\rho=$&$-0.8$&$-0.4$&$0$&$0.4$&$0.8$&$-0.8$&$-0.4$&$0$&$0.4$&$0.8$\\ \midrule
            $n=50$& 0.888 & 0.880 & 0.884 & 0.885 & 0.880& 0.891 & 0.893 & 0.896 & 0.891 & 0.899 \\     
			$n=200$& 0.892 & 0.893 & 0.886 & 0.890 & 0.877&  0.887 & 0.908 & 0.887 & 0.887 & 0.881 \\  
			$n=800$& 0.892 & 0.907 & 0.908 & 0.897 & 0.899&  0.897 & 0.909 & 0.909 & 0.897 & 0.898 \\  \midrule
            DGP (\ref{eq:t1iidDGP}) / $\tau=$&$-0.8$&$-0.4$&$0$&$0.4$&$0.8$&$-0.8$&$-0.4$&$0$&$0.4$&$0.8$\\ \midrule
            $n=50$& 0.862 & 0.877 & 0.860 & 0.870 & 0.837& 0.875 & 0.892 & 0.866 & 0.867 & 0.861 \\     
			$n=200$& 0.873 & 0.891 & 0.896 & 0.891 & 0.866& 0.884 & 0.894 & 0.897 & 0.893 & 0.875 \\  
			$n=800$& 0.871 & 0.908 & 0.909 & 0.904 & 0.895& 0.883 & 0.913 & 0.909 & 0.902 & 0.895 \\  \midrule
            DGP (\ref{eq:t1iidDGP}) / $\rho=$&$-0.8$&$-0.4$&$0$&$0.4$&$0.8$&$-0.8$&$-0.4$&$0$&$0.4$&$0.8$\\ \midrule
            $n=50$& 0.861 &0.884 &0.867&0.873&0.835& 0.871 &0.902 &0.879 &0.881 &0.853   \\     
			$n=200$& 0.870&0.885&0.893&0.889&0.871& 0.883 &0.896 &0.896 &0.895 &0.885  \\  
			$n=800$& 0.880&0.909&0.907&0.901&0.894& 0.886 &0.907 &0.908 &0.901 &0.897  \\  \midrule
            DGP (\ref{eq:normalexpiidDGP}) / $\tau=$&$-0.8$&$-0.4$&$0$&$0.4$&$0.8$&$-0.8$&$-0.4$&$0$&$0.4$&$0.8$\\ \midrule
            $n=50$& 0.917 &0.890 &0.880 &0.883 &0.926& 0.921 &0.891 &0.883 &0.887 &0.929 \\     
			$n=200$& 0.925 &0.897 &0.895 &0.893 &0.898& 0.931 &0.901 &0.896 &0.889 &0.903 \\  
			$n=800$& 0.896 &0.891 &0.894 &0.887 &0.908& 0.897 &0.892 &0.894 &0.889 &0.914 \\  \midrule
            DGP (\ref{eq:normalexpiidDGP}) / $\rho=$&$-0.8$&$-0.4$&$0$&$0.4$&$0.8$&$-0.8$&$-0.4$&$0$&$0.4$&$0.8$\\ \midrule
 			$n=50$& 0.910 &0.891& 0.882& 0.881& 0.900&0.914& 0.901 &0.889& 0.884 &0.911  \\     
			$n=200$& 0.905 &0.907& 0.904& 0.895& 0.889& 0.905 &0.907 &0.904 &0.896 &0.891  \\  
			$n=800$& 0.887 &0.896& 0.893 &0.886& 0.892& 0.889 &0.899 &0.894 &0.886 &0.891  \\  
            \bottomrule
		\end{tabular}
\end{table}

\begin{table}[tb]
		\footnotesize\centering
    \caption{Coverage of 90\% confidence intervals for $\gamma$ and $\rho_b$ for discrete iid DGPs with $MC=1,000$, see Section~\ref{subsec:sims_iid}.}
    \label{tabConfIIDdiscr}

    \smallskip
        	\begin{tabular}{cccc@{\qquad}ccc}
            & \multicolumn{3}{c}{\textbf{without F.\ T.}}&\multicolumn{3}{c}{\textbf{with F.\ T.}}\\
            \addlinespace
            \toprule
            DGP (\ref{eq:Pois1iidDGP}) / $\gamma=$&$0$&$0.4$&$0.8$&$0$&$0.4$&$0.8$\\ \midrule
            $n=50$& 0.865 &0.862 &0.832& 0.890 &0.885 &0.868\\     
			$n=200$& 0.900& 0.888 &0.876& 0.905 &0.896 &0.886\\  
			$n=800$&0.891 &0.893 &0.891& 0.895 &0.894 &0.887\\ \midrule
            DGP (\ref{eq:Pois1iidDGP}) / $\rho_b=$&$0$&$0.4$&$0.8$&$0$&$0.4$&$0.8$\\ \midrule
            $n=50$& 0.875 &0.871 &0.855& 0.888& 0.885& 0.877\\     
			$n=200$& 0.905 &0.882& 0.864& 0.905& 0.882& 0.871\\  
			$n=800$&0.895 &0.891 &0.895& 0.895& 0.887& 0.899\\ \midrule
            DGP (\ref{eq:Zipf1iidDGP}) / $\gamma=$&$0$&$0.4$&$0.8$&$0$&$0.4$&$0.8$\\ \midrule
			$n=50$& 0.868 &0.861 &0.855  & 0.896 &0.885& 0.900\\     
			$n=200$&0.885& 0.899& 0.890 &  0.893& 0.903& 0.895 \\  
			$n=800$&0.895& 0.899& 0.917 & 0.896& 0.900& 0.918 \\ \midrule
            DGP (\ref{eq:Zipf1iidDGP}) / $\rho_b=$&$0$&$0.4$&$0.8$&$0$&$0.4$&$0.8$\\ \midrule
            $n=50$& 0.863 &0.837& 0.843 & 0.879 &0.861& 0.834\\
			$n=200$& 0.891& 0.896& 0.836 & 0.891 &0.906& 0.852\\  
			$n=800$& 0.894& 0.890& 0.868 & 0.894& 0.886& 0.864\\ \midrule
            DGP (\ref{eq:Ske11amiidDGP}) / $\gamma=$&$0$&$0.4$&$0.8$&$0$&$0.4$&$0.8$\\ \midrule
            $n=50$&   0.870 & 0.890 & 0.831& 0.884 & 0.894&0.858\\
			$n=200$& 0.900 & 0.894 & 0.869 & 0.902 & 0.903 & 0.875 \\  
			$n=800$& 0.895 & 0.892 & 0.899 & 0.895 & 0.892 & 0.900  \\ \midrule
            DGP (\ref{eq:Ske11amiidDGP}) / $\rho_b=$&$0$&$0.4$&$0.8$&$0$&$0.4$&$0.8$\\ \midrule
            $n=50$& 0.871 & 0.889  &0.838 & 0.880  &0.898&  0.842 \\     
			$n=200$& 0.898 & 0.908 & 0.874 & 0.900&  0.916&  0.883 \\  
			$n=800$& 0.897  &0.885&  0.902 &0.897&  0.886  &0.903 \\  
            \bottomrule
		\end{tabular}
\end{table}

\clearpage

\begin{table}[tb]                      \centering
	\caption{Size and power for $r, \tau$ and $\rho$ for continuous time series DGPs with $MC=1,000$ and $10\%$ significance level, see Section~\ref{subsec:sims_time_series}.} 
    \label{tabTestsAR1cont}

    \smallskip
\resizebox{\linewidth}{!}{
\begin{tabular}{cc@{\ \ }c@{\ \ }cc@{\ \ }c@{\ \ }cc@{\ \ }c@{\ \ }cc@{\ \ }c@{\ \ }cc@{\ \ }c@{\ \ }c}
			\addlinespace
			\toprule
			DGP (\ref{eq:normalTSDGP})&\multicolumn{3}{c}{$\tau=-0.1$}&\multicolumn{3}{c}{$\tau=-0.05$}&\multicolumn{3}{c}{$\tau=0$}&\multicolumn{3}{c}{$\tau=0.05$}&\multicolumn{3}{c}{$\tau=0.1$}\\
			\midrule 
            & $r$ &  $\tau$& $\rho$ &$r$ &  $\tau$& $\rho$ &$r$ &  $\tau$& $\rho$ &$r$ &  $\tau$& $\rho$ &$r$ &  $\tau$& $\rho$  \\
			$n=50$& 0.258 &0.266&0.251&0.208&0.215&0.211& 0.178 &0.199&0.184&0.192 &0.211&0.199&0.265&0.272&0.259  \\
			$n=200$&0.350 &0.342&0.336&0.184 &0.182&0.182&0.148 &0.144&0.144&0.208 &0.211&0.203&0.366&0.360&0.356  \\
			$n=800$&0.714 &0.687&0.686&0.354 &0.342&0.339&0.133 &0.137&0.131&0.342 &0.336&0.335&0.714 &0.697&0.694\\
			\midrule
            DGP (\ref{eq:t4TSDGP})&\multicolumn{3}{c}{$\tau=-0.1$}&\multicolumn{3}{c}{$\tau=-0.05$}&\multicolumn{3}{c}{$\tau=0$}&\multicolumn{3}{c}{$\tau=0.05$}&\multicolumn{3}{c}{$\tau=0.1$}\\
            \midrule
            & $r$ &  $\tau$& $\rho$ &$r$ &  $\tau$& $\rho$ &$r$ &  $\tau$& $\rho$ &$r$ &  $\tau$& $\rho$ &$r$ &  $\tau$& $\rho$  \\
			$n=50$& 0.270 &0.287&0.275&0.220 &0.224&0.219&0.194 &0.197&0.194&0.224 &0.223&0.203&0.274 & 0.265&0.257 \\
			$n=200$& 0.347 &0.337&0.335&0.207 &0.204&0.205&0.152 &0.144&0.149&0.205 &0.215&0.214&0.345&0.353&0.345 \\
			$n=800$& 0.674 &0.660&0.661&0.331 &0.322&0.322&0.129 &0.124&0.122&0.346 &0.343&0.344&0.689&0.695&0.691\\
            \midrule
            DGP (\ref{eq:t1TSDGP})&\multicolumn{3}{c}{$\tau=-0.1$}&\multicolumn{3}{c}{$\tau=-0.05$}&\multicolumn{3}{c}{$\tau=0$}&\multicolumn{3}{c}{$\tau=0.05$}&\multicolumn{3}{c}{$\tau=0.1$}\\
            \midrule
            & $r$ &  $\tau$& $\rho$ &$r$ &  $\tau$& $\rho$ &$r$ &  $\tau$& $\rho$ &$r$ &  $\tau$& $\rho$ &$r$ &  $\tau$& $\rho$  \\
			$n=50$&0.250 &0.304&0.271&0.216 &0.275&0.247&0.184 &0.256&0.227&0.226 &0.288&0.262&0.270 &0.308&0.289\\
			$n=200$&0.279 &0.337&0.316&0.202 &0.239&0.221&0.110 &0.186&0.174&0.180 &0.211&0.199&0.267&0.296 & 0.281\\
			$n=800$&0.351 &0.523&0.507&0.232 &0.304&0.295&0.075 &0.146&0.142&0.232 &0.295&0.292&0.355&0.512&0.502\\
            \midrule
            DGP (\ref{eq:TEAR1DGP})&&&&&&&\multicolumn{3}{c}{$\tau=0$}&\multicolumn{3}{c}{$\tau=0.05$}&\multicolumn{3}{c}{$\tau=0.1$}\\
            \midrule
            &&&&&&& $r$ &  $\tau$& $\rho$ &$r$ &  $\tau$& $\rho$ &$r$ &  $\tau$& $\rho$  \\
			$n=50$&&&&&&& 0.152&0.181&0.152& 0.176 &0.209&0.163&0.264  &0.304 &0.233  \\
			$n=200$&&&&&&& 0.125 &0.136&0.132&0.214 &0.250&0.215&0.386 &0.513 &0.443  \\
			$n=800$&&&&&&& 0.122 &0.121&0.119 &0.330&0.432&0.400& 0.702 &0.898 &0.856 \\
            \bottomrule
            \end{tabular}}
\end{table}

\begin{table}[tb]                       \footnotesize\centering
	\caption{Size for $r,\tau, \tau_b, \gamma$  and $\rho_b$ for discrete time series DGPs with $MC=1,000$ and $10\%$ significance level, see Section~\ref{subsec:sims_time_series}.}
    \label{tabTestsAR1discrH0}

    \smallskip
        	\begin{tabular}{ccccccc}
			\addlinespace
			\toprule
			DGP (\ref{eq:Pois1TSDGP})&\multicolumn{5}{c}{$\gamma=0$}\\			
            \midrule 
            &$r$ &  $\tau$& $\tau_b$&$\gamma$& $\rho_b$ \\
			$n=50$& 0.190 &0.229&0.211 &0.186 &0.217\\
			$n=200$& 0.151 &0.168 &0.164&0.161 &0.165\\
			$n=800$& 0.110 &0.099&0.099 &0.099 &0.100 \\
			\midrule
            DGP (\ref{eq:ZipfTSDGP})&\multicolumn{5}{c}{$\gamma=0$}\\			
            \midrule 
            &$r$ &  $\tau$& $\tau_b$&$\gamma$& $\rho_b$ \\
			$n=50$&  0.193  &0.233  &0.223  &0.211 &0.222\\
			$n=200$& 0.120 &0.154&0.150 &0.145 &0.151 \\
			$n=800$&0.096 &0.137&0.137 &0.137 &0.139\\
			\midrule
            DGP (\ref{eq:SkellamTSDGP})&\multicolumn{5}{c}{$\gamma=0$}\\			
            \midrule 
            &$r$ &  $\tau$& $\tau_b$&$\gamma$& $\rho_b$ \\
			$n=50$&  0.181 &0.201  &0.192  &0.177&0.192 \\
			$n=200$&  0.150 &0.145&0.145 &0.140 &0.143  \\
			$n=800$& 0.116  &0.132  &0.131 &0.131 &0.135\\
            \bottomrule
		\end{tabular}
\end{table}

\begin{table}[tb]
    \footnotesize\centering
	\caption{Power for $r,\tau, \tau_b, \gamma$  and $\rho_b$ for discrete time series DGPs with $MC=1,000$ and $10\%$ significance level, see Section~\ref{subsec:sims_time_series}.}
    \label{tabTestsAR1discrH1}

    \smallskip
        	\begin{tabular}{c@{\qquad}ccccc@{\qquad}ccccc}
			\addlinespace
			\toprule
			DGP (\ref{eq:Pois1TSDGP})&\multicolumn{5}{c}{$\gamma=0.05$}&\multicolumn{5}{c}{$\gamma=0.1$}\\			
            \midrule
            &$r$ &  $\tau$& $\tau_b$&$\gamma$& $\rho_b$&$r$ &  $\tau$& $\tau_b$&$\gamma$& $\rho_b$ \\
			$n=50$& 0.187 & 0.213 & 0.204 & 0.181& 0.202&0.215& 0.225&0.219& 0.202&0.221  \\
			$n=200$& 0.175 & 0.160 & 0.158& 0.152&0.158&0.216&0.211& 0.210& 0.204&0.210  \\
			$n=800$& 0.164 & 0.160 & 0.159& 0.157&0.160&0.374& 0.358& 0.357 & 0.356&0.358  \\
			\midrule
            DGP (\ref{eq:ZipfTSDGP})&\multicolumn{5}{c}{$\gamma=0.05$}&\multicolumn{5}{c}{$\gamma=0.1$}\\			
            \midrule
            &$r$ &  $\tau$& $\tau_b$&$\gamma$& $\rho_b$&$r$ &  $\tau$& $\tau_b$&$\gamma$& $\rho_b$ \\
			$n=50$&  0.223 &0.260& 0.248 &0.233 &0.244 & 0.278&0.285&0.273&0.256& 0.274\\
			$n=200$&  0.185 & 0.194 & 0.193 &0.190 & 0.193& 0.305&0.313&0.311&0.307 & 0.303 \\
			$n=800$&  0.237 & 0.268 & 0.267 &0.267 & 0.265& 0.472& 0.579& 0.579&0.578&0.575 \\
			\midrule
            DGP (\ref{eq:SkellamTSDGP})&\multicolumn{5}{c}{$\gamma=0.05$}&\multicolumn{5}{c}{$\gamma=0.1$}\\			
            \midrule
            &$r$ &  $\tau$& $\tau_b$&$\gamma$& $\rho_b$&$r$ &  $\tau$& $\tau_b$&$\gamma$& $\rho_b$ \\
			$n=50$&0.186&0.204&0.191&0.171&0.192&0.196&0.221&0.204&0.185&0.200  \\
			$n=200$&0.168& 0.171&0.171&0.164&0.166&0.236 &0.244 &0.241&0.233&0.238  \\
			$n=800$&0.203&0.200&0.198&0.197&0.201&0.418&0.406&0.406&0.404&0.406 \\
            \bottomrule
		\end{tabular}
\end{table}

\begin{figure}
    \centering
    \includegraphics[width=0.6\linewidth]{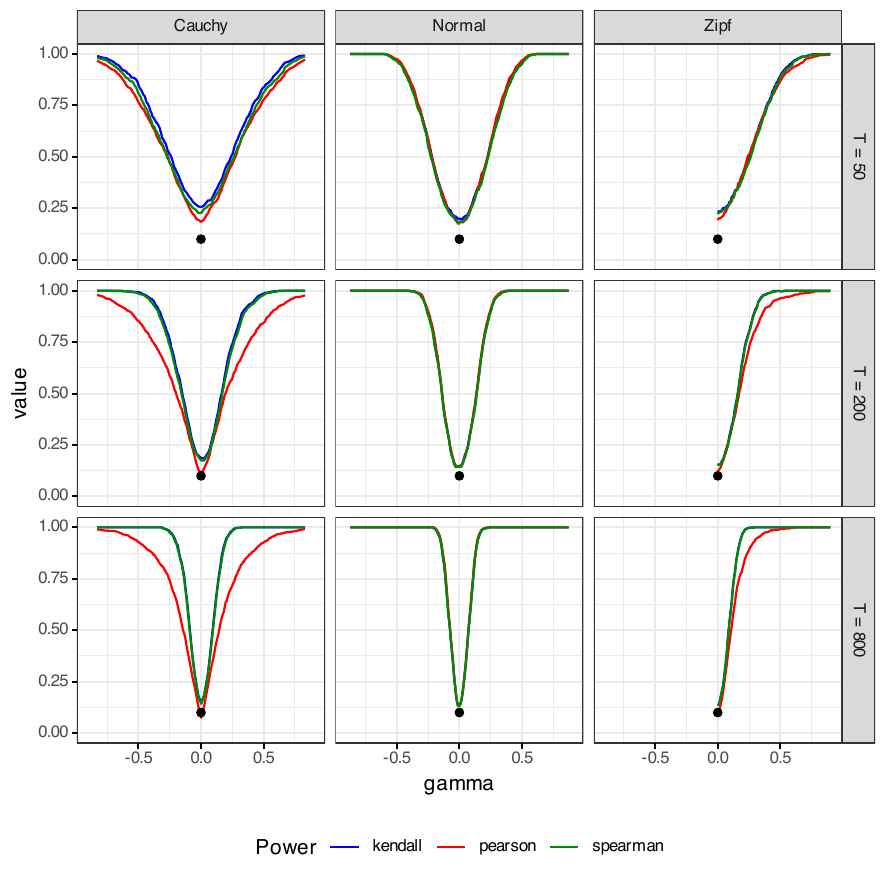}
    \caption{Rejection rates for independence tests based on Kendall's $\tau$, Pearson's $r$ and Spearman's $\rho$ for DGPs (\ref{eq:normalTSDGP}), (\ref{eq:t1TSDGP}) and (\ref{eq:ZipfTSDGP}) with $MC=1,000$ simulation runs.}
    \label{fig:TS_rejrates}
\end{figure}

\begin{table}[tb]
		\footnotesize\centering
    \caption{Coverage of 90\% confidence intervals for $\tau$ and $\rho$ for continuous time series DGPs with $MC=1,000$, see Section~\ref{subsec:sims_time_series}.}
    \label{tabConfAR1cont}

    \smallskip
        	\begin{tabular}{cccccc@{\qquad}ccccc}
            & \multicolumn{5}{c}{\textbf{without Fisher Transformation}}&\multicolumn{5}{c}{\textbf{with Fisher Transformation}}\\
            \addlinespace
            \toprule
            DGP (\ref{eq:normalTSDGP}) / $\tau=$&$-0.8$&$-0.4$&$0$&$0.4$&$0.8$&$-0.8$&$-0.4$&$0$&$0.4$&$0.8$\\ \midrule
            $n=50$& 0.766 &0.730 &0.729& 0.733& 0.770& 0.789 &0.759 &0.750 &0.739& 0.797  \\     
			$n=200$& 0.838& 0.839& 0.832& 0.837& 0.825& 0.834 &0.841& 0.837& 0.848& 0.831  \\  
			$n=800$& 0.847 &0.856& 0.857& 0.850& 0.833&0.851 &0.859& 0.859& 0.849& 0.837 \\ \midrule
            DGP (\ref{eq:normalTSDGP}) / $\rho=$&$-0.8$&$-0.4$&$0$&$0.4$&$0.8$&$-0.8$&$-0.4$&$0$&$0.4$&$0.8$\\ \midrule
            $n=50$& 0.746 & 0.735 & 0.724 & 0.714 & 0.746& 0.753 & 0.765 & 0.761 & 0.744 & 0.764   \\     
			$n=200$&0.844 & 0.834 & 0.832 & 0.831 & 0.820& 0.837 & 0.845 & 0.838 & 0.843 & 0.831   \\  
			$n=800$& 0.852 & 0.861 & 0.858 & 0.845 & 0.846&  0.856 & 0.863 & 0.862 & 0.849 & 0.836 \\ \midrule
            DGP (\ref{eq:t4TSDGP}) / $\tau=$&$-0.8$&$-0.4$&$0$&$0.4$&$0.8$&$-0.8$&$-0.4$&$0$&$0.4$&$0.8$\\ \midrule
            $n=50$&  0.733 &0.696 &0.724& 0.739 &0.759& 0.770 &0.723 &0.737& 0.759 &0.767  \\     
			$n=200$&  0.806& 0.806& 0.816& 0.830& 0.820& 0.813& 0.806& 0.820& 0.834& 0.829 \\  
			$n=800$& 0.859 &0.879& 0.857 &0.852 &0.867&  0.862 &0.875& 0.858& 0.853& 0.870\\  \midrule
            DGP (\ref{eq:t4TSDGP}) / $\rho=$&$-0.8$&$-0.4$&$0$&$0.4$&$0.8$&$-0.8$&$-0.4$&$0$&$0.4$&$0.8$\\ \midrule
            $n=50$&  0.707& 0.703& 0.717 &0.724 &0.754& 0.727 &0.738& 0.747& 0.759 &0.753  \\     
			$n=200$& 0.810& 0.809& 0.809& 0.816& 0.838& 0.810& 0.826& 0.819& 0.830 &0.844  \\  
			$n=800$& 0.864 &0.865& 0.856& 0.849& 0.856&  0.864 &0.867 &0.863& 0.853 &0.866 \\  \midrule
            DGP (\ref{eq:t1TSDGP}) / $\tau=$&$-0.8$&$-0.4$&$0$&$0.4$&$0.8$&$-0.8$&$-0.4$&$0$&$0.4$&$0.8$\\ \midrule
            $n=50$&  0.564 & 0.671 & 0.700 & 0.658&  0.564 & 0.603 & 0.707 & 0.726 & 0.682 & 0.604  \\     
			$n=200$& 0.697&  0.772&  0.803 & 0.771&  0.689& 0.699 & 0.784 & 0.812 & 0.779 & 0.695  \\  
			$n=800$& 0.790 & 0.820 & 0.857 & 0.849 & 0.794&  0.797 & 0.823 & 0.857 & 0.853 & 0.809  \\  \midrule
            DGP (\ref{eq:t1TSDGP}) / $\rho=$&$-0.8$&$-0.4$&$0$&$0.4$&$0.8$&$-0.8$&$-0.4$&$0$&$0.4$&$0.8$\\ \midrule
            $n=50$&   0.593 & 0.662 & 0.688 & 0.646 & 0.585&  0.640 & 0.704 & 0.731 & 0.699 & 0.635   \\     
			$n=200$& 0.697 & 0.765&  0.797 & 0.769 & 0.702&  0.730 & 0.781 & 0.807 & 0.783 & 0.706  \\  
			$n=800$& 0.789 & 0.813 & 0.850 & 0.837 & 0.808&  0.809 & 0.819 & 0.857 & 0.838 & 0.821  \\   \midrule
            DGP (\ref{eq:TEAR1DGP}) / $\tau=$&$-0.8$&$-0.4$&$0$&$0.4$&$0.8$&$-0.8$&$-0.4$&$0$&$0.4$&$0.8$\\ \midrule
            $n=50$&&&0.779 &0.789 &0.804 &&& 0.793 &0.804& 0.826\\     
			$n=200$&&& 0.837 &0.871 &0.858 &&& 0.841 &0.876& 0.862 \\  
			$n=800$&&& 0.877 &0.883 &0.891 &&&  0.878& 0.889 &0.898\\    \midrule
            DGP (\ref{eq:TEAR1DGP}) / $\rho=$&$-0.8$&$-0.4$&$0$&$0.4$&$0.8$&$-0.8$&$-0.4$&$0$&$0.4$&$0.8$\\ \midrule
            $n=50$& && 0.790& 0.780& 0.790 && &0.812 &0.792 &0.802 \\     
			$n=200$& &&0.836& 0.878& 0.851 && & 0.843 &0.879 &0.862\\  
			$n=800$& &&0.876 &0.879 &0.872 && &0.879 &0.875 &0.886\\ 
            \bottomrule
		\end{tabular}
\end{table}

\begin{table}[tb]
		\footnotesize\centering
    \caption{Coverage of 90\% confidence intervals for $\gamma$ and $\rho_b$ for discrete time series DGPs with $MC=1,000$, see Section~\ref{subsec:sims_time_series}.}
    \label{tabConfAR1discr}

    \smallskip
        	\begin{tabular}{cccc@{\qquad}ccc}
            & \multicolumn{3}{c}{\textbf{without F.\ T.}}&\multicolumn{3}{c}{\textbf{with F.\ T.}}\\
            \addlinespace
            \toprule
            DGP (\ref{eq:Pois1TSDGP}) / $\gamma=$&$0$&$0.4$&$0.8$&$0$&$0.4$&$0.8$\\ \midrule
            $n=50$& 0.676 &0.720 &0.764 & 0.724 &0.750 &0.816\\     
			$n=200$& 0.798 &0.810 &0.816 &0.820 &0.820& 0.823\\  
			$n=800$& 0.875 &0.886 &0.892 & 0.880 &0.886& 0.887\\  \midrule
            DGP (\ref{eq:Pois1TSDGP}) / $\rho_b=$&$0$&$0.4$&$0.8$&$0$&$0.4$&$0.8$\\ \midrule
            $n=50$& 0.706 &0.747& 0.784 & 0.721& 0.770& 0.773\\     
			$n=200$& 0.808& 0.818 &0.842 & 0.821& 0.818& 0.854\\  
			$n=800$& 0.879& 0.886& 0.884 & 0.884& 0.885& 0.878\\  \midrule
            DGP (\ref{eq:ZipfTSDGP}) / $\gamma=$&$0$&$0.4$&$0.8$&$0$&$0.4$&$0.8$\\ \midrule
            $n=50$& 0.705 & 0.679 & 0.762 & 0.723 & 0.699 & 0.765 \\     
			$n=200$& 0.822 & 0.814 & 0.799& 0.829 & 0.815 & 0.793\\  
			$n=800$& 0.857 & 0.836 & 0.824 & 0.858 & 0.839 & 0.841 \\  \midrule
            DGP (\ref{eq:ZipfTSDGP}) / $\rho_b=$&$0$&$0.4$&$0.8$&$0$&$0.4$&$0.8$\\ \midrule
            $n=50$& 0.704& 0.702& 0.738 & 0.728& 0.720& 0.728\\     
			$n=200$& 0.819& 0.805& 0.805& 0.830& 0.819& 0.807\\  
			$n=800$& 0.856& 0.836& 0.829& 0.857& 0.848& 0.846\\  \midrule
            DGP (\ref{eq:SkellamTSDGP}) / $\gamma=$&$0$&$0.4$&$0.8$&$0$&$0.4$&$0.8$\\ \midrule 
            $n=50$&  0.709& 0.743& 0.724&  0.753& 0.773& 0.780\\
			$n=200$&  0.818& 0.812& 0.805&  0.827& 0.827& 0.817 \\  
			$n=800$&  0.854& 0.849& 0.841&  0.856& 0.850& 0.840 \\  \midrule
            DGP (\ref{eq:SkellamTSDGP}) / $\rho_b=$&$0$&$0.4$&$0.8$&$0$&$0.4$&$0.8$\\ \midrule 
            $n=50$&  0.727& 0.763& 0.795&  0.758& 0.779& 0.790 \\     
			$n=200$&  0.821& 0.822& 0.831&  0.832& 0.834& 0.837 \\  
			$n=800$&  0.856& 0.857& 0.862&  0.859& 0.861& 0.873 \\ 
            \bottomrule
		\end{tabular}
\end{table}

\end{document}